\documentclass[11pt]{article}
\usepackage{amssymb, amsthm}
\usepackage{algorithm}
\usepackage{algpseudocode}
\usepackage{amsmath, amssymb, amsthm}
\usepackage{float,multirow,subcaption, setspace}
\usepackage{graphicx}
\usepackage{comment}
\usepackage{enumitem}
\usepackage{tikz}
\usepackage{bm, bbm}
\usetikzlibrary{shapes.geometric, positioning}
\usetikzlibrary{quotes, angles}
\usepackage{rotating}
\usepackage{hyperref}
\usepackage{natbib}
\usepackage{multicol}
\usepackage{mathtools}
\usepackage{nicefrac}
\usepackage{mathrsfs}

\definecolor{SkyBlue}{RGB}{14, 118, 188}
\definecolor{BrightRed}{RGB}{223, 82, 78}
\definecolor{Green638}{RGB}{165,255,118} % from colours.cafe on instagram; pallete638
\definecolor{BurntOrange}{HTML}{BF5900}

\newcommand{\R}{\mathbb{R}} % boldfaced R for the reals
\newcommand{\E}{\mathbb{E}} % boldfaced E for expectations
\newcommand{\calE}{\mathcal{E}}
\newcommand{\calL}{\mathcal{L}}
\newcommand{\calM}{\mathcal{M}}
\newcommand{\calT}{\mathcal{T}}

\newcommand{\calI}{\mathcal{I}}

\newcommand{\bcalE}{\boldsymbol{\mathcal{E}}}

\newcommand{\ind}[1]{\mathbbm{1}\left( #1 \right)} % indicator function, with an argument
\newcommand{\normaldist}[2]{\mathcal{N}\left(#1,#2\right)} % normal distribution
\newcommand{\mvnormaldist}[3]{\mathcal{N}_{#1}\left(#2,#3\right)} % multivariate normal distribution
\newcommand{\betadist}[2]{\textrm{Beta}\left(#1,#2\right)} % Beta distribution

\newcommand{\by}{\bm{y}}
\newcommand{\bx}{\bm{x}}
\newcommand{\bz}{\bm{z}}

\newcommand{\br}{\bm{r}}

\newcommand{\bX}{\bm{X}}
\newcommand{\bZ}{\bm{Z}}

\newcommand{\btheta}{\boldsymbol{\theta}}
\newcommand{\bbeta}{\boldsymbol{\beta}}

\newcommand{\blambda}{\boldsymbol{\lambda}}
\newcommand{\bTheta}{\boldsymbol{\Theta}}
\newcommand{\bEta}{\boldsymbol{\eta}}
\theoremstyle{plain}
\newtheorem{theorem}{Theorem}
\newtheorem{corollary}{Corollary}
\newtheorem{lemma}{Lemma}

\theoremstyle{definition}
\newtheorem{definition}{Definition}

\newtheorem{ex}{Example}

\theoremstyle{plain}
\newtheorem{remark}[theorem]{Remark}

\usepackage{fullpage, parskip}
\usepackage{cleveref}
\usepackage{booktabs}
\usepackage{siunitx}

\creflabelformat{equation}{#2\textup{#1}#3}

\hypersetup{pdfborder = {0 0 0.5 [3 3]}, colorlinks = true, linkcolor = BrightRed, citecolor = SkyBlue}

\newcommand{\includesupp}{1}

\newcommand{\switchref}[2]{%
  \if\includesupp1%
    #1%
  \else%
    #2%
  \fi%
}

\title{Fitting sparse high-dimensional varying-coefficient models with Bayesian regression tree ensembles}
\author{Soham Ghosh\thanks{Dept.\ of Statistics, University of Wisconsin--Madison. \texttt{sghosh39@wisc.edu}} \and Saloni Bhogale\thanks{Dept. of Political Science, University of Wisconsin--Madison. \texttt{bhogale@wisc.edu}} \and Sameer K. Deshpande\thanks{Dept. of Statistics, University of Wisconsin--Madison. \texttt{sameer.deshpande@wisc.edu}}}
\begin{document}
\maketitle

\begin{abstract}
    By allowing the effects of $p$ covariates in a linear regression model to vary as functions of $R$ additional \emph{effect modifiers}, varying-coefficient models (VCMs) strike a compelling balance between interpretable-but-rigid parametric models popular in classical statistics and flexible-but-opaque methods popular in machine learning.
But in high-dimensional settings where $p$ and/or $R$ exceed the number of observations, existing approaches to fitting VCMs fail to identify which covariates have a non-zero effect and which effect modifiers drive these effects.
We propose \texttt{sparseVCBART}, a fully Bayesian model that approximates each coefficient function in a VCM with a regression tree ensemble and encourages sparsity with a global--local shrinkage prior on the regression tree leaf outputs and a hierarchical prior on the splitting probabilities of each tree.
We show that the \texttt{sparseVCBART} posterior contracts at a near-minimax optimal rate, automatically adapting to the unknown sparsity structure and smoothness of the true coefficient functions.
Compared to existing state-of-the-art methods, \texttt{sparseVCBART} achieved competitive predictive accuracy and substantially narrower and better-calibrated uncertainty intervals, especially for null covariate effects.
We use \texttt{sparseVCBART} to investigate how the effects of interpersonal conversations on prejudice could vary according to the political and demographic characteristics of the respondents.

\end{abstract}

\section{Introduction}
\label{sec:intro}
Varying coefficient models \citep[VCMs;][]{Hastie1993} assert a linear relationship between an outcome $Y$ and $p$ covariates $X_{1}, \ldots, X_{p}$ but allow the relationship to change with respect to $R$ additional variables (known as \emph{effect modifiers}) $Z_{1}, \ldots, Z_{R}$: $\E[Y \vert \bX, \bZ] = \beta_{0}(\bZ) + \sum_{j = 1}^{p}{\beta_{j}(\bZ)X_{j}}.$
VCMs occupy a middle ground between inflexible-but-interpretable parametric models and flexible-but-opaque nonparametric models \citep{ZhouHooker2022,Lengerich2023}.

Since their introduction, VCMs have been used extensively in statistics and econometrics.
Important examples include \citet{Hoover1998}, which estimated how the effects of maternal health on infant weight varied over time, \citet{Li2011}, which estimated how returns to education vary with respect to several demographic factors, and \citet{Mu2018}, which estimated spatial variation in the effects of wind speed and air temperature on air pollution concentration; see \citet{FanZhang2008} and \citet{Franco-Villoria2019} for comprehensive reviews of VCMs. 
VCMs are also useful in causal inference: \citet{Hahn2020}'s Bayesian Causal Forest (BCF) model, which has proven highly effective at uncovering treatment effect heterogeneity \citep{Dorie2019_automated}, is just a VCM with a single covariate (the binary treatment indicator).

Early work on VCMs largely focused on settings with a single effect modifier where it is popular either (i) to approximate each function $\beta_{j}(\bZ)$ with a linear combination of basis functions \citep[see, e.g.,][]{Huang2002, Wang2008, Wei2011, Bai2023} or (ii) to estimate the $\beta_{j}(\bZ)$'s in a fully nonparametric fashion with kernel smoothing \citep{Hoover1998, Wu2000, Li2002, Chiang2011}.
In more modern settings, when the $\beta_{j}(\bZ)$'s might depend on complex interactions between $R > 1$ modifiers, several authors have used regression trees to estimate the covariate effects in a VCM \citep[see, e.g.,][]{WangHastie2012,Burgin2015,ZhouHooker2022}. 
Generally speaking, tree-based approaches are better equipped to capture \textit{a priori} unknown interactions and scale much more gracefully with $R$ and the number of observations $N$ than kernel methods like the one proposed in \citet{LiRacine2010}, which involves intensive hyperparameter tuning.

Recently, \citet{Deshpande2024} introduced VCBART, a fully Bayesian method for fitting VCMs that exhibits superior covariate effect recovery and uncertainty quantification compared to the current state-of-the-art without requiring any hand-tuning.
VCBART approximates each $\beta_{j}(\bZ)$ with its own regression tree ensemble, avoiding overfitting with a regularizing prior that encourages each tree to be a ``weak learner''.
Across several synthetic and real datasets, VCBART displayed state-of-the-art empirical performance.
\citet{Deshpande2024} also derived posterior contraction rates for VCBART in the regime where $p$ was fixed and all the $\beta_{j}(\bZ)$'s were non-null.

In many applications, however, only a small subset $S_0\subset\{1,\dots,p\}$ of covariates have a non-null effect,.
Further, for each covariate $j\in S_0$, the function $\beta_j(\bZ)$ often only depends on just a few modifiers $S_{0,j}\subset\{1,\dots,R\}$. 
To meet the demands of such applications, we introduce \texttt{sparseVCBART}, which \emph{simultaneously} learns covariate sparsity and per-coefficient modifier sparsity. 
At the covariate level, we place global-local shrinkage priors on the output in each regression tree leaf so that the $\beta_{j}(\bZ)$ functions are adaptively shrunk towards zero \citep{Carvalho2010,Polson2012,bhadra2016gl}. 
At the modifier level, we endow each coefficient’s tree ensemble with a heavy-tailed Dirichlet prior over split variables, which concentrates mass on a small subset of modifiers relevant to $\beta_j(\bZ)$.

Our main theoretical results (\Cref{thm:contraction-main,thm:minimax-main}) show that the \texttt{sparseVCBART} posterior contracts at nearly the minimax-optimal rate $r_{N}$ where
\[
r_N^2 =N^{-1}{(d_0\log p)}+\sum_{j\in S_0} (\log N)N^{-{2\alpha_j}/{(2\alpha_j+s_{0,j})}},
\]
$d_{0} = \lvert S_{0}\rvert$ is the number of covariates with non-null effects, and $s_{0,j} = \lvert S_{0,j}\rvert$ is the number of effect modifiers on which $\beta_{j}(\bZ)$ depends. 
Our theoretical analysis significantly extends the work of \citet{Deshpande2024}. In the more challenging regime where $p$ grows with $N,$ we establish a faster contraction rate for estimating sparse VCMs than that obtained in their analysis.
We further show that the posterior over the splitting probabilities in each tree concentrates on the true modifier sets $S_{0,j}$ (\Cref{thm:modsel}) and that the posteriors for the local scale parameters corresponding to null $\beta_{j}(\bz)$'s concentrate around zero (\Cref{thm:lambda-shrink}).
These latter results essentially imply that, at least asymptotically, \texttt{sparseVCBART} can identify the modifiers that drive the truly non-null covariate effects while simultaneously shrinking away truly null effects.

%adaptively to \emph{unknown} covariate sparsity ($d_0=|S_0|$) and modifier sparsity ($s_{0,j}=|S_{0,j}|$) without needing to know either in advance (Theorem~\ref{thm:contraction-main}). The rate improves over \citet{Deshpande2024} using the ambient modifier dimension $R$ by replacing it with the intrinsic $s_{0,j}$ for each active coefficient. Moreover, our Dirichlet split prior consistently concentrates on the true modifier sets $S_{0,j}$ for active coefficients (Theorem~\ref{thm:modsel}), and the RHS prior drives truly null $\beta_j(\cdot)$’s toward zero while allowing signals to escape shrinkage (Theorem~\ref{thm:lambda-shrink}). 
Empirically in \Cref{sec:experiments}, across synthetic and semi-synthetic high-dimensional VCMs, \texttt{sparseVCBART} matches or improves predictive accuracy relative to strong baselines while delivering markedly sparser models: it prunes irrelevant covariates and modifiers and provides coherent posterior uncertainty for the remaining effects. 
These properties make \texttt{sparseVCBART} a practical, theoretically grounded tool for high-dimensional varying coefficient analysis well-suited to numerous applications.

\section{Proposed Model}
\label{sec:model}
We observe $N$ triplets $(\bx_{1}, y_{1}, \bz_{1}), \ldots, (\bx_{N}, y_{N}, \bz_{N})$ of a covariate vector $\bx \in \R^{p},$ scalar output $y \in \R,$ and an effect modifier vector\footnote{For simplicity, we present our model in the setting with only continuous effect modifiers. In our actual implementation, we allow for both categorical and continuous modifiers. See \citet{Deshpande2025} for formal definitions and priors for regression trees defined over mixed input spaces.} $\bz \in [0,1]^{R}$ and model
\begin{equation}
\label{eq:vcmodel}
y_{i} = \beta_{0}(\bz_{i}) + \sum_{j = 1}^{p}{\beta_{j}(\bz_{i})x_{ij}} + \sigma \epsilon_{i},
\end{equation}
where $\epsilon_{1}, \ldots, \epsilon_{N}$ are i.i.d. $\normaldist{0}{1}$ and $\beta_{0}(\bZ), \ldots, \beta_{p}(\bZ)$ are unknown functions mapping $[0,1]^{R}$ to $\R.$
Following \citet{Deshpande2024}, we (i) represent each unknown function $\beta_{j}$ in \Cref{eq:vcmodel} as a sum of $M_{j}$ regression trees; (ii) specify a regularizing prior over the regression trees that discourages overfitting; (iii) compute a posterior distribution over regression tree ensembles; and (iv) use Markov chain Monte Carlo to approximate posterior summaries.
The key difference between our proposed \texttt{sparseVCBART} and the original VCBART model is our use of a global-local prior to model the outputs of each regression tree.

\textbf{Notation}.
Formally, a regression tree is a pair $(\calT, \calM)$ of a (i) rooted binary decision tree $\calT$ with $\calL(\calT)$ terminal leaf nodes and several non-terminal decision nodes and (ii) a collection of scalar jumps $\calM = \{\mu_{1}, \ldots, \mu_{\calL(\calT)}\},$ one for each leaf.
Each decision node is associated with a decision rule of the form $\{Z_{r} < t\}.$
Given $\calT,$ we can associate every $\bz \in [0,1]^{R}$ with a single leaf node $\ell(\bz; \calT)$ by tracing a path down from the root node that, upon encountering the decision rule $\{Z_{r} < t\}$, proceeds to the left (resp. right) if $z_{r} < t$ (resp. $z_{r} \geq t$). 
By associating each leaf $\ell$ in $\calT$ with a scalar jump $\mu_{\ell},$ the pair $(\calT, \calM)$ represents a piecewise-constant step function.
Given a regression tree $(\calT, \calM)$, let $g(\bz; \calT, \calM) = \mu_{\ell(\bz; \calT)}$ be the evaluation function returning the jump associated with leaf $\ell(\bz).$

%Given $\calT,$ we can associate every $\bz \in [0,1]^{R}$ with a single leaf node $\ell(\bz; \calT)$ by tracing a path down from the root node that, upon encountering the decision rule $\{Z_{r} < t\}$, proceeds to the left (resp. right) if $z_{r} < t$ (resp. $z_{r} \geq t$). 
%Given a regression tree $(\calT, \calM)$, let $g(\bz; \calT, \calM) = \mu_{\ell(\bz; \calT)}$ be the evaluation function returning the jump associated with leaf $\ell(\bz).$

\subsection{Prior Specification}
\label{sec:prior}
For each $j = 0, \ldots, p,$ we introduce an ensemble of $M_{j}$ regression trees $\calE_{j} = \{(\calT^{(j)}_{m}, \calM^{(j)}_{m})\}_{m = 1}^{M_{j}}$ and express 
$\beta_{j}(\bz) = \sum_{m = 1}^{M_{j}}{g(\bz; \calT_{m}^{(j)}, \calM_{m}^{(j)})}.$
We denote the collection of all ensembles with $\bcalE.$

\textbf{Decision tree prior.} We model individual decision trees $\calT_{m}^{(j)}$ in $\calE_{j}$ as \textit{a priori} i.i.d. and further model the trees as \textit{a priori} independent across ensembles.
To sample a random decision tree from the prior, we first use a Galton-Watson branching process to generate the graphical structure (i.e., the arrangement of leaf and decision nodes) and then draw the decision rules associated to each decision node.
To ensure the tree is of finite depth, we calibrate the branching process so that the probability that a node at depth $d$ is non-terminal is $0.95\times(1+d)^{-2}.$
Having drawn the graphical structure, we draw a decision rule for each non-terminal node in two steps.
For a tree in the ensemble $\calE_{j},$ we select the splitting variable $r \sim \textrm{Multinomial}(\theta_{j1}, \ldots, \theta_{jR})$ and then, conditionally on $r,$ draw the cutting point $t$ uniformly.
We specify independent priors for the vectors $\btheta_{j} := (\theta_{j1}, \ldots, \theta_{jR}) \sim \textrm{Dirichlet}(\eta_{j}/R, \ldots, \eta_{j}/R)$ with hyperpriors $\eta_{j}/(\eta_{j} + R) \sim \betadist{1}{0.5}$.
This prior hierarchy, which was first introduced by \citet{Linero2018}, encourages the regression trees in each ensemble to split on only a few effect modifiers, thereby encouraging modifier sparsity. 

\textbf{Global--local jump prior.}
The original VCBART model specified independent $\normaldist{0}{1/(4M_{j})}$ priors on all the jumps in each ensemble $\calE_{j},$ which implies the marginal prior $\beta_{j}(\bz) \sim \normaldist{0}{1/4}.$
So, the original VCBART applies the exact same amount of shrinkage to all functions $\beta_{j}(\bz).$
Within the Bayesian variable selection literature, it is well-established that \emph{adaptive shrinkage} methods, which shrink smaller signals to zero more aggressively than larger signals, recover underlying sparsity patterns much more effectively than constant shrinkage methods \citep{Polson2012, bhadra2016gl,GeorgeRockova2020_penalty_mixing}. 

On this view, to encourage selective shrinkage of the $\beta_{j}(\bz)$'s in our high-dimensional setting, we specify regularized horseshoe priors \citep[RHS;][]{PiironenVehtari2017} on our regression tree jumps.
Specifically, we introduce a shared global parameter $\tau>0$, a slab parameter $c>0$, and local scales $\lambda_{0},\ldots,\lambda_{p}$ and model
\[
\mu^{(j)}_{m,\ell}\mid \lambda_j,c,\tau \sim \normaldist{0}{s_{j}^{2}},
\]
where $s_{j}^{2} = \tau^{2}\lambda_{j}^{2}\,c^{2}/(c^{2} + \tau^{2}\lambda_{j}^{2}).$
All leaves in $\calE_{j}$ share the same local scale $\lambda_j$; when $\lambda_j$ shrinks toward zero, \emph{every} leaf in that coefficient is pulled to (near) zero, effectively switching off $\beta_j(\bz)$ and yielding covariate-level sparsity. 
The global scale $\tau$ controls the overall amount of shrinkage across ensembles, adapting to the number of truly active coefficients. 
The finite slab $c$ prevents over-shrinkage of large signals and stabilizes inference relative to the plain horseshoe prior \citep{pmlr-v54-piironen17a,PiironenVehtari2017}, especially when signals are strong but sparse. 

We endow the local and global scale parameters with independent half-Cauchy priors: \(
\lambda_j\sim \mathcal{C}^+(0,1) \), and \(\tau\sim \mathcal{C}^+(0,\tau_0)\).
The slab variance uses the conjugate prior $c^2\sim\mathcal{IG}(\nu_c/2,\nu_cs_c^2/2)$. 
Finally, we specify the conjugate prior $\sigma^{2} \sim \mathcal{IG} \big({\nu}/{2},\,{\nu\lambda}/{2}\big)$, where $\nu$ and $\lambda$ are fixed hyperparameters.

%As shown in \Cref{sec:experiments}, tying all leaf jumps within an ensemble to a covariate-specific local scale \(\lambda_j\) under a shared global scale $(\tau)$ induces \emph{covariate-level} sparsity: truly null coefficients are shrunk as a group toward zero, while large signals pass through the slab. 
%This selective shrinkage yields more accurate support recovery and substantially tighter, better-calibrated credible bands in high dimensions (see \Cref{fig:funrec}).

\subsection{Posterior Computation}
\label{sec:sampler}

We generate posterior draws with a high-level Gibbs sampler that sequentially updates each of $\bcalE,$ $\bTheta \coloneqq\{\btheta_j\}_{j=0}^{p},$ $\bEta \coloneqq \{\eta_j\}_{j=0}^{p},$ $\blambda \coloneqq \{\lambda_j\}_{j=0}^{p},$ $\tau,$ $c,$ and $\sigma^{2},$ conditionally on all the others. 
Here, we briefly describe each Gibbs update and defer the full sampler derivation to \switchref{\Cref{sec:extraalgorithm}}{Section S2 of the Supplementary Materials}.
We update each of $\bcalE, \bTheta, \bEta,$ and $\sigma^{2}$ using essentially the same sampler as the one developed by \citet{Deshpande2024} for the original VCBART model.
Compared to that sampler, each iteration of our \texttt{sparseVCBART} sampler includes additional steps to update $\blambda, \tau,$ and $c.$
%As we describe below, we update each of $\bcalE, \bTheta, \bEta,$ and $\sigma^{2}$ using essentially the same strategy as \citet{Deshpande2024} and update each of $\{\lambda\}_{j=0}^{p}$ and  $\tau$ by univariate slice sampling on the log-scale. 
%The slab variance $c^{2}$ is updated by a conjugate inverse-gamma draw.

\textbf{Updating $(\calT^{(j)}, \calM^{(j)}_{m}$).} 
We update the trees in $\bcalE$ one at a time and update individual regression trees in two steps.
First, we draw a new decision tree $\calT$ with a Metropolis-Hastings step that involves randomly growing or pruning the existing tree.
Then, we draw new jumps $\calM$ conditionally on the new tree decision tree.
This sampling strategy is enabled by the fact that given $\by, \blambda, c, \sigma^{2}$, and all other trees, the conditional posterior density of a single decision tree $\calT$ and the conditional distribution of its jumps $\calM \vert \calT$ are available in closed form (\switchref{\Cref{eqS:treemarg,eqS:leafpost}}{Equations S2.4 \& S2.6}).

\begin{comment}
To see this, suppose we are updating tree $m$ in ensemble $\calE_{j}.$
Let $r_{1}, \ldots, r_{n}$ be the partial residuals\footnote{For compactness, we suppress the dependence on $j$ and $m$ from the notation. For additional compactness in \Cref{eq:partial_resid}, we set $x_{i0} = 1$ for the intercept term}
\begin{align}
\begin{split}
\label{eq:partial_resid}
r_{i} = y_{i} &- \sum_{j' \neq j}{\sum_{m' = 1}^{M_{j}}{x_{ij'}g(\bz_{i}; \calT^{(j')}_{m'}, \calM^{(j')}_{m'})}} \\
&- \sum_{m' \neq m}{x_{ij}g(\bz_{i}; \calT^{(j)}_{m}, \calT^{(j)}_{m'})}
\end{split}
\end{align}
Additionally, for each leaf $\ell$ of an arbitrary tree, let $\mathcal{I}(\ell; \calT) = \{i : \ell(\bz_{i};\calT) = \ell\}$ be the set of indices of observations associated with leaf $\ell$ in tree $\calT.$
Further let $\br_{\ell} = (r_{i}: i \in \calI(\ell, \calT))$ and $\bz_{j,\ell} = (x_{ij}: i \in \calI(\ell,\calT))$ be the vector of partial residuals and covariate values corresponding to observations in leaf $\ell.$
The conditional posterior density of $(\calT_{m}^{(j)}, \calM_{m}^{(j)})$ factorizes over the leafs of $\calT_{m}^{(j)}$:
\begin{align}
\begin{split}
\label{eq:tree_joint_posterior}
p(\calT, \calM \vert \ldots) &\propto p(\calT)\prod_{\ell = 1}^{\calL(\calT)}{s_{j}^{-1}\exp\left\{-\frac{\mu_{\ell}^{2}}{2s_{j}^{2}}\right\}} \\
&~\times \prod_{\ell = 1}^{\calL(\calT)}{\exp\left\{-\frac{\lVert \br_{\ell} - \mu_{\ell}\bx_{j,\ell} \rVert^{2}_{2}}{2\sigma^{2}}\right\}}
\end{split}
\end{align}
\end{comment}

\textbf{Updating $\bTheta, \bEta,$ and $\sigma^{2}$}. We respectively update each vector of splitting probabilities $\btheta_{j}$ and $\sigma^{2}$ with standard conjugate updates (\switchref{\Cref{eqS:etaTarget,eqS:sigma2}}{Equations S2.8 \&  S2.11}). We update each $\eta_{j}$ with a Metropolis-Hastings step.

\textbf{Updating $\blambda$, $\tau$, and $c$.} 
It turns out that the conditional posterior density $\pi(\blambda, \tau, c \vert \ldots)$ factorizes, allowing us to update each parameter separately (\switchref{\Cref{eqS:GLtarget}}{Equation S2.9}).
We specifically use one-dimensional slice sampling to update each $\lambda_{j}$'s and $\tau.$ 
We draw a new $c^{2}$ using a conjugate normal-inverse gamma update (\switchref{\Cref{eqS:c2update}}{Equation S2.10}). 

\subsection{Default Hyperparameter Settings}
\label{sec:hyperparameters}

\texttt{sparseVCBART} depends on several hyperparameters: the number of trees $M_{j}$ in ensemble $\bcalE_{j}$; the scale $\tau_{0}$ for the prior on the global scale $\tau;$ and the inverse-gamma hyperparameters $\nu_{c}, s_{c}, \nu,$ and $\lambda$ for the slab variance $c^{2}$ and the residual variance $\sigma^{2},$ respectively. We recommend setting each $M_{j} = 50$ and, like \citet{Chipman2010}, setting $\nu = 3$ and $\lambda$ such that the $\mathcal{IG}(\nu/2, \nu\lambda/2)$ density places 90\% probability on the event that $\sigma < \textrm{sd}(\by),$ the observed standard deviation of the outcomes.
We further recommend setting \(\tau_{0} = {p}/{(p - p_0)}\cdot {\mathrm{sd}(\bm{y})}/{\sqrt{N}},\) where \(p_0 = \min\Big\{10, \max\big(1, \lfloor p/4 \rfloor\big)\Big\}\), as recommended by \citet{pmlr-v54-piironen17a}.
Finally, we recommend setting $\nu_{c} = 4$ and $s_{c} = 2$.
In our experiments, we have found that these settings generally yield very accurate results; see \switchref{\Cref{sec:sensitivity}}{Section S3.2 in the Supplementary Materials} for a comprehensive sensitivity analysis on the ensemble size $M_j$.

\section{Theoretical Results}
\label{sec:theory}
In this section, we analyze the asymptotic behavior of \texttt{sparseVCBART} in a high–dimensional regime where both $p$ and $R$ grow with the sample size $N$. 
Under standard design and regularity conditions, we show that the posterior contracts around the truth at a rate that is nearly minimax–optimal, while adapting to unknown sparsity and smoothness. 
We further establish that the posterior consistently identifies which covariates $X_j$ have non–null effects on $Y$ and, for each active covariate, which modifiers $Z_r$ drive its effect. 
Finally, we prove that the posterior distributions of the local shrinkage scales for truly null covariates collapse to zero as $N$ diverges.

To study \texttt{sparseVCBART} theoretically, we assume the data were generated from a true VCM indexed by a vector of coefficient functions $\bbeta_0 = (\beta_{0,0}, \dots, \beta_{0,p})$ and a noise variance $\sigma_0^2$. Collectively, we call $\theta_0 \coloneqq (\bbeta_0, \sigma_0 ^2)$. We assume that the set of active covariates, $S_0 \subset \{1, \dots, p\}$, is small, with size $d_0 := |S_0| \ll p$. 
For any inactive covariate $j \notin S_0$, we assume that the true function is null (i.e., $\beta_{0,j}(\bZ) \equiv 0$). 
For each active covariate $j \in S_0$, we assume that its coefficient $\beta_{j}(\bZ)$ depends only on a small subset of modifiers $S_{0,j} \subset \{1, \ldots, R\}$ of size $s_{0,j} := |S_{0,j}| \ll R$. 
Finally, we assume that for each $j \in S_{0},$ the truly non-null function $\beta_{0,j}(\bZ)$ lies in a particular H\"{o}lder class (see Assumption \textbf{(A4)}).

\textbf{Notation.} Because we let the number of covariates grow with $N,$ we will use the notation $p_{N}$ instead of $p.$
For any index set $S,$ let $\bz_{S} = (z_{u})_{u \in S}.$
We denote the empirical $L_{2}$ norm $\|f\|_N:=\{N^{-1}\sum_{i=1}^N f(\bm Z_i)^2\}^{1/2}.$
With $\alpha>0$, let $s=\lfloor\alpha\rfloor$ and $\zeta=\alpha-s$, for a multi–index $\bm{k}\in\mathbb N^R$ with $|\bm{k}|=\sum_{r=1}^{R} k_r$, we write $D^{\bm{k}} f$ for the mixed derivative. 
The H\"older ball of radius $B>0$ on $[0,1]^R$ is
\begin{equation*}
\begin{split}
\mathcal H_B^{\alpha}\big([0,1]^R\big)
=
\Big\{ f:[0,1]^R\!\to\!\mathbb R :
\max_{|\bm{k}|\le s}\|D^{\bm{k}} f\|_\infty \\+
\max_{|\bm{k}|=s}\ \sup_{\bm{x}\neq \bm{y}}
\frac{|D^{\bm{k}} f(\bm{x})-D^{\bm{k}} f(\bm{y})|}{\|\bm{x}-\bm{y}\|_\infty^{\zeta}}
\ \le\ B \Big\}.
\end{split}
\end{equation*}

% For any $\alpha > 0$ and constant $B > 0,$ let $\mathcal{H}^{\alpha}_{B}([0,1]^{R})$ be the H\"{o}lder ball
% \[
% \mathcal{H}^{\alpha}_{B}\big([0,1]^R\big)
% = \Big\{\, f:[0,1]^R\to\mathbb R : \|f\|_{C^{\alpha}([0,1]^R)}\le B \Big\},
% \]
% where the H\"older norm $
% \|f\|_{C^{\alpha}([0,1]^R)}
% :=\max_{|k|\le s}\,\|D^k f\|_{\infty}
% +
% \max_{|k|=s} \sup_{x\ne y}
% \frac{|D^k f(x)-D^k f(y)|}{\|x-y\|_{\infty}^{\beta}}.
% $
For fixed sequences $a_N,b_N>0$, $a_N=\mathcal O(b_N)$ means $\sup_N a_N/b_N<\infty;$ $a_N=o(b_N)$ means $a_N/b_N\to0;$ and $a_N=\Theta(b_N)$ means $a_N=\mathcal O(b_N)$ and $b_N=\mathcal O(a_N)$. 
Similarly, for a random sequence $\{X_{N}\},$ we write $X_{N} = \mathcal{O}_{P}(a_{N})$ to mean that $X_{N}/a_{N}$ is bounded in probability.
We write $a\lesssim b$ if $a\le C\,b$ for an absolute constant $C$ independent of $(N,p_N,R)$; $a\gtrsim b$ is defined analogously. For a set $A\subset\mathbb R^d$, $\mathrm{vol}(A)$ denotes its $d$–dimensional Lebesgue measure.

\subsection{Posterior contraction}
\label{sec:postconc}
To establish posterior contraction, we make the following assumptions:
\begin{itemize}
   \item[\textbf{(A1)}]The number of trees per coefficient is uniform and bounded: $M_{j} \equiv M$ for all $j$, with $M=\mathcal O(1)$.\footnote{Our theoretical results are unaffected if $M$ grows slowly (e.g., $M = \mathcal O(\log N)$), since the proofs depend on the total leaf budget per ensemble, not the specific number of trees. Allowing $M$ to grow merely changes constants and bookkeeping.}

   \item[\textbf{(A2)}] The sample $\{\bm Z_i\}_{i=1}^N\subset[0,1]^R$ is \emph{k-d–regular}:   
    there exist constants $0<c_\mathrm{kd}\le C_\mathrm{kd}<\infty$ such that for any axis-aligned rectangle $A\subset[0,1]^R$,
\[
c_\mathrm{kd}\,\mathrm{vol}(A)
\le
\frac{1}{N}\#\{i:\bm Z_i\in A\}
\le
C_\mathrm{kd}\,\mathrm{vol}(A).
\] 
\item[\textbf{(A3)}] $\max_{i\le N}\max_{1\le j\le p_N}|X_{ij}|\le D<\infty$.
\item[\textbf{(A4)}] $\beta_{0,j} \equiv 0 \ \  \forall j \notin S_0,$ and for each $j \in S_0,$ $\exists \tilde{\beta}_{0,j} \in \mathcal{H}_{B} ^{\alpha_j} \left([0,1]^{s_{0,j}}\right)$ with $0 < \alpha_j \le 1$ such that $\beta_{0,j}(\bm z) = \tilde{\beta}_{0,j}(\bm{z}_{S_{0,j}}).$

\item[\textbf{(A5)}] There exists $0<\xi<\min_{j \in S_0} s_{0,j}/(2 \alpha_j + s_{0,j}),$ such that $\log p_N=\mathcal O(N^\xi)$, $d_0=\mathcal O(\log N)$, and $R=o(\log N)$.

\item[\textbf{(A6)}] There exist constants $C\ge 1$ and $0<c_{\min}\le c_{\max}<\infty$ such that for any
$S\subset\{1,\dots,p_N\}$ with $|S|\le C d_0$ and any arrays $u_{ij}\in\mathbb R$,
\begin{equation*}
\begin{split}
c_{\min}\sum_{j\in S}\frac{1}{N}\sum_{i=1}^N u_{ij}^2
 & \le
\frac{1}{N}\sum_{i=1}^N\Big(\sum_{j\in S} X_{ij} u_{ij}\Big)^2 \\
& \le
c_{\max}\sum_{j\in S}\frac{1}{N}\sum_{i=1}^N u_{ij}^2.
\end{split}
\end{equation*}
\end{itemize}

Assumptions \textbf{(A1)} -- \textbf{(A3)} are basic BART regularity conditions: \textbf{(A1)} bounds model complexity, the k-d regularity in \textbf{(A2)} ensures good approximation properties for the tree ensembles \citep{Rockova2019}, and \textbf{(A3)} provides numerical stability by preventing explosion of the prediction norms. 
The high-dimensional regime in \textbf{(A5)} allows for sub-exponentially growing $p_N$, a strictly more challenging setting than the fixed-$p$ regime studied in \citet{Deshpande2024}. 
Our key statistical assumption is the functional restricted-eigenvalue condition \textbf{(A6)}, a functional analogue of standard restricted eigenvalue conditions in sparse regression \citep[\S3]{Bickel_2009} that precludes pathological collinearity and enables the construction of powerful tests and stable identifiability: it lets us translate empirical function error $\sum_{j \in S}\|\beta_j-\beta_{0,j}\|_N ^2$ into prediction error $\|\sum_{j \in S} X_{ij}(\beta_j - \beta_{0,j})^2 \|_N ^2,$ preventing near-collinear designs from masking nonzero effects. 
Finally, the smoothness constraint $\alpha_j \le 1$ in $\textbf{(A4)}$ is an inherent consequence of using piecewise-constant step-functions to approximate the $\beta_{j}(\bZ)$'s \citep{jeong2023artbartminimaxoptimality}; this limitation can be relaxed by using the so-called ``soft'' decision trees (as in \citet{LineroYang2018}) or replacing constant jumps with smoother functions (as in \citet{yee2025scalablepiecewisesmoothingbart}).

To control model capacity while retaining covariate and modifier-specific adaptivity, we make the following prior assumptions on the tree topology, modifier-split coordinates, leaf values with covariate-level global–local shrinkage, and the noise variance. 

\begin{itemize}
\item[\textbf{(P1)}] Each regression tree follows a depth-penalized Galton–Watson process:
a node at depth $m$ splits with probability $\gamma^{m}$ where $N^{-1}<\gamma<1/2$ is fixed. 
% \item[\textbf{(P2)}] For each covariate $j$, we employ the prior on the split-coordinate probabilities 
% $\bm{\theta}_j=(\theta_{j1},\dots,\theta_{jR})$ 
% \[
% \bm \theta_j\sim\mathrm{Dir}\!\left({\eta_j}/{R},\dots,{\eta_j}/{R}\right),
% \ 
% \eta_j\sim\pi(\eta)\propto (R+\eta)^{-(R+1)}.
% \]

\item[\textbf{(P2)}] We draw $d\in\{0,1,\dots,p_N\}$ with $\pi(d)\propto 1/d!$. The global scale $\tau\mid d \sim \mathcal C^+\!\bigl(0, A_{g,d}\bigr)$ with
$A_{g,d}=\Theta\left(dp_N^{-1}(\log p_N)^{-2}(\log N)^{-2}\right)$.

\item[\textbf{(P3)}] Let $r_N^2$ denote the target contraction rate (stated later). We use a mildly informative prior
\[
\sigma^2 \sim \mathcal{IG}\!\left(\frac{\nu_N}{2},\,\beta\right),
\quad \nu_N=\mathcal O(N r_N^2), \quad \beta>0\ \text{fixed}.
\]

% \emph{Local scales.} $\lambda_j \stackrel{\text{i.i.d.}}{\sim}\mathcal C^+(0,1)$ for $j=0,1,\dots,p_N$.

% \emph{Leaf values.} For every leaf $\ell$ in tree $m$ for coefficient $j$,
% \[
% \mu^{(j)}_{\ell,m} \,\Big|\, \tau,\lambda_j
% ~\sim~
% \mathcal N\!\Big(0,\ \frac{1}{M}\,\sigma_{j,\text{leaf}}^2\Big),
% \]
% with the regularized horseshoe prior by introducing a slab scale $c>0$
% (e.g., $c\sim\mathcal C^+(0,1)$ independent of $N$) and set    \[
%       \sigma_{j,\text{leaf}}^2
%       = \tau^2\lambda_j^2\,\frac{c^2}{c^2+\tau^2\lambda_j^2}.
%     \]
% \paragraph{(P4) Residual variance.}
% Let $r_N^2$ denote the target contraction rate (stated later). We use a mildly informative prior
% \[
% \sigma^2 \sim \mathcal{IG}\!\left(\frac{\nu_N}{2},\,\beta\right),
% \quad \nu_N=\mathcal O(N r_N^2), \quad \beta>0\ \text{fixed}.
% \]
% This choice simplifies sieve mass control and does not affect rates.
\end{itemize}
% The Galton-Watson process \textbf{(P1)} is the primary regularization mechanism, penalizing tree depth to keep the individual learners simple \citep{rockovasaha}. 

Assumption \textbf{(P1)} specifies a slightly different Galton-Watson prior on the decision tree structure than what we presented in \Cref{sec:model}: under \textbf{(P1)}, the probability that the tree continues to grow decays exponentially with depth as opposed to quadratically. 
As \citet{rockovasaha} show, such an exponential decay is necessary for posterior contraction.
Moreover, the restriction $\gamma \in (N^{-1},1/2)$ is a key technical condition that balances two opposing needs: the upper bound ensures trees remain small to provide regularization, while the lower bound is required to place sufficient prior mass on the specific k-d approximants used in the posterior concentration proof of \Cref{thm:contraction-main}.
Taking $\pi(d) \propto 1/d!$ in \textbf{(P2)} imposes a super-exponential penalty on larger supports, discouraging overfitting to many covariates. 
The scaling choice of $\tau$ is also utilized to keep cumulative false activation across $p_N$ covariates under control. 
Prior \textbf{(P3)} acts as a weakly informative prior in practice; tying the shape parameter $\nu_N$ to the target contraction rate $r_N ^2$ is a crucial step that facilitates the proof of our main theorems by ensuring the prior for the noise variance is appropriately concentrated. 
\Cref{thm:contraction-main} shows that the \texttt{sparseVCBART} posterior contracts around the true function at a rate $r_N$. 
This rate is adaptive, composed of a parametric component for covariate selection $(N^{-1}d_0 \log p_N)$ and a sum of nonparametric components ($N^{-2\alpha_j/(2 \alpha_j+s_{0,j})}$) whose complexity depends on the intrinsic dimension and smoothness of each active coefficient ($s_{0,j}$). 
For $j \in S_0,$ we define the dimension-adaptive rate $r_{jN,\text{ad}} ^2 \coloneqq (\log N)N^{-2\alpha_j/(2 \alpha_j+s_{0,j})}$ and the overall target rate as $r_N^2
\coloneqq
N^{-1}(d_0\log p_N) + \sum_{j\in S_0} r_{jN,\mathrm{ad}}^2.$
\begin{theorem}
\label{thm:contraction-main}
Under \textbf{(A1)}-\textbf{(A6)} and \textbf{(P1)}-\textbf{(P3)}, there exists $C>0$ such that as $N \rightarrow \infty$,
\[
\Pi \ \Big(\ \|\bm\beta-\bm\beta_0\|_N > C\,r_N \ \Bigm|\ \bm{Y}\Big)
\;\longrightarrow\; 0
\ \text{in $P_{\bbeta_0}$–probability.}
\]
\end{theorem}

We prove \Cref{thm:contraction-main} in \switchref{\Cref{sec:proofs}}{Section S1 of the Supplementary Materials} using the general framework outlined in \citet{Ghosal2000}. 
First, we prove \switchref{\Cref{lem:approx-main}}{Lemma S1.1}, which shows that any sparse H\"older function can be well-approximated by a tree ensemble.
More specifically, for each active covariate $j \in S_{0},$ we construct a piecewise-constant step function $g_{j}$ on a balanced k-d partition that splits only along modifiers in $S_{0,j}.$
We then show that $g_{j}$ can be represented as a regression tree with $K_j \asymp N r_{jN,\textrm{ad}}^2/\log N$ leaves and bound the approximation error $\| \beta_{0,j}-g_j \|_N \lesssim N^{-\alpha_j/(2\alpha_j + s_{0,j})}.$
\switchref{\Cref{lem:smallball-main,lem:tree-mass}}{Lemmas S1.2 and S1.3} show that our prior assigns sufficient mass (i.e.,$\exp(-c N r_{jN,\mathrm{ad}}^2)$) in a small ball around each efficient approximation $g_{j}.$
Finally, \switchref{\Cref{lem:sieve-mass_corrected,lem:entropy,lem:tests}}{Lemmas S1.4--S1.6} construct a sieve with controlled metric entropy and uniformly powerful hypothesis tests, which are used to verify the conditions of \citet{Ghosal2000}'s general theorem for posterior contraction.

The rate $r_N$ in \Cref{thm:contraction-main} improves upon the original VCBART theory (see Theorem 3.1 in \citet{Deshpande2024}) in two key aspects. 
First, it formally handles the high-dimensional setting where the number of covariates $p_N$ grows with $N$, a case not covered by the original analysis that assumes $p_{N}$ is fixed. 
Second, the nonparametric component of the rate adapts to the intrinsic modifier dimension $s_{0,j}$ for each coefficient, which is faster than the VCBART contraction rate $(\log N)\sum_{j \in S_0} N^{-2 \alpha_j/(2 \alpha_j +R)}$ whenever the modifier structure is truly sparse $(s_{0,j}<R)$. 
\Cref{thm:minimax-main} shows that the rate $r_{N}$ from \Cref{thm:contraction-main} is, up to a logarithmic factor, the minimax optimal rate for learning $\bbeta_{0}$ in the empirical $L_2$ norm.
%--------------------------
% Minimax lower bound
%--------------------------
\begin{theorem}
\label{thm:minimax-main}
Under \textbf{(A1)}–\textbf{(A6)}, there exists $c>0$ (independent of $N,R,p_N$) such that
\begin{align*}
& \inf_{\widehat{\bm\beta}}
 \sup_{\bm\beta_0\in\mathcal F}
 \E_{\bm\beta_0}\big\|\widehat{\bm\beta}-\bm\beta_0\big\|_N^{2}
\ge \\
&c\Bigg\{
\frac{d_0\log (p_N/d_0)}{N} 
+ \sum_{j\in S_0} N^{-\frac{2\alpha_j}{2\alpha_j+s_{0,j}}}
\Bigg\},
\end{align*}
where $\mathcal F$ denotes the class in $\textbf{(A4)}$. \end{theorem}

%--------------------------
% Proof trajectory (short)
%--------------------------
We prove \Cref{thm:minimax-main} in \switchref{\Cref{sec:minimax}}{Section S1.4 of the Supplementary Materials}.
The lower bound combines a Gilbert–Varshamov packing over covariate supports (yielding the $N^{-1}(d_0\log (p_N/d_0))$ term) with an Assouad bump function construction \citep[][Lemma 2.12]{Tsybakov2009} in the modifier variables for each active coefficient (yielding the $N^{-2\alpha_j/(2\alpha_j+s_{0,j})}$ terms). 
Under \textbf{(A5)} we have $\log(p_N/d_0) \asymp \log p_N;$ hence the selection term in $r_N$ matches the minimax lower bound’s $N^{-1}d_0 \log(p_N/d_0)$ up to constants. 
The only gap is a benign $\log N$ factor which arises from standard entropy controls for sieve priors over step function partitions, making $r_N$ \emph{near}-minimax rather than exactly minimax. 
The two-term rate in \Cref{thm:minimax-main} has been observed in several other high-dimensional contexts: in particular, \citet{klopp2014} establish an analogous rate for a sparse VCM. 
% For example, \citet{raskutti2009minimaxratesestimationhighdimensional} derive a minimax lower bound for sparse additive models that decouples into a $(d_0 \log p)/N$ subset-selection cost and an $s_0$-dimensional nonparametric estimation cost.
\subsection{Modifier selection consistency}

When fitting VCMs with multiple modifiers ($R > 1$), a key goal is identifying which of them truly drive the non-null effects.
\Cref{thm:modsel} shows that the posterior for the split-variable probabilities $\bm{\theta}_j$ correctly concentrates on the true active modifier set.

\begin{theorem}
\label{thm:modsel}
Assume \textbf{(A1)}–\textbf{(A6)} and \textbf{(P1)}–\textbf{(P3)}. Fix $j\in S_0$ and let
$S_{0,j}$ be the true active modifier set for $\beta_{0,j}$.
Then for any $\varepsilon>0$, as $N\to\infty,$
\[
\Pi \ \Big(\sum_{r\notin S_{0,j}} \theta_{jr} > \varepsilon\ \Bigm \vert \bm{Y}\Big)
\;\longrightarrow\; 0
\quad\text{in $P_{\theta_0}$–probability}. 
\]
\end{theorem}
We prove \Cref{thm:modsel} in \switchref{\Cref{sec:proof_modifier}}{Section S1.5 of the Supplementary Materials}.
Essentially, the result guarantees that each active ensemble allocates only a vanishing fraction of splits to irrelevant modifiers; a Dirichlet–Beta tail bound then converts a few irrelevant split counts into small posterior mass on $(\sum_{r\notin S_{0,j}}\theta_{jr})$, yielding consistency. 
%This result also applies to the original VCBART model \citep{Deshpande2024}, which uses the identical hierarchical Dirichlet prior for modifier sparsity. The proof is detailed in the Supplementary Section S1.5.
% This forces the empirical fraction of irrelevant splits to concentrate near zero. 
\subsection{Concentration of local scales \texorpdfstring{\(\lambda_j\)}{lambda\_j}}
\label{sec:lambda-concentration}

% \Cref{thm:lambda-shrink} shows that the local scales $\lambda_j$ for truly null coefficients $(j \notin S_0)$ concentrate near zero at the optimal rate $r_N$. This concentration provides a practical screening device: covariates with posterior $\lambda_j$ below a vanishing threshold can be flagged as null with high probability.
\Cref{thm:lambda-shrink} demonstrates that the local scales $\lambda_j$ for truly null coefficients concentrate near zero at the optimal rate $r_N$, providing a practical screening device to identify inactive covariates.

To formalize how our model achieves covariate-level sparsity, we introduce one additional mild assumption. 
\begin{itemize}
\item[\textbf{(A7)}] For a deterministic sequence $\underline{L}_N \rightarrow \infty$, every inactive ensemble \(j\notin S_0\) satisfies (i) its total number of leaves \(L_j \ge \underline{L}_N\) and (ii) its induced partition of $[0,1]^R$ remains k-d regular.
\end{itemize}
% Assumption \textbf{(A7)} is a natural consequence of the BART prior structure. The depth-penalizing tree prior \textbf{(P1)} ensures that ensembles with no true signal to fit are unlikely to grow complex trees, satisfying (i). The k-d regularity of the design \textbf{(A2)} ensures the resulting partitions remain well-behaved, satisfying (ii). Under this condition, the following theorem confirms that the posterior for $\lambda_j$ for any null coefficient concentrates near zero, achieving the parametric variable selection rate.
Assumption \textbf{(A7)} ensures inactive ensembles are informative yet well-behaved: $L_j \rightarrow \infty$ guarantees enough splits so a Gamma–type tail bound in the $\lambda_j$ analysis yields exponential decay, forcing $\lambda_j$ to concentrate near zero.

\begin{theorem}\label{thm:lambda-shrink}
Suppose \textbf{(A1)}–\textbf{(A7)}, and \textbf{(P1)}–\textbf{(P3)} hold, and let \(j\notin S_0\).  Assume also \(\tau\) and \(c^2\) are bounded in probability away from \(0\) and \(\infty\) under the posterior.  Then, for a constant \(C\) independent of \(N\), and letting \(S_j=\sum_{\ell,m}(\mu^{(j)}_{\ell,m})^2\) and \(L_j\) be the leaf count,
\[
\Pi\left(\ \lambda_j \le C\,\sqrt{\frac{MS_j}{L_j}} \Bigm|\ \bm{Y}\right) \longrightarrow\ 1, \
\text{and} \
\frac{S_j}{L_j} = \mathcal O_{\Pi} \ \Big(r_N ^2\Big).
\]
\end{theorem}
Assumption \textbf{(A7)} also requires k-d regularity of the induced partition so that each terminal cell carries mass $\asymp 1/L_j;$ this prevents vanishing-mass leaves and gives the uniform leaf-mass bounds used to control $S_j/L_j$ via $\|\beta_j\|_N ^2,$ which is crucial for translating function-level contraction into shrinkage of $\lambda_j$.

We prove \Cref{thm:lambda-shrink} in \switchref{\Cref{sec:proof_local_scale}}{Section S1.6 of the Supplementary Materials}.
The result implies that, at least asymptotically, \texttt{sparseVCBART} can filter out the truly null signals.
To identify those covariates with non-null signals in practice, we recommend thresholding the posterior mean or median of the $\lambda_{j}$'s at some data-dependent level $t_{N} \asymp r_{N}.$
In our experiments, we observed a sharp ``elbow'' when plotting the posterior medians of the $\lambda_{j}$'s, with the scales for active covariates exceeding the scales of the noise covariates; see \switchref{\Cref{fig:real-lambdas}}{Figure S3.2 in the Supplementary Materials}.

\section{Numerical Experiments}
\label{sec:experiments}
We benchmark \texttt{sparseVCBART} against strong parametric, kernel, and tree-ensemble baselines: (1) the original VCBART model (\texttt{vanillaVCBART}); (2) the standard linear model (\texttt{LM}); (3) \citet{LiRacine2010}'s multivariate kernel-smoothing estimator (\texttt{KS}); (4) \citet{Burgin2017}'s method that approximates each $\beta_{j}(\cdot)$ with a single regression tree (\texttt{TVC}); and (5) \citet{ZhouHooker2022}'s boosted tree approach that approximates each coefficient with an ensemble of shallow trees (\texttt{BTVCM}).
Since these competitors are not designed for the high-dimensional regime studied in \Cref{sec:theory}, we focus on establishing \texttt{sparseVCBART}'s finite-sample performance in the $p,R  < N$ setting.
In \Cref{sec:synthetic-experiments} using synthetic data, we demonstrate that \texttt{sparseVCBART} correctly identifies null $\beta_{j}(\cdot)$'s when present but does not inappropriately introduce sparse structure when it is not present. 
We then demonstrate \texttt{sparseVCBART}'s ability to perform variable selection in a semi-synthetic data example in \Cref{sec:semisynth}

\subsection{Synthetic experiments}
\label{sec:synthetic-experiments}

We conduct two experiments with $R = 20$ effect modifiers, one with $p = 3$ covariates (Experiment 1) and one with $p = 50$ covariates (Experiment 2).
In both experiments, we generated 25 synthetic datasets with $N_{\text{train}} = 1000$ training observations and $N_{\text{test}} = 200$ testing observations from a varying coefficient model with (i) i.i.d. covariate vectors $\bx \sim \mvnormaldist{p}{\mathbf{0}_{p}}{\Sigma}$ with $\Sigma_{ij} = 0.5^{\lvert i - j \rvert}$; (ii) i.i.d. modifier vectors $\bz_{i}$ drawn uniformly from $[0,1]^{R}$; (iii) i.i.d. $\normaldist{0}{1}$ errors; and (iv) the following non-null covariate functions
\begin{align*}
\beta_{0}(\bz) &= 3z_{1} + (2 - 5 \ind{z_{2} > 0.5}\sin(\pi z_{1}) - 2 \ind{z_{2} > 0.5}, \\
\beta_{1}(\bz) &= (3 - 3z^{2}_{1}) \times \ind{z_{1} > 0.6} - 10\sqrt{z_{1}} \times \ind{z_{1} < 0.25}, \\
\beta_{2}(\bz) &= 1, \\
\beta_{3}(\bz) &= 10\sin(\pi z_{1}z_{2}) + 20(z_{3} - 0.5)^{2} + 10z_{4} + 5z_{5}.
\end{align*}
We ran each method with the default hyperparameters suggested by their respective authors and computed point estimates and pointwise 95\% uncertainty intervals for the testing set evaluations of the coefficient functions $\beta_{j}(\bz).$ 
For \texttt{KS}, \texttt{TVC}, and \texttt{BTVCM}, we formed these intervals using $50$ bootstrap re-samples.
We used $M_{j} = 50$ per ensemble for \texttt{sparseVCBART} and \texttt{vanillaVCBART} and computed posterior means and 95\% pointwise credible intervals after running four Markov chains for 2000 iterations each, discarding the first 400 samples of each as burn-in. 
All experiments were run on a shared university-based high-throughput computing cluster \citep{chtc_2006}. 

\Cref{fig:expfig} compares the average mean square error for evaluating $\beta_{j}(\bz)$ (panels (a) and (c)) and the average coverage of the pointwise 95\% uncertainty intervals (panels (b) and (d)).
Generally speaking, both BART-based approaches achieved smaller estimation error and higher uncertainty interval coverage than the non-Bayesian methods.
In Experiment 1 (top row \Cref{fig:expfig}) with $p = 3$, \texttt{sparseVCBART} and \texttt{vanillaVCBART} displayed virtually indistinguishable covariate effect recovery and uncertainty quantification.
These findings suggest that in the low-dimensional setting with no null coefficient functions, \texttt{sparseVCBART}'s shrinkage priors do not ``manufacture'' sparsity when it does not exist.

\begin{figure}[t]
  \centering
  \includegraphics[width=\linewidth]{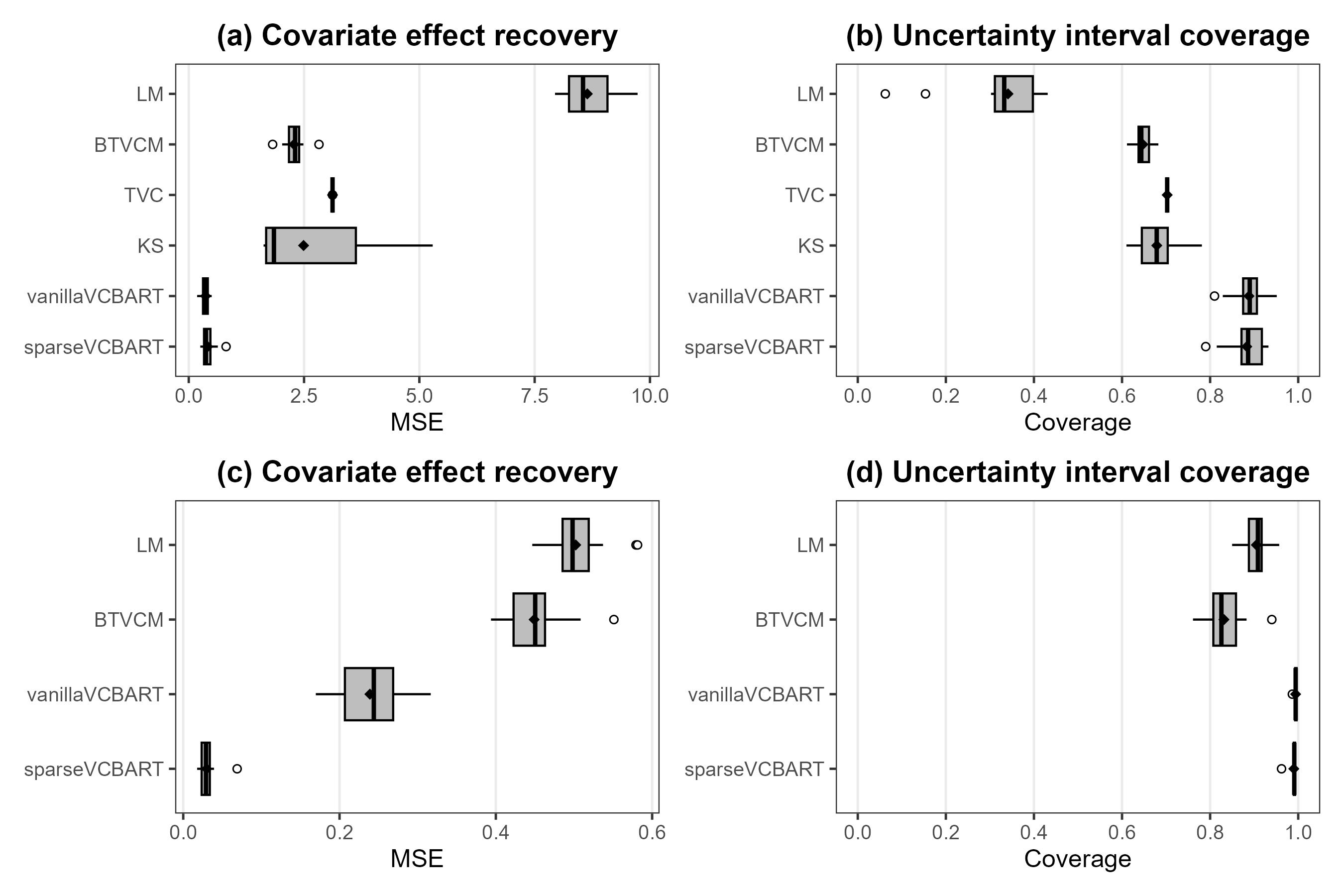}
  \caption{\textbf{Experiment 1 (top row) and Experiment 2 (bottom row)}.
  (a,c) Average MSE for evaluations $\beta_j(\bm z)$; (b,d) average $95\%$ coverage for $\beta_j(\bm z)$.}
  \label{fig:expfig}
\end{figure}

Unlike Experiment 1, in which all covariates had a non-null effect on the outcome, Experiment 2 involves $p = 50$ covariates of which all but the first three have null effects.
In Experiment 2, \texttt{sparseVCBART} achieved much smaller estimation error than \texttt{vanillaVCBART} thanks to the adaptive shrinkage provided by the global-local jump prior. 
Note that we omitted \texttt{KS} and \texttt{TVC} from the bottom row of \Cref{fig:expfig} because we were unable to compute bootstrap uncertainty intervals within our clusters' 72-hour time limit. 

Both \texttt{sparseVCBART} and \texttt{vanillaVCBART} were much faster than \texttt{KS} and \texttt{TVC} in both experiments.
Somewhat surprisingly, \texttt{sparseVCBART} was slightly faster than \texttt{vanillaVCBART} in both experiments.
Averaging across the 25 simulation replications, it took 10.53 seconds and 2.86 minutes to draw from the \texttt{sparseVCBART} posteriors in Experiments 1 and 2.
In contrast, drawing the same number of samples from the \texttt{vanillaVCBART} posterior took 14 seconds and 4.1 minutes, respectively.

\Cref{fig:funrec} shows function recovery from one representative replication of Experiment 2. 
We plot posterior means and 95\% bands for \texttt{sparseVCBART} (blue) and \texttt{vanillaVCBART} (grey) against the truth (solid black) for $(\beta_0), (\beta_1), (\beta_2)$, and one of the null functions \((\beta_4)\). 
For the active functions, \texttt{sparseVCBART} closely tracks \texttt{vanillaVCBART}, confirming that it does not degrade recovery when signals are present. 
On the null, \texttt{sparseVCBART}’s estimate and uncertainty intervals are tightly centered at zero, whereas \texttt{vanillaVCBART}'s estimate and uncertainty intervals ``wiggle” around zero, partially explaining \texttt{sparseVCBART}’s lower MSE in panel (c) of \Cref{fig:expfig}.

\begin{figure}[t]
  \centering
  \includegraphics[width=\linewidth]{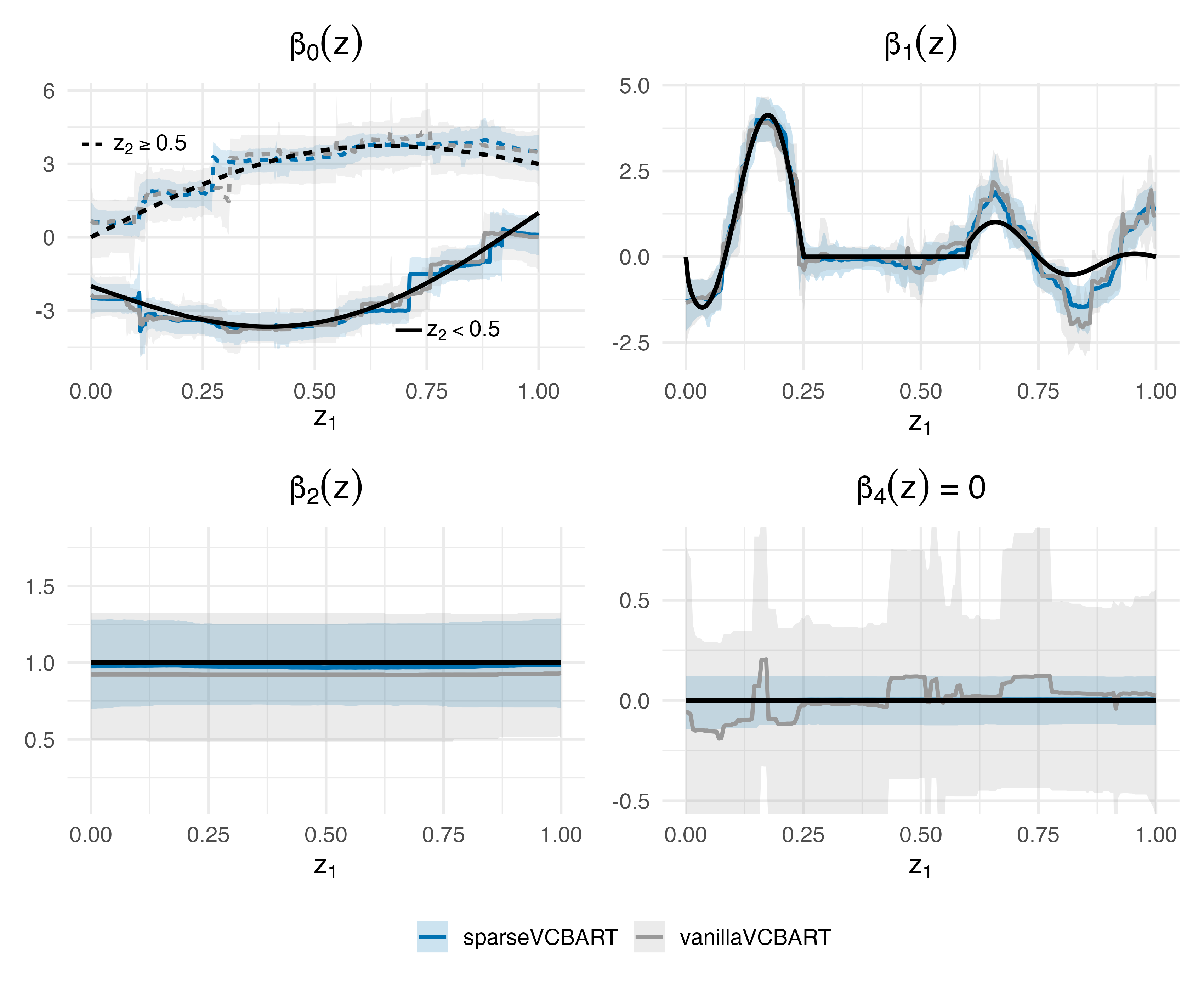}
  \caption{Function recovery in Experiment~2 for $\beta_0$, $\beta_1$, $\beta_2$, and a zero function $(\beta_4)$.}
  \label{fig:funrec}
\end{figure}

\subsection{A semi-synthetic experiment}
\label{sec:semisynth}

Interpersonal conversations can reduce exclusionary attitudes towards outgroups \citep{broockman_kalla16, williamson_etal21}.
\citet{kalla_broockman23} recently evaluated whether the type of conversation affects the reduction of prejudice. 
Their study involved canvassers having conversations with voters who were randomly assigned to one of the three intervention arms.
In the first arm, voters were told a story about an immigrant they knew before being asked to discuss a time when they needed support.
The second arm involved only the first part of the full intervention (i.e., the story telling), but not the second.
The third arm was a placebo in which the canvasser and the voter had a short conversation about unrelated topics.

We downloaded the raw data from \citet{kalla_broockman23}'s study from the Harvard Dataverse Network.
The data set contains $N = 4921$ observations and $R = 99$ effect modifiers that record voter demographics and current political attitudes.
The original data also contains $p = 2$ binary covariates that encode whether voters received the full or partial intervention arm.
Because this was a randomized trial, the coefficients associated with these covariates in a VCM capture the conditional average treatment effects of both interventions relative to placebo.

To assess \texttt{sparseVCBART}'s ability to identify non-null covariate effects, we augmented the data with 18 additional noise covariates generated from a standard Normal distribution.
We found a clear separation between the local scales corresponding to the two original covariates and the additional noise covariates.
Specifically, the posterior median $\lambda_{j}$'s for the intervention indicators exceed 3.5 while the posterior median scales for the null covariates were all less than 1.7; see \switchref{\Cref{fig:real-lambdas}}{Figure S3.2}.
Such separation is consistent with the conclusions of \Cref{thm:lambda-shrink} and suggest that both interventions had a non-null effect on the outcome.
Examining the corresponding vectors of the split probabilities $\btheta_{j}$, we further found that only a handful of effect modifiers drove the heterogeneity of each intervention's effect.
The most prominent were related to the voter's support for allowing undocumented immigrants to access public schools: the more supportive the voters were for inclusive schooling, the more effective both narrative interventions were in reducing prejudice; see \switchref{\Cref{sec:datadesc}}{Section S3.3 in the Supplementary Materials} for a more detailed overview of the data and analysis of effect heterogeneity.

\section{Discussion}
\label{sec:discussion}
We have introduced \texttt{sparseVCBART} for fitting sparse VCMs with potentially more covariates $p$ and/or effect modifiers $R$ than observations $N.$
Our main theoretical results show that, under mild conditions, (i) the \texttt{sparseVCBART} posterior contracts at nearly the minimax-optimal rate that adapts to the underlying sparsity structure and (ii) \texttt{sparseVCBART} consistently identifies which modifiers drive the non-null effects and shrinks null effects to zero. 
Empirically, the model delivers sparse, interpretable representations and strong predictive performance with coherent posterior uncertainty.

A key practical consideration is the choice of ensemble sizes $M_{j}.$ 
While both models generally benefit from larger ensembles, our preliminary sensitivity analyses (see \switchref{\Cref{sec:sensitivity}}{Section S3.2 of the Supplementary Materials}) suggest that \texttt{sparseVCBART} is more robust to a smaller ensemble size.
Beyond determining an optimal ensemble size, possible extensions include sharing variable splitting probabilities across coefficient ensembles, adapting the framework to generalized outcomes (e.g., logistic regression), and establishing formal covariate-selection consistency beyond the shrinkage results presented here.

%A very large $M$ can distribute a signal across too many small leaf values, making it harder for the prior to distinguish from noise and potentially leading to over-shrinkage.

\newpage
\bibliography{refs}

\newpage
\begin{center}
{\large 
\textbf{Supplementary Materials}
}
\end{center}
\appendix 

\renewcommand{\thesection}{S\arabic{section}}
\renewcommand{\thefigure}{\thesection.\arabic{figure}}
\renewcommand{\thetable}{\thesection.\arabic{table}}
\renewcommand{\theequation}{\thesection.\arabic{equation}}
\renewcommand{\thelemma}{\thesection.\arabic{lemma}}
\renewcommand{\thetheorem}{\thesection.\arabic{theorem}}

This supplement provides additional technical and empirical details to support the main paper. 
\Cref{sec:proofs} contains the complete mathematical proofs for our theoretical results. 
\Cref{sec:extraalgorithm} offers a detailed description of the MH-in-Gibbs sampler and its implementation. 
Finally, \Cref{sec:extraempirical} presents additional figures, tables, and analyses from our synthetic and semi-synthetic experiments.
An \textsf{R} package implemented \texttt{sparseVCBART} and code to replicate all experiments are available at \url{https://github.com/ghoshstats/sparseVCBART}.

\setcounter{section}{0}
\setcounter{figure}{0}
\setcounter{equation}{0}
\setcounter{table}{0}
\setcounter{theorem}{0}
\section{Proofs}
\label{sec:proofs}
The proof of Theorem 1 follows the general framework for posterior contraction rates established by \citet{Ghosal2000}, which requires verifying three main conditions. First, we establish sufficient prior concentration in a neighborhood of the true function (\Cref{cor:joint-smallball}). This is achieved by constructing an efficient dimension-adaptive tree approximant (\Cref{lem:approx-main}) and then showing that our priors collectively assign exponentially sufficient mass to this target (\Cref{lem:tree-mass,lem:smallball-main}). Second, we define a sieve of well-behaved functions (\Cref{def:sieve_revised}) and prove that the prior mass outside this sieve is negligible (\Cref{lem:sieve-mass_corrected}) while its metric entropy is appropriately bounded (\Cref{lem:entropy}). Finally, we construct uniformly exponentially consistent hypothesis tests (\Cref{lem:tests}) that can distinguish the truth from alternatives on the sieve. With these three conditions verified, the posterior contraction result follows.
\subsection{Dimension-adaptive H\"older approximation lemma}
Our proof of posterior contraction begins by establishing a non-asymptotic approximation bound. The following lemma shows that any true coefficient function satisfying our sparse H\"older assumption can be well-approximated by a tree ensemble whose complexity depends on the intrinsic modifier dimension $s_{0,j},$ not the ambient dimension $R$.
\begin{lemma}
\label{lem:approx-main}
Fix $j\in S_0$ with modifier set size $s_{0,j}$ and smoothness $0<\alpha_j\le 1$.
Assume $\beta_{0,j}(\bm{z})=\tilde\beta_{0,j}(\bm{z}_{S_{0,j}})$ with
$\tilde\beta_{0,j}\in \mathcal H^{\alpha_j}_B([0,1]^{s_{0,j}})$.
Then there exists a tree ensemble for the $j$th coefficient using
\[
K_j \;=\; \Big\lceil N^{\frac{s_{0,j}}{2\alpha_j+s_{0,j}}}\Big\rceil
\ =\ \mathcal O \Big(\tfrac{N r_{jN,\mathrm{ad}}^2}{\log N}\Big)
\]
terminal leaves (splitting only on $S_{0,j}$) and an associated step function $g_j$ such that
\[
\|\beta_{0,j}-g_j\|_N \;\le\; \tfrac{1}{4}\,r_{jN,\mathrm{ad}}
\qquad\text{for all $N$ large enough.}
\]
\end{lemma}
\begin{proof}
The proof establishes the existence of a suitable tree-based approximator by adapting the logic from \citet[Lemma 3.2]{Rockova2019} for k-d trees to the varying-coefficient setting. 
Let $s:=s_{0,j}$ and write the active coordinates as an index set $S_{0,j}$.
Consider the unit cube $[0,1]^s$ and construct a balanced $s$–dimensional k–d tree
by recursive empirical median splits cycling through the $s$ coordinates until $K_j$ terminal
cells are created. Denote this partition by $\{A_k\}_{k=1}^{K_j}$, $A_k\subset[0,1]^s$.
Lift it to a partition of $[0,1]^R$ by taking
\[
\Omega_k \;:=\; A_k \times [0,1]^{R-s}\subset[0,1]^R,
\qquad k=1,\dots,K_j.
\]
By construction, the partition $\{\Omega_k\}$ only splits along coordinates in $S_{0,j}$.

Let $\mathrm{diam}_\infty(\cdot)$ denote the $\ell_\infty$–diameter. A balanced k-d tree on $[0,1]^s$ splits axis-aligned; each cell $A_k$ is a rectangle $A_k = \prod_{\ell=1}^{s} I_{k,\ell},$ with $|I_{k,\ell}| =: \ell_{k,\ell}$, so $\text{diam}_{\infty}(A_k) = \max_{\ell \le s} \ell_{k,\ell}$.
For such a tree on $[0,1]^s$ with $K_j$ leaves, each coordinate is split
$\lfloor m/s\rfloor$ or $\lceil m/s\rceil$ times when $K_j=2^m$ \footnote{If $K_j$ is not a power of two, pick $m$ with $2^{m-1} < K_j \le 2^m.$ A balanced k-d tree with $2^m$ leaves has side length $\le 2^{-\lfloor m/s \rfloor}$ per active coordinate; merging adjacent cells to reach $K_j$ leaves increases diameters by at most a factor $\le 2,$ hence $\max_{k} \text{diam}_{\infty} (A_k) \le C K_j^{-1/s}.$};
hence $\ell_{k,\ell} \le 2^{-\lfloor m/s\rfloor} \implies \text{diam}_{\infty}(A_k) \le 2 \cdot 2^{-m/s} = 2K_j^{-1/s}.$
Therefore,
\begin{equation}\label{eq:intrinsic-diam}
\max_{1\le k\le K_j}\mathrm{diam}_\infty(A_k)
\;\le\; 2\,K_j^{-1/s},
\end{equation}
and, since $\beta_{0,j}$ is constant in the inactive coordinates,
\begin{equation}\label{eq:intrinsic-diam-extended}
\sup_{\bm{u},\bm{v}\in \Omega_k}\big|\beta_{0,j}(\bm{u})-\beta_{0,j}(\bm{v})\big|
~=~\sup_{\bm{u'},\bm{v'}\in A_k}\big|\tilde\beta_{0,j}(\bm{u'})-\tilde\beta_{0,j}(\bm{v'})\big| .
\end{equation}

Define $g_j$ to be constant on each $\Omega_k$, equal to the Lebesgue average
of $\beta_{0,j}$ on $\Omega_k$:
\[
g_j(\bm{z}) \;:=\; \frac{1}{\mathrm{vol}(\Omega_k)}\int_{\Omega_k}\beta_{0,j}(\bm{u})\,d\bm{u},
\quad \bm{z}\in\Omega_k.
\]
Because $\beta_{0,j}(\bm{u})=\tilde\beta_{0,j}(\bm{u}_{S_{0,j}})$ and $\Omega_k=A_k\times[0,1]^{R-s}$,
this average reduces to the average of $\tilde\beta_{0,j}$ over $A_k$.
For any cell $\Omega_k$ and any $\bm{z}\in\Omega_k$,
\[
\big|\beta_{0,j}(\bm{z})-g_j(\bm{z})\big|
~\le~
\sup_{\bm{u},\bm{v}\in \Omega_k}\big|\beta_{0,j}(\bm{u})-\beta_{0,j}(\bm{v})\big|
~=~ \sup_{\bm{u'},\bm{v'}\in A_k}\big|\tilde\beta_{0,j}(\bm{u'})-\tilde\beta_{0,j}(\bm{v'})\big|,
\]
where the equality is \eqref{eq:intrinsic-diam-extended}. Since
$\tilde\beta_{0,j}\in \mathcal{H}_B^{\alpha_j}([0,1]^s)$, we have the Hölder bound
$|\tilde\beta_{0,j}(\bm{u}')-\tilde\beta_{0,j}(\bm{v}')|\le B\,\|\bm{u}'-\bm{v}'\|_\infty^{\alpha_j}$.
Using \eqref{eq:intrinsic-diam},
\[
\|\beta_{0,j}-g_j\|_\infty
\;\le\; B\Big(\max_k \mathrm{diam}_\infty(A_k)\Big)^{\alpha_j}
\;\le\; C_2\,K_j^{-\alpha_j/s}.
\]
Therefore the empirical $L_2$–error also satisfies
\begin{equation}\label{eq:empL2-bound}
\|\beta_{0,j}-g_j\|_N \;\le\; \|\beta_{0,j}-g_j\|_\infty
\;\le\; C_2\,K_j^{-\alpha_j/s}.
\end{equation}
Take
\(
K_j = \Big\lceil N^{\,s/(2\alpha_j+s)} \Big\rceil
,\) 
then $K_j^{-\alpha_j/s}\le C_3\,N^{-\alpha_j/(2\alpha_j+s)}$ for a constant $C_3$ ($\because \lceil x\rceil \ge x$ and the map $t \mapsto t^{-\alpha_j/s}$ is monotonically decreasing on $t>0$),
and by \eqref{eq:empL2-bound},
\[
\|\beta_{0,j}-g_j\|_N \;\le\; C_4\,N^{-\alpha_j/(2\alpha_j+s)}.
\]
Since $r_{jN,\mathrm{ad}}=(\log N)^{1/2}N^{-\alpha_j/(2\alpha_j+s)}$, we have for $C_4 \le \frac{1}{4} \sqrt{\log N}$, $C_4\,N^{-\alpha_j/(2\alpha_j+s)} \le \tfrac14\,r_{jN,\mathrm{ad}}$ for all $N\ge N_0 \coloneqq \exp((4C_4)^2)$, because $(\log N)^{1/2}\to\infty$.
Finally,
\[
\frac{N\,r_{jN,\mathrm{ad}}^2}{\log N}
~=~ N^{\,1-2\alpha_j/(2\alpha_j+s)}
~=~ N^{\,s/(2\alpha_j+s)}
~\asymp~ K_j,
\]
so $K_j=\mathcal{O}(N r_{jN,\mathrm{ad}}^2/\log N)$.

Similar to the proof of Theorem 2 in \citet{jeong2023artbartminimaxoptimality}, we assign all $K_j$ terminal leaves to a single tree in the ensemble $E_j^\star$
and let the remaining $M-1$ trees be stumps. The total number
of leaves is $K_j$, and $g_j$ is realized as a step function represented by $E_j^\star$.
This construction uses splits only on coordinates in $S_{0,j}$. Combining the steps gives the claim.
\end{proof}

\subsection{Dimension-adaptive small-ball probability}
Having established the existence of an efficient approximation, we now verify the second key condition for posterior contraction: the prior must assign sufficient mass to a “small ball’’ around the true function. The following result confirms that our combined priors on tree structure, modifier selection, and leaf values concentrate enough mass near the sparse target.

\begin{lemma}
\label{lem:smallball-main}
Let $j\in S_0$ and define
\[
B_j(r)\ :=\ \{\beta_j:\ \|\beta_j-\beta_{0,j}\|_N\le r\},
\qquad r=r_{jN,\mathrm{ad}}.
\]
Under \textbf{(A1)}–\textbf{(A5)} and \textbf{(P1)}–\textbf{(P3)}, there exist constants $C_1,N_0>0$ (independent of $N,R,p_N$) such that
\[
\Pi\big(B_j(r_{jN,\mathrm{ad}})\big)
\;\ge\;
\exp\!\big\{-C_1\,N\,r_{jN,\mathrm{ad}}^2\big\},
\quad \forall\,N\ge N_0.
\]
\end{lemma}

We first control the prior mass of the specific $K_j$–leaf tree realizing the approximation from \Cref{lem:approx-main}.

\begin{lemma}[Prior mass of the dimension–adaptive target tree]
\label{lem:tree-mass}
Fix $j\in S_0$ and write $s:=s_{0,j}$. Let $r_{jN,\mathrm{ad}}=(\log N)^{1/2}N^{-\alpha_j/(2\alpha_j+s)}$ and let
$E_j^\star$ be the ensemble from \Cref{lem:approx-main} that realizes a step function $g_j$
using a single $K_j$–leaf tree that splits only on coordinates in $S_{0,j}$, with
\(
K_j=\Big\lceil N^{\,s/(2\alpha_j+s)}\Big\rceil=\mathcal{O} \left({N\,r_{jN,\mathrm{ad}}^2}/{\log N}\right).
\)
Assume {\textbf{(P1)}}’s split proposal chooses a cutpoint \emph{uniformly} from the set of admissible empirical midpoints within a node. Under {\textbf{(A2)}} and {\textbf{(P1)}}, there exists $c>0$ (independent of $N,R,p_N$) such that
\[
  \Pi \ \!\bigl(\text{topology+split-variables+cutpoints of }E^{\!*}_{j}\bigr)
  \;\ge\;
  \exp \bigl\{-c\,K_j\log N\bigr\}
  \;=\;
  \exp\bigl\{-c'\,N r_{jN,\mathrm{ad}}^{2}\bigr\}.
\]
\end{lemma}

\begin{proof}
Let $m_j=K_j-1$ be the number of internal nodes. We factor the mass into three pieces:

\medskip\noindent\emph{(i) Topology.}
Under the Galton–Watson (GW) depth–decay prior of {\textbf{(P1)}} with $\mathbb P(\text{split at depth }d)=\gamma^d$ and $N^{-1}<\gamma<1/2$, the GW tail bound of \citet{rockovasaha} (Sec.\ 5) yields, for any fixed $K$–leaf full binary topology $\mathcal T$ with the set of internal nodes being $\mathcal{I}$, and the set of leaves being $\mathcal{L}$, 
\begin{equation}\label{eq:topology-mass-fixed}
\begin{split}
\Pi(\text{topology =}\mathcal{T}) & = \prod_{v \in \mathcal{I}} \gamma^{\text{depth}(v)} \prod_{u \in \mathcal{L}} (1-\gamma^{\text{depth}(u)}) \\
& \ge (1-\gamma)^{K} \gamma^{\sum_{v \in \mathcal{I}}\text{depth}(v)} \\
& = (1-\gamma)^{K} \gamma^{\Theta(K \log K)} \\
& \ge \exp(-c_1 K \log K)
\end{split}
\end{equation}
With $K=K_j \asymp N^{s/(2\alpha_j+s)} \le N,$ we have $\log K \le \log N$ and thus $\exp(-c_1 K \log K) \ge \exp(-c_1 K \log N).$ This proves that $\Pi(\text{topology of} \ E_j ^\star) \ge \exp\{-c_1K_j \log N\}.$

\medskip\noindent\emph{(ii) Split variables restricted to $S_{0,j}$.}
Let $\bm{\theta}_j\sim\mathrm{Dir}(\eta_j/R,\ldots,\eta_j/R)$ with $\eta_j\sim\pi(\eta)\propto (R+\eta)^{-(R+1)}$. As also noted in \citet{Deshpande2024}, we define $V:=\sum_{u\in S_{0,j}}\theta_{ju}$. Conditional on $\bm{\theta}_j$, the probability that all $m_j$ split variables lie in $S_{0,j}$ equals $V^{m_j}$. This can be explained by observing the fact that conditional on $\bm{\theta}_j = (\theta_{j1},\dots,\theta_{jR}),$ the split variable at node $\ell$ is a categorical draw $U_{\ell} \vert \bm{\theta}_j \sim \text{Cat}(\theta_{j1},\dots,\theta_{jR}),$ independently across nodes. Then for any node $\ell$, $\mathbb{P}(U_{\ell} \in S_{0,j} \vert \bm{\theta}_j) = \sum_{u \in S_{0,j}} \theta_{ju} = V.$ By the conditional independence across the $m_j$ nodes, we have $\mathbb{P}(U_1 \in S_{0,j},\dots,U_{m_j} \in S_{0,j} \vert \bm{\theta}_j) = \prod_{\ell=1}^{m_j} \mathbb{P}(U_{\ell} \in S_{0,j} \vert \bm{\theta}_j) = V^{m_j}$ which is the claimed equality.
Thus taking expectation w.r.t. $\bm{\theta}_j$ gives 
\[
\Pi(\text{all $m_j$ split variables in }S_{0,j}) \ =\ \mathbb E[V^{m_j}].
\]
For each fixed $\eta_j$, $V\sim\mathrm{Beta}(a_S,a_0-a_S)$ with $a_S=s\,\eta_j/R$ and $a_0=\eta_j$, hence by Jensen (since $x\mapsto x^{m_j}$ is convex on $[0,1]$ for $m_j\ge1$),
\[
\mathbb E(V^{m_j}\mid\eta_j)\ \ge\ \big(\mathbb E(V\mid\eta_j)\big)^{m_j} \ =\ \Big(\frac{s}{R}\Big)^{m_j}.
\]
Taking expectation over $\eta_j$ preserves the bound:
\begin{equation}\label{eq:splitvar-mass-fixed}
\begin{split}
\Pi(\text{all split variables in }S_{0,j}) & = \mathbb{E}_{\eta_j}\mathbb{E}(V^{m_j} \vert \eta_j) \\
&\ge\ \Big({s}/{R}\Big)^{m_j} \\
& = \exp \{m_j \log(s/R)\}\\
&\ge\ \exp\{-c_2\,K_j\log(R/s)\} \quad (\because m_j =K_j-1 \asymp K_j) 
\end{split}
\end{equation}
Since {\textbf{(A5)}} assumes $R=o(\log N)$, the factor $K_j\log(R/s)$ is $o\big(K_j\log N\big)$ and will be absorbed by the topology/cutpoint terms.

\medskip\noindent\emph{(iii) Cutpoints.}
By the stated cutpoint rule,\footnote{Conditional on a chosen split variable and the set $\{Z_{iu}: \ i \in \text{node}\}$, we draw the cutpoint uniformly from the set of admissible empirical midpoints $\{(z_{(r)}+z_{(r+1)})/2\}$ between consecutive order statistics that yield non-empty children.} at each internal node the cutpoint is drawn uniformly from a finite grid of admissible empirical midpoints (at most $N$ options). Therefore, the probability to hit the \emph{exact} midpoint that realizes the balanced $k$–d split for $g_j$ is at least $N^{-1}$ per internal node, independently of other nodes, conditional on the topology and split variables (following \textbf{(P1)}). Hence
\begin{equation}\label{eq:cut-mass-fixed}
\Pi(\text{all cutpoints equal the target midpoints})\ \ge\ N^{-m_j}\ \ge\ \exp\{-c_3 K_j\log N\}.
\end{equation}

Multiplying \eqref{eq:topology-mass-fixed}, \eqref{eq:splitvar-mass-fixed}, and \eqref{eq:cut-mass-fixed}, we obtain
\[
\Pi(\text{topology+split-variables+cutpoints of }E_j^\star)
\ \ge\ \exp\{-c\,K_j\log N\},
\]
and since $K_j\asymp N r_{jN,\mathrm{ad}}^2/\log N$, this equals $\exp\{-c' N r_{jN,\mathrm{ad}}^2\}$ for some $c,c'>0$.
\end{proof}

\paragraph{Leaf-variance calibration event.}
In ordinary BART, each leaf mean is $\mu_{\ell,m} \sim \mathcal{N}(0,\sigma_{\mu} ^2/M)$ with a fixed variance scale $\sigma_{\mu} ^2$. That makes ``small-ball" lower bounds straightforward: the mass of a rectangle of width $\delta$ around a target step value is $\gtrsim (c\delta)^K$ with a constant $c>0$ that does not vanish with $N$. Here, under the regularized horseshoe leaf prior in {\textbf{(P2)}}, each leaf value satisfies
\[
\mu^{(j)}_{\ell,m} \sim \mathcal N\Big(0,s_j^2\Big),
\quad
s_j^2 =\frac{1}{M}\,\frac{\tau^2\lambda_j^2\,c^2}{c^2+\tau^2\lambda_j^2},
\]
so $s_j ^2$ can be arbitrarily small if $\tau$ or $\lambda_j$ is too small. If we do not control $s_j ^2$, the Gaussian neighborhood around the target step values might carry negligible mass. The calibration event $\mathcal{E}_{\textrm{sc}}$ (defined below) pins $s_j ^2$ to a fixed interval $[a/M,b/M]$ (with $0<a<b<c_L ^2),$ ensuring a uniform per-leaf small-ball bound \[
\inf_{|x| \le B} \mathbb{P} \left(| \mu_{\ell,m}^{(j)} - x| \le \delta \mid \mathcal{E}_{\textrm{sc}} \right) \ge c \delta
\]
with $c>0$. That is the exact analogue of fixed-variance BART and is what allows us to get $\exp \{-CNr_{jN,\textrm{ad}}^2\}$ after multiplying over $K_j$ leaves.

Recall that $\lambda_j\sim\mathcal C^+(0,1)$, $\tau\sim\mathcal C^+(0,A_{g,d})$, and \emph{slab} $c^2\sim\mathcal{IG}(\nu_c/2,\nu_c s_c^2/2)$. We fix constants $0<a<b<c_L^2$ and $0<c_L<c_U<\infty$, and set
\(
a_\sigma:=a\,A_{g,d_\star},\quad b_\sigma:=b\,A_{g,d_\star}.
\)
We define the calibration event
\[
\mathcal E_{\mathrm{sc}}
:=\ \Bigl\{a_\sigma\le \tau \le b_\sigma\Bigr\}
 \ \cap\ \Bigl\{c^2\in[c_L^2,c_U^2]\Bigr\}
 \ \cap\ \Bigl\{\lambda_j\in[\Lambda_L(\tau,c),\Lambda_U(\tau,c)]\Bigr\}
 \ \cap\ \{d=d_\star\},
\]
where (inverting the map $\lambda\mapsto s_j^2$)
\[
\Lambda_t(\tau,c)
=\sqrt{\frac{t\,c^2}{\tau^2(c^2-t)}}\,,\qquad t\in\{a,b\}.
\]
On $\mathcal E_{\mathrm{sc}}$ we have $s_j^2\in[a/M,b/M]$, i.e. the leaf variance is bounded between two positive constants independent of $(N,R,p_N)$.

\medskip\noindent\emph{Mass of $\mathcal E_{\mathrm{sc}}$.}
Exploiting the scale invariance of half-Cauchy, we can write $\tau=A_{g,d_{*}} X$ with $X \sim \mathcal{C}^{+}(0,1).$ Then, 
\[
\mathbb{P}(a_{\sigma} \le \tau \le b_{\sigma} \vert d_{*}) = \mathbb{P}(a \le X \le b) = 2\pi^{-1}(\arctan(b)-\arctan(a)) =: \kappa_{\tau} >0,
\]
a positive constant independent of $(N,R,p_N)$ and of $A_{g,d_*}.$ Again, with $c^2 \sim \mathcal{IG}(\nu_c /2,\nu_c s_c ^2/2),$ any fixed interval $[c_L ^2, c_U ^2] \subset (0,\infty)$ has strictly positive probability $\mathbb{P}(c^2 \in [c_L ^2, c_U ^2]) =: \kappa_c >0.$ Since $\lambda_j \sim \mathcal{C}^{+}(0,1),$ for any $0 \le x<y,$ 
\begin{equation*}
    \begin{split}
        \mathbb{P}(x \le \lambda_j \le y) & = 2\pi^{-1} (\arctan(y)-\arctan(x)) \\
        & \ge 2\pi^{-1} \frac{y-x}{1+y^2} \quad (\text{by the mean-value theorem}).
    \end{split}
\end{equation*}
Imposing $x=\Lambda_{L}(\tau,c)$ and $y=\Lambda_{U}(\tau,c)$, we have \[
\Lambda_{U}(\tau,c)-\Lambda_{L}(\tau,c) = \tau^{-1} \left[\sqrt{\frac{bc^2}{c^2-b}}-\sqrt{\frac{ac^2}{c^2-a}}\right]=\tau^{-1} \kappa_{\textrm{gap}}(c)
\]
where $\kappa_{\textrm{gap}}(c) \coloneqq \sqrt{bc^2/(c^2-b)}-\sqrt{ac^2/(c^2-a)} \ge \underline{\kappa}_{\textrm{gap}} > 0$. Because $c \in [c_L,c_U]$ lives in a compact set and $a,b$ are fixed with $0<a<b<c_L^2$. Likewise, 
\[
\Lambda_{U}(\tau,c) = \tau^{-1}\sqrt{\frac{bc^2}{c^2-b}} \le \tau^{-1} \bar{\kappa},
\]
where $\bar{\kappa} \coloneqq \max_{c \in [c_L,c_U]} \sqrt{bc^2/(c^2-b)} < \infty.$ Therefore, \begin{align*}
    \mathbb{P}(\lambda_j \in [\Lambda_L,\Lambda_U] \mid \tau,c) & \ge 2 \pi^{-1} \frac{\Lambda_{U}-\Lambda_{L}}{1+\Lambda_U ^2} \\
    & \ge 2 \pi^{-1} \frac{\underline{\kappa}_{\textrm{gap}}/\tau}{1+ (\bar{\kappa}/\tau)^2} \\
    & = 2 \pi^{-1} \frac{\underline{\kappa}_{\textrm{gap}} \tau}{\tau^2 + \bar{\kappa}^2}.
\end{align*}

Now, observe that on $\{a_\sigma \le \tau \le b_{\sigma}\}$ with $A_{g,d_{*}} \rightarrow 0,$ we have $0 < \tau \le b_{\sigma} \le \bar{\kappa}$ eventually, so for all large $N$, \[
\frac{\tau}{\tau^2+\bar{\kappa}^2} \ge \frac{\tau}{2 \bar{\kappa}^2} \implies \mathbb{P}[\lambda_j \in [\Lambda_L,\Lambda_U] \mid \tau,c] \ge \underbrace{\pi^{-1} \bar{\kappa}^{-2} \underline{\kappa}_{\textrm{gap}}}_{=:\kappa_{\lambda}} \tau. 
\]
Hence, uniformly over $c \in [c_L,c_U]$ and $\tau \in [a_{\sigma},b_{\sigma}],$ \[
\mathbb{P}(\lambda_j \in [\Lambda_L,\Lambda_U] \mid \tau,c) \ge \kappa_{\lambda} \tau \ge \kappa_{\lambda} a_{\sigma}.
\]

With $\pi(d) \propto 1/d!$ on $\{0,1,\dots,p_N\}$ the normalizer is $Z_{p_N} = \sum_{d=0}^{p_N} 1/d! \in (e-1,e) \ \forall p_N \ge 1.$ Thus, 
\[
\pi(d_{*}) = \frac{1/d!}{Z_{p_N}} \ge \frac{1}{e d_{*} !} =: \kappa_{d} (d_{*}).
\]

By independence across the prior blocks, 
\begin{equation}\label{eq:RHS-Esc-mass-fixed}
    \begin{split}
    \Pi(\mathcal{E}_{\textrm{sc}}) & \ge \pi(d_{*}) \cdot \mathbb{P}(a_{\sigma} \le \tau \le b_{\sigma} \mid d_{*}) \cdot \mathbb{P}(c^2 \in [c_L ^2, c_U ^2]) \cdot \inf_{(\tau,c) \in [a_{\sigma,b_{\sigma}}] \times [c_L,c_U]} \mathbb{P}(\lambda_j \in [\Lambda_L,\Lambda_{U}] \mid \tau,c) \\
     & \ge \kappa_d (d_{*}) \kappa_c \kappa_{\tau} \kappa_{\lambda} a_{\sigma} \\
     & \gtrsim \kappa_{d}(d_{*}) \cdot \frac{d_{*}}{p_N(\log p_N)^2 (\log N)^2} 
    \end{split}
\end{equation}
% For a half–Cauchy with scale $A_{g,d_\star}$, the interval $[a_\sigma,b_\sigma]=[aA_{g,d_\star},bA_{g,d_\star}]$ has probability $\kappa_\tau>0$ independent of $N$; similarly $\Pi(c^2\in[c_L^2,c_U^2])=\kappa_c>0$, and $\pi(d_\star)>0$ under the prior on $d$. Conditional on $(\tau,c)$, the half–Cauchy mass of $[\Lambda_L,\Lambda_U]$ satisfies
% \[
% \mathbb P\big(\lambda_j\in[\Lambda_L,\Lambda_U]\mid \tau,c\big)
% =\frac{2}{\pi}\big(\arctan\Lambda_U-\arctan\Lambda_L\big)
% \ \ge\ \frac{2}{\pi}\,\frac{\Lambda_U-\Lambda_L}{1+\Lambda_U^2}.
% \]
% Since $\Lambda_U-\Lambda_L\asymp 1/\tau$ and $\Lambda_U\lesssim 1/\tau$ when $c\in[c_L,c_U]$ and $t\in\{a,b\}$, the right-hand side is $\ge \kappa_\lambda\,\tau$ for some $\kappa_\lambda>0$. With $\tau\in[a_\sigma,b_\sigma]$ we obtain
% \begin{equation}\label{eq:RHS-Esc-mass-fixed}
% \Pi(\mathcal E_{\mathrm{sc}})\ \ge\ \kappa_d\,\kappa_c\,\kappa_\tau\,\kappa_\lambda\,a_\sigma
% \ \asymp\ \frac{d_\star}{p_N}.
% \end{equation}

\begin{proof}[Proof of \Cref{lem:smallball-main}]
By \Cref{lem:approx-main}, there exists a step function $g_j$ realized by a tree splitting only on $S_{0,j}$ with
\begin{equation}\label{eq:approx-error-fixed}
\|\beta_{0,j}-g_j\|_N \ \le\ \frac{1}{4}\,r_{jN,\mathrm{ad}},
\qquad
K_j\ \asymp\ \frac{N\,r_{jN,\mathrm{ad}}^2}{\log N}.
\end{equation}
By \Cref{lem:tree-mass},
\begin{equation}\label{eq:tree-splits-mass-fixed}
\Pi\bigl(\text{topology+split-variables+cutpoints realizing }g_j\bigr)
\ \ge\ \exp\{-c_T\,N r_{jN,\mathrm{ad}}^2\}.
\end{equation}

On $\mathcal E_{\mathrm{sc}}$ we have $s_j^2\in[\underline v,\overline v]$ with $\underline v:=a/M$, $\overline v:=b/M$. Label the $K_j$ terminal cells by $k=1,\dots,K_j$ and let $g_{jk}$ be the value of $g_j$ on cell $k$. Because $\tilde\beta_{0,j}\in\mathcal H_B^{\alpha_j}$, we have $\|g_j\|_\infty\le B$ for some $B>0$.\footnote{By definition, $\|\tilde{\beta}_{0,j}\|_{\infty} \le B$. If we write $g_j$ on each cell $A_k$ as the empirical average  $v_k \coloneqq |A_k|^{-1} \sum_{i:\bm{Z}_i \in A_k} \tilde{\beta}_{0,j}(\bm{Z}_i),$ then $|v_k| \le \|\tilde{\beta}_{0,j}\|_{\infty} \le B.$ Hence $\|g_j\|_{\infty} = \max_{k} |v_k| \le B.$} For each $k$,
\[
\mathbb P\big(|\mu_{jk}-g_{jk}|\le \delta\ \big|\ \mathcal E_{\mathrm{sc}}\big)
\ \ge\ 2\delta\cdot \inf_{|x|\le B}\frac{1}{\sqrt{2\pi}\,\overline v^{1/2}}\exp\!\Big(-\frac{x^2}{2\underline v}\Big)
\ :=\ c_6\,\delta,
\]
with $c_6>0$ depending only on $(\underline v,\overline v,B)$.
Take \(
\delta_N :=\frac{r_{jN,\mathrm{ad}}}{4\sqrt{K_j}}.
\) Conditional on $\mathcal E_{\mathrm{sc}}$ and the event in \eqref{eq:tree-splits-mass-fixed}, independence of the leaf means across cells gives
\begin{align}
\mathbb P\Big(\cap_{k=1}^{K_j} A_k \mid \mathcal E_{\mathrm{sc}}\Big) & = \prod_{k=1}^{K_j}\int_{g_{jk}-\delta_{N}}^{g_{jk}+\delta_N} \frac{1}{\sqrt{2\pi} s_j}\exp(-t^2/(2s_j ^2)) \ dt \ge \prod_{k=1}^{K_j} \delta_N \cdot \underbrace{2\inf_{|x| \le B}\frac{1}{\sqrt{2\pi} \bar{v}^{1/2}}\exp(-x^2/(2\underline{v}))}_{=:c_6 >0} \nonumber\\
&\ge (c_6\,\delta_N)^{K_j}
= \exp\!\Big\{ -K_j \log\!\frac{1}{c_6\delta_N}\Big\}\nonumber\\
&= \exp\!\Big\{ -K_j\Big(\tfrac{1}{2}\log K_j + \log\!\tfrac{4}{c_6 r_{jN,\mathrm{ad}}}\Big)\Big\}.\label{eq:leaf-close-fixed}
\end{align}
where we let $A_k \coloneqq \{|\mu_{jk}-g_{jk}| \le \delta_N\}$ and the event $\cap_{k=1}^{K_j} A_k$ depicts $\{ \max_{1 \le k \le K_j} |\mu_{jk}-g_{jk}| \le \delta_{N} \}.$ Using $K_j\asymp N r_{jN,\mathrm{ad}}^2/\log N$ and $\log K_j\asymp \log N$, the exponent in \eqref{eq:leaf-close-fixed} is bounded by $-c_L\,N r_{jN,\mathrm{ad}}^2$ for some $c_L>0$ and all large $N$, hence
\begin{equation}\label{eq:leaf-product-final}
\mathbb P\Big(\max_{k}|\mu_{jk}-g_{jk}|\le \delta_N\ \Big|\ \mathcal E_{\mathrm{sc}}\Big)
\ \ge\ \exp\{-c_L\,N r_{jN,\mathrm{ad}}^2\}.
\end{equation}
By construction of $\delta_N$ and \eqref{eq:approx-error-fixed}, the step function $\tilde g_j$ with leaves $\{\mu_{jk}\}$ satisfies
\[
\|\tilde g_j-\beta_{0,j}\|_N\ \le\ \|\tilde g_j-g_j\|_N + \|g_j-\beta_{0,j}\|_N
\ \le\ \delta_N\sqrt{K_j} \footnote{This can be seen by writing $w_k$ as the empirical fraction of sample points in cell $k$ ($0 \le w_k \le 1$). Then we have a loose upper bound as $\|\tilde{g}_j - g_j \|_N ^2 = \sum_{k=1}^{K_j} w_k (\mu_{jk}-g_{jk})^2 \le \sum_{k=1}^{K_j} w_k \delta_N ^2 \le K_j \delta_N ^2$.} + \frac{1}{4}r_{jN,\mathrm{ad}}
\ =\ \frac{1}{2}\,r_{jN,\mathrm{ad}}.
\]

Multiplying \eqref{eq:tree-splits-mass-fixed}, \eqref{eq:RHS-Esc-mass-fixed}, and \eqref{eq:leaf-product-final} gives
\begin{align*}
\Pi\bigl(\|\beta_j-\beta_{0,j}\|_N\le r_{jN,\mathrm{ad}}\bigr)
&\ge
\exp\!\big\{-c_T\,N r_{jN,\mathrm{ad}}^2\big\}\cdot
\Pi(\mathcal E_{\mathrm{sc}})\cdot
\exp\!\big\{-c_L\,N r_{jN,\mathrm{ad}}^2\big\}\\
&\ge
\exp\!\big\{-C\,N r_{jN,\mathrm{ad}}^2\big\}\cdot \frac{d_\star}{p_N (\log p_N)^2 (\log N)^2}.
\end{align*}
By the growth condition in {\textbf{(A5)}} we have $\log p_N=o(N r_{jN,\mathrm{ad}}^2)$, so the polynomial factor can be absorbed into the exponential by enlarging $C$, which can be seen by \[
\log \frac{d_{*}}{p_N(\log p_N)^2 (\log N)^2} = \log d_{*} -\log p_N - 2 \log \log p_N - 2\log \log N,
\]
with each term being $o(Nr_{jN,\textrm{ad}}^2)$ and because $\log \log p_N, \log \log N \ll \log p_N$, implying that for any $\epsilon>0,$ $d_{*}p_N^{-1}(\log p_N)^{-2}(\log N)^{-2} \ge \exp\{-\epsilon N r_{jN,\textrm{ad}}^2\}$ for all large $N$. Thus there exist $C_1,N_0>0$ such that
\[
\Pi\bigl(\|\beta_j-\beta_{0,j}\|_N\le r_{jN,\mathrm{ad}}\bigr)
\ \ge\ \exp\!\big\{-C_1\,N r_{jN,\mathrm{ad}}^2\big\},
\quad \forall\,N\ge N_0,
\]
which proves the lemma.
\end{proof}

\subsection{Proof of Theorem 1}
\label{sec:proof_thm1}
The remaining steps of \citet{Ghosal2000} to show posterior contraction require us to construct the sieve and bound its metric entropy along with proving the exponential decay rate of the Type-I/II errors of likelihood-ratio tests constructed on the sieve. 

Throughout, indices for predictors sometimes appear as $\{0,1,\dots,p_N\}$ (including the intercept block $(j=0)$) and elsewhere as $\{1,\dots,p_N\}$ (excluding it). Both conventions are equivalent for our arguments: the intercept is never subject to sparsity selection and does not affect leaf budgets or rates, so any bounds and constants remain unchanged up to immaterial additive ($\mathcal{O}(1)$) terms. For definiteness, when needed we treat ($j=0$) as fixed and excluded from counts over active predictors.

Corollary \ref{cor:joint-smallball} extends the single-coefficient analysis (Lemma \ref{lem:tree-mass}) to the full model, establishing the joint prior mass condition required by the general posterior contraction framework. It confirms that our prior assigns sufficient mass to functions that correctly mimic the true two-level (predictor and modifier) sparse structure.

\begin{corollary}[Joint concentration under predictor sparsity]
\label{cor:joint-smallball}
Let $S_0\subset\{1,\dots,p_N\}$ be the active predictor set with $|S_0|=d_0$, and for each $j\in S_0$
let $s_{0,j}$ be the size of its active modifier set and $\alpha_j\in(0,1]$ its H\"older smoothness.
Define the dimension–adaptive rates
\[
r_{jN,\mathrm{ad}}^2 \;:=\; (\log N)\,N^{-2\alpha_j/(2\alpha_j+s_{0,j})},
\quad
r_N^2 \;:=\; \frac{d_0\log p_N}{N} \;+\; \sum_{j\in S_0} r_{jN,\mathrm{ad}}^2.
\]
Fix $s_N=\lceil C_s\log N\rceil$ with any constant $C_s>0$, and a constant threshold $\varepsilon_0>0$.
Define the target neighborhood
\[
 B_{S}(r_N) \;:=\;
 \Bigl\{
    \bm\beta:\;
    \|\beta_{j}-\beta_{0,j}\|_{N}\le r_{jN,\mathrm{ad}}\ \ (\forall j\in S_0),\quad
    \#\{j\notin S_0:\ \|\beta_j\|_\infty>\varepsilon_0\}\le s_N
 \Bigr\}.
\]
Under {\textbf{(A1)}}–{\textbf{(A5)}} and {\textbf{(P1)}}–{\textbf{(P2)}}, there exists a constant $C_1^\star>0$ (independent of $N,R,p_N$) such that, for all sufficiently large $N$,
\[
 \Pi\bigl(B_{S}(r_N)\bigr)\;\ge\;\exp \ \!\bigl\{-C_{1}^{\star}\,N r_N^{2}\bigr\}.
\]
\end{corollary}

\begin{proof}
Let $\mathcal E_0:=\{d=d_0\}$ and $\mathcal E_\tau:=\{a_\tau A_{g,d_0}\le \tau\le b_\tau A_{g,d_0}\}$ for fixed $0<a_\tau<b_\tau<\infty$.
Since $\pi(d)\propto 1/d!$, Stirling’s bound yields $\pi(d_0)\ge \exp\{-c_1 d_0\log d_0\}\ge \exp\{-c_1' d_0\log p_N\}$ for large $N$.
The half–Cauchy scale prior implies $\mathbb P(\mathcal E_\tau\mid d_0)=c_\tau>0$ independent of $N$.
Hence
\begin{equation}\label{eq:mass-d-tau}
\Pi(\mathcal E_0\cap\mathcal E_\tau)\ \ge\ \exp\{-c_s\,d_0\log p_N\}.
\end{equation}
From now on we work conditionally on $\mathcal E_0\cap\mathcal E_\tau$ and multiply by \eqref{eq:mass-d-tau} at the end.
Conditional on $(d,\tau)$, the priors for different predictors $j$ are independent by {\textbf{(P1)}}–{\textbf{(P2)}}.

\medskip\noindent
\textbf{Active coordinates.}
Fix $j\in S_0$. By the approximation lemma (\Cref{lem:approx-main}), there exists a step function $g_j$ realized by a single $K_j$–leaf tree splitting only on $S_{0,j}$ with
\begin{equation}\label{eq:approx-step}
\|\beta_{0,j}-g_j\|_N \ \le\ \tfrac14\,r_{jN,\mathrm{ad}},
\quad
K_j\ \asymp\ \frac{N\,r_{jN,\mathrm{ad}}^2}{\log N}.
\end{equation}
By the tree-mass lemma (\Cref{lem:tree-mass}), the prior mass of the \emph{exact} topology+split-variables+cutpoints producing $g_j$ is
\begin{equation}\label{eq:tree-mass}
\Pi\bigl(\text{topology+splits realizing }g_j\bigr)\ \ge\ \exp\{-c_T\,K_j\log N\}\ =\ \exp\{-c_T' N r_{jN,\mathrm{ad}}^2\}.
\end{equation}
Following the proof of Lemma \ref{lem:tree-mass}, we fix $0<c_L<c_U$ and $0<a<b<c_L^2$ and define
\[
\mathcal E_{\mathrm{sc},j}:=\Bigl\{c^2\in[c_L^2,c_U^2]\Bigr\}\cap \Bigl\{\lambda_j\in[\Lambda_L(\tau,c),\Lambda_U(\tau,c)]\Bigr\},
\
\Lambda_t(\tau,c):=\sqrt{\frac{t\,c^2}{\tau^2(c^2-t)}},\ t\in\{a,b\}.
\]
and essentially mimicking the same line of argument, we obtain
\begin{equation}\label{eq:calib-mass-dp-new}
\Pi\bigl(\mathcal E_{\mathrm{sc},j}\mid \mathcal E_\tau\bigr)
\ \ge\ c_\lambda'\,\frac{d_0}{\,p_N\,(\log p_N)^2(\log N)^2}.
\end{equation}
% Conditional on the event in \eqref{eq:tree-mass} and on $\mathcal E_{\mathrm{sc},j}$, the $K_j$ leaf means are i.i.d.\ $N(0,s_j^2)$ with $s_j^2\in[a/M,b/M]$.
% Let $\delta_N:=r_{jN,\mathrm{ad}}/(4\sqrt{K_j})$. 
Then arguing similarly as in the proof of Lemma \ref{lem:smallball-main} yields
% \begin{align}
% \Pi\Bigl(\|\beta_j-g_j\|_N\le \tfrac14 r_{jN,\mathrm{ad}}\ \Bigm|\ \text{tree},\ \mathcal E_{\mathrm{sc},j}\Bigr)
% &\ge \prod_{\ell=1}^{K_j}\ \inf_{|x|\le B}\mathbb P\bigl(|\mu^{(j)}_{\ell}-x|\le \delta_N\bigr)\nonumber\\
% &\ge \bigl(c_0\,\delta_N\bigr)^{K_j}
% = \exp\!\Bigl\{-K_j\log\frac{1}{c_0\delta_N}\Bigr\}\nonumber\\
% &\ge \exp\{-c_L\,N r_{jN,\mathrm{ad}}^2\}, \label{eq:leaf-smallball}
% \end{align}
% where $B<\infty$ bounds $\|g_j\|_\infty$ (H\"olderness), and $K_j\asymp N r_{jN,\mathrm{ad}}^2/\log N$ so that $K_j\log(1/\delta_N)\lesssim N r_{jN,\mathrm{ad}}^2$.
% Combining \eqref{eq:approx-step} and \eqref{eq:leaf-smallball} gives
% \[
% \Pi\Bigl(\|\beta_j-\beta_{0,j}\|_N\le r_{jN,\mathrm{ad}}\ \Bigm|\ \text{tree},\ \mathcal E_{\mathrm{sc},j}\Bigr)
% \ \ge\ \exp\{-c_L\,N r_{jN,\mathrm{ad}}^2\}.
% \]
% Multiplying with \eqref{eq:tree-mass} and using \eqref{eq:calib-mass-dp-new} yields
\begin{equation}\label{eq:active-one-new}
\Pi\Bigl(\|\beta_j-\beta_{0,j}\|_N\le r_{jN,\mathrm{ad}}\ \Bigm|\ \mathcal E_0\cap\mathcal E_\tau\Bigr)
\ \ge\ \exp\{-c_a\,N r_{jN,\mathrm{ad}}^2\}\cdot
c''\,\frac{d_0}{\,p_N\,(\log p_N)^2(\log N)^2}.
\end{equation}
Conditional independence across $j\in S_0$ then gives
\begin{align}
\Pi\Bigl(\bigcap_{j\in S_0}\{\|\beta_j-\beta_{0,j}\|_N\le r_{jN,\mathrm{ad}}\}\ \Bigm| \ \mathcal E_0\cap\mathcal E_\tau\Bigr)
&\ge \exp \Bigl\{-c_a\,N\sum_{j\in S_0} r_{jN,\mathrm{ad}}^2\Bigr\}\cdot
\Biggl(c''\,\frac{d_0}{\,p_N\,(\log p_N)^2(\log N)^2}\Biggr)^{d_0}\nonumber\\
&= \exp \Bigl\{-c_a\,N\sum_{j\in S_0} r_{jN,\mathrm{ad}}^2
- c_s''\,d_0\log p_N - 2d_0\log\log p_N \nonumber\\
& - 2d_0\log\log N + d_0\log(c''d_0)\Bigr\}\nonumber\\
&\ge \exp \Bigl\{-c_a\,N\sum_{j\in S_0} r_{jN,\mathrm{ad}}^2 - C\,d_0\log p_N\Bigr\},
\label{eq:actives-joint-new}
\end{align}
for a suitable constant $C>0$, since $\log\log p_N,\log\log N,\log d_0=o(\log p_N)$ under {\textbf{(A5)}}.

\medskip\noindent
\textbf{Inactive coordinates.}
For $j\notin S_0$, define $I_j:=\mathbbm{1}\{\|\beta_j\|_\infty>\varepsilon_0\}$ and
$S_N:=\sum_{j\notin S_0} I_j$. We show that
$\Pi(S_N\le s_N\mid \mathcal E_0\cap\mathcal E_\tau)$ is large.

Write $\beta_j=\sum_{m=1}^M \beta_{j,m}$ with $M= \mathcal{O}(1)$ (Assumption \textbf{(A1)}) and let $L_j$ be the total number
of leaves across the $M$ trees for predictor $j$. If \emph{all} leaves satisfy
$|\mu^{(j)}_{\ell,m}|\le \varepsilon_0/M$, then $\|\beta_j\|_\infty\le \sum_{m=1}^M\|\beta_{j,m}\|_\infty
\le \sum_{m,\ell}|\mu^{(j)}_{\ell,m}|\le M\cdot(\varepsilon_0/M)=\varepsilon_0$, hence $I_j=0$.
Therefore, conditioning on $(\tau,\lambda_j,L_j)$ and using a union bound plus the Gaussian tail,
\begin{align}
\mathbb P(I_j=1\mid \tau,\lambda_j,L_j)
&\le L_j\cdot \mathbb P\!\left(|Z|>\frac{\varepsilon_0}{M}\ \Bigm|\ Z\sim \mathcal{N} \Big(0,\frac{\tau^2\lambda_j^2}{M}\Big)\right)\nonumber\\
&\le L_j\cdot C\,\min \Bigl\{1,\ \frac{\tau\lambda_j}{\varepsilon_0\sqrt{M}}\Bigr\}
\ \le\ L_j\cdot C'\,\min \Bigl\{1,\ \frac{\tau\lambda_j}{\varepsilon_0}\Bigr\},
\label{eq:Ij-union-new}
\end{align}
absorbing $\sqrt{M}$ into $C'$ since $M=\mathcal{O}(1)$.

Taking expectation over $L_j$ (independent of $(\tau,\lambda_j)$ under {\textbf{(P1)}}–{\textbf{(P2)}}),
$\mathbb E[L_j]\le C_L$ for the depth–decay Galton–Watson prior, so
\[
\sup_{j\notin S_0}\ \mathbb P(I_j=1\mid \tau)
\ \le\ C_L C'\ \mathbb E_{\lambda\sim\mathcal C^+(0,1)}\!\left[\min\!\left\{1,\ \frac{\tau\lambda}{\varepsilon_0}\right\}\right].
\]
For any $a\in(0,1]$ and $\lambda\sim\mathcal C^+(0,1)$,
\begin{equation}\label{eq:min-logineq}
\begin{split}
\mathbb E\bigl[\min\{1,a\lambda\}\bigr]
&=\frac{2}{\pi}\int_0^\infty \frac{\min\{1,a\lambda\}}{1+\lambda^2}\,d\lambda
=\frac{a}{\pi}\log \Bigl(1+\frac{1}{a^2}\Bigr)+\frac{2}{\pi}\Bigl(\frac{\pi}{2}-\arctan\!\frac{1}{a}\Bigr)\\
&\le C_0\,a\,\log \Bigl(1+\frac{1}{a}\Bigr),
\end{split}
\end{equation}
for a universal $C_0>0$.
With $a:=\tau/\varepsilon_0$ and $\tau\in\mathcal E_\tau\subset [a_\tau,b_\tau]\cdot A_{g,d_0}$, {\textbf{(P2)}} gives, for large $N$,
\begin{align}
p_\star(N)
&:=\ \sup_{j\notin S_0}\ \mathbb P(I_j=1\mid \tau)\nonumber\\
&\le C\,\frac{d_0}{\,p_N(\log p_N)^2(\log N)^2}\ \log\Bigl(1+c\,\frac{p_N(\log p_N)^2(\log N)^2}{d_0}\Bigr)\nonumber\\
&\le C'\,\frac{d_0}{\,p_N\,(\log p_N)\,(\log N)^2},\label{eq:pstar-new}
\end{align}
since $\log\!\big(1+x\big)\le \log x + 1$ and
$\log\!\big(d_0^{-1}{p_N(\log p_N)^2(\log N)^2}\big)
\le c_1\log p_N + c_2\log\log p_N + c_3\log\log N
\le C''\log p_N$ under {\textbf{(A5)}} (the $\log\log$ terms are $o(\log p_N)$).
In particular, $p_\star(N)\le 1/2$ for all sufficiently large $N$.

Now fix any subset $T\subset \{1,\dots,p_N\}\setminus S_0$ of size
\(
m\ :=\ p_N - d_0 - s_N.
\)
If $I_j=0$ for all $j\in T$, then at most the remaining $(p_N-d_0)-m=s_N$ inactive predictors can satisfy $I_j=1$; thus
\(
\{I_j=0\ \forall j\in T\}\subset\{S_N\le s_N\}.
\)
Conditional on $\tau$, the indicators $\{I_j\}_{j\notin S_0}$ are independent, so using \eqref{eq:pstar-new} and $\log(1-x)\ge-2x$ for $x\in[0,1/2]$,
\begin{align}
\mathbb P(S_N\le s_N\mid \tau)
&\ge \prod_{j\in T}\bigl(1-\mathbb P(I_j=1\mid \tau)\bigr)
\ \ge\ \bigl(1-p_\star(N)\bigr)^{m}\nonumber\\
&\ge\ \exp\bigl\{\,m\,\log(1-p_\star(N))\,\bigr\}
\ \ge\ \exp\bigl(-2\,m\,p_\star(N)\bigr).\label{eq:inactive-lb-new}
\end{align}
By {\textbf{(A5)}}, $s_N+d_0=\mathcal O(\log N)=o(p_N)$, hence $m/p_N\to 1$ and
\[
2\,m\,p_\star(N)\ \le\ 2(1+o(1))\,p_N\cdot C'\frac{d_0}{\,p_N(\log p_N)(\log N)^2}
\ =\ (1+o(1))\,\frac{2C' d_0}{(\log p_N)(\log N)^2}.
\]
Since the right-hand side is $o \big(d_0\log p_N\big)$, there exists $C_4>0$ such that, for all large $N$,
\[
2\,m\,p_\star(N)\ \le\ C_4\,d_0\log p_N,
\]
and therefore from \eqref{eq:inactive-lb-new},
\begin{equation}\label{eq:inactive-final-new}
\Pi\bigl(S_N\le s_N\ \bigm|\ \mathcal E_0\cap\mathcal E_\tau\bigr)
\ \ge\ \exp \bigl(-C_4\,d_0\log p_N\bigr).
\end{equation}

\medskip\noindent
Multiplying \eqref{eq:mass-d-tau}, \eqref{eq:actives-joint-new}, and \eqref{eq:inactive-final-new},
\begin{align*}
\Pi\bigl(B_S(r_N)\bigr)
&\ge \Pi(\mathcal E_0\cap\mathcal E_\tau)\cdot 
\Pi\Bigl(\bigcap_{j\in S_0}\{\|\beta_j-\beta_{0,j}\|_N\le r_{jN,\mathrm{ad}}\}\ \Bigm| \ \mathcal E_0\cap\mathcal E_\tau\Bigr)
\cdot
\Pi\bigl(S_N\le s_N\ \bigm|\ \mathcal E_0\cap\mathcal E_\tau\bigr)\\[2pt]
&\ge \exp \Bigl\{-c_a\,N\sum_{j\in S_0} r_{jN,\mathrm{ad}}^2 - C\,d_0\log p_N\Bigr\}
\ =\ \exp \Bigl\{-C_1^\star\,N r_N^2\Bigr\},
\end{align*}
for a suitable $C_1^\star>0$. This completes the proof.
\end{proof}

\textbf{Rationale of sieve construction.} The construction of a sieve is a critical part of proving posterior contraction. At a high level, the sieve is a well-behaved subset of the entire parameter space. Proving that the posterior concentrates on a tiny neighborhood around the truth is too hard to do in the full, infinitely complex model space. The strategy is to first show that the posterior concentrates on the much smaller, more manageable sieve, and then work within that sieve. A good sieve must satisfy two competing goals: (i) It must be \emph{large enough} so that the true function is inside it with very high prior probability, and (ii) it must be \emph{small enough} so that its complexity can be controlled.

The key to \texttt{sparseVCBART} theory is adaptation: the final contraction rate depends on the unknown $\alpha_j$ and the unknown $s_{0,j}$. However, the sieve itself must be constructed without knowing these true values. This is why we use a worst-case rate $\bar{r}_N$ to define the complexity cap, $L_N$ (the maximum number of leaves). We need a single leaf budget, $L_N$, that is guaranteed to be sufficient for any possible true function that satisfies our assumptions. The number of leaves required to approximate a function gets larger as the function becomes less smooth (smaller $\alpha$) and higher-dimensional (larger $s$). Therefore, the \emph{hardest} possible function to approximate is one with the minimum allowed smoothness ($\alpha_{\min}$) and the maximum possible dimension ($R$).

In \Cref{def:sieve_revised}, we ensure our sieve is large enough to contain a good approximation of the true function, no matter what its true $\alpha_j$ and $s_{0,j}$ are.

\begin{definition}[Dimension–adaptive sparse sieve]\label{def:sieve_revised}
Let
\(
s_N=\big\lceil C_s\,\log N\big\rceil
\)
for a constant $C_s>0$ chosen large enough so that $d_0\le s_N$ for all large $N$.
Fix a minimal smoothness $\alpha_{\min}\in(0,1]$ and define the worst–case per–component rate
\[
r_{N,\mathrm{wc}}^2 \;:=\; (\log N)\,N^{-2\alpha_{\min}/(2\alpha_{\min}+R)}.
\]
Set the envelope rate
\[
\bar r_N^2 \;:=\; \frac{s_N\log p_N}{N}\;+\; s_N\, r_{N,\mathrm{wc}}^2,
\qquad
L_N \;:=\; C_L\,\frac{N\,\bar r_N^2}{\log N},
\]
for a sufficiently large constant $C_L>0$, and the sieve inactivity threshold
\(
\varepsilon_N :={C_\varepsilon}/{\log N}, \ C_\varepsilon>0.
\)
For a set $S\subset\{0,1,\dots,p_N\}$ of size $|S|=k\le s_N$, define
\[
\mathcal F_N(S)\footnote{Since $\beta_j: 
\mathcal{Z}\subseteq [0,1]^R \mapsto \mathbb{R}$, so the notation $\|\beta_j\|_{\infty} \coloneqq \sup_{\bm{z} \in \mathcal{Z}}|\beta_{j}(\bm{z})|$} \;:=\; \left\{(\bm\beta,\sigma^2):
\begin{array}{l}
\|\beta_j\|_\infty \le \varepsilon_N \quad (j\notin S)\\[2pt]
\text{the $M$ trees forming }\beta_j\text{ have in total }\le L_N\text{ leaves}\quad (j\in S),\\[2pt]
0<\sigma^2\le c_0
\end{array}\right\}.
\]
The full sieve is the model–selection union
\[
\mathcal F_N^\star \;:=\; \bigcup_{k=0}^{s_N}\ \bigcup_{\substack{S\subset\{0,1,\dots,p_N\}\\ |S|=k}} \mathcal F_N(S).
\]
\end{definition}

\begin{remark}
The specific structure of $\mathcal{F}_N^\star$, is designed to discard unreasonable functions in three ways, perfectly mirroring the assumptions of our model: The most important step is taking the union over only sparse sets of predictors. This immediately throws away the vast majority of the model space (all non-sparse models). This is justified by the assumption that the true model is sparse ($d_0 \ll p_N$). Next, for the covariates that are considered active, we do not allow their corresponding functions $\beta_j$ to be infinitely complex. The cap on the total number of leaves, $L_N$, prevents overfitting by ensuring the functions remain reasonably simple. Finally, for predictors that are considered inactive, we do not require them to be exactly zero, which is difficult with continuous shrinkage priors. Instead, we only require that they are negligibly small ($\|\beta_j\|_\infty \le \varepsilon_N$). This creates a ``soft sparsity" that is compatible with the prior and is sufficient for the theoretical arguments.
\end{remark}

\begin{lemma}[Prior mass of the sieve]\label{lem:sieve-mass_corrected}
Assume {\textbf{(A1)}}–{\textbf{(A5)}} and priors {\textbf{(P1)}}–{\textbf{(P3)}}. Then there exists $c_3>0$ such that, for all large $N$,
\[
\Pi\bigl((\mathcal F_N^\star)^c\bigr)\ \le\ \exp \bigl\{-c_3\,N\,r_N^2\bigr\}.
\]
\end{lemma}

\begin{proof}
By definition, $(\mathcal F_N^\star)^c$ is contained in the union of:
\[
E_{1,N} := \{\exists \ j \in \{0,1,\dots,p_N\} \ \text{such that} \ \sum_{m=1}^{M} L_{j,m} > L_N \},\ \
E_{2,N} := \Bigl\{\#\{j:\|\beta_j\|_\infty>\varepsilon_N\}> s_N\Bigr\},\ \
E_{3,N} := \{\sigma^2>c_0\}.
\]
We bound $\Pi(E_{m,N})$ for $m=1,2,3$.

\medskip\noindent\textbf{(i) Bounding $\Pi(E_{1,N})$.}
Under \textbf{(P1)}, trees are independent across $(j,m)$ and, for a single tree,
Theorem~5.1 of \citet{rockovasaha} implies that there exist constants $c_1>0$ and $\ell_0\ge 1$ such that
\begin{equation}\label{eq:RS-tail}
\Pi(L>\ell)\ \le\ \exp\{-c_1\,\ell\log \ell\}\qquad\text{for all }\ell\ge \ell_0.
\end{equation}

Fix $j$. By the elementary inclusion
\[
\Bigl\{\sum_{m=1}^{M}L_{j,m}>L_N\Bigr\}\ \subset\ \bigcup_{m=1}^{M}\Bigl\{L_{j,m}>\tfrac{L_N}{M}\Bigr\},
\]
we have, by a union bound and \eqref{eq:RS-tail},
\begin{equation}\label{eq:perj-tail}
\Pi\!\left(\sum_{m=1}^{M} L_{j,m}>L_N\right)
\ \le\ 
\sum_{m=1}^{M}\Pi\!\left(L_{j,m}>\tfrac{L_N}{M}\right)
\ \le\
M\,\exp\!\left\{-c_1\,\frac{L_N}{M}\,\log \left(\frac{L_N}{M}\right)\right\},
\end{equation}
for all $N$ large enough so that $L_N/M\ge \ell_0$.
Union bounding over $j\in\{0,1,\dots,p_N\}$ yields
\begin{equation}\label{eq:E1N-union}
\Pi(E_{1,N})
\ \le\ 
(p_N+1)\,M\,\exp\!\left\{-c_1\,\frac{L_N}{M}\,\log \left(\frac{L_N}{M}\right)\right\}.
\end{equation}

We write $A_N:=N\bar r_N^2$ and $x_N:={L_N}/{M}={C_L A_N}/{(M \log N)}$.
Since $A_N\ge N s_N r_{N,\mathrm{wc}}^2$ and $r_{N,\mathrm{wc}}^2=(\log N)\,N^{-\theta_N}$ with 
$\theta_N:={2\alpha_{\min}}/{(2\alpha_{\min}+R)}\in(0,1)$, we have
\[
x_N\ \ge\ \frac{C_L}{M}\,\frac{N s_N r_{N,\mathrm{wc}}^2}{\log N}
\ =\ \frac{C_L}{M}\,N^{1-\theta_N}\,s_N
\ \asymp\ N^{1-\theta_N}\,\log N.
\]
Hence, there exists $c_0\in(0,1)$ and $N_0$ such that, for all $N\ge N_0$,
\begin{equation}\label{eq:logxN-lb}
\log x_N \ \ge\ c_0\,\log N.
\end{equation}
Therefore,
\begin{equation}\label{eq:xlogx-lb}
x_N\log x_N
\ =\ \frac{C_L}{M}\,\frac{A_N}{\log N}\,\log x_N
\ \ge\ 
\frac{C_L}{M}\,\frac{N s_N r_{N,\mathrm{wc}}^2}{\log N}\cdot c_0\log N
\ =\ c_\ast\,N s_N r_{N,\mathrm{wc}}^2,
\end{equation}
for $c_\ast:=c_0 C_L/M>0$.
Insert \eqref{eq:xlogx-lb} into \eqref{eq:E1N-union} to obtain
\[
\Pi(E_{1,N})
\ \le\
(p_N+1)\,M\,\exp\!\Big\{-c_1 c_\ast\,N s_N r_{N,\mathrm{wc}}^2\Big\}.
\]
Finally, since $N\bar r_N^2=s_N\log p_N + N s_N r_{N,\mathrm{wc}}^2\ge N s_N r_{N,\mathrm{wc}}^2$
and $\log((p_N+1)M)=o\big(N s_N r_{N,\mathrm{wc}}^2\big)$ under \textbf{(A5)}, the polynomial prefactor is absorbed into the exponential:
there exist $c_1''>0$ and $N_1$ such that, for all $N\ge N_1$,
\[
\Pi(E_{1,N})\ \le\ \exp\{-c_1''\,N s_N r_{N,\mathrm{wc}}^2\}\ \le\ \exp\{-c_1''\,N\bar r_N^2\}.
\]
This proves the desired bound for $\Pi(E_{1,N})$.
% For a single binary tree under {\textbf{(P1)}} with depth–$d$ split probability $\gamma^d$ ($N^{-1}<\gamma<1/2$),
% Theorem~5.1 of \citet{rockovasaha} gives, for $\ell$ large enough,
% \[
% \Pi(L>\ell)\ \le\ \exp\{-c_1\,\ell\log \ell\}.
% \]
% There are $(p_N+1)M$ independent trees across predictors (including the baseline). By a union bound,
% \[
% \Pi(E_{1,N})\ \le\ (p_N+1)M\ \exp\{-c_1\,L_N\log L_N\}.
% \]
% With $L_N=C_L\,N\bar r_N^2/\log N$, we have for large $N$,
% \[
% L_N\log L_N\;=\;\frac{C_L N\bar r_N^2}{\log N}\,\log\!\Big(\frac{C_L N\bar r_N^2}{\log N}\Big)
% \ \ge\ c_1'\,N\bar r_N^2,
% \]
% because $\log(C_L N\bar r_N^2/\log N)\asymp \log N\to\infty$. Hence
% \[
% \Pi(E_{1,N})\ \le\ \exp\{-c_1''\,N\bar r_N^2\},
% \]
% after absorbing the polynomial pre-factor $(p_N+1)M$ into the exponential (since $\log((p_N+1)M)=o(N\bar r_N^2)$ as $N\bar r_N^2\gtrsim s_N\log p_N = C_s (\log N)(\log p_N)$ by construction).

\medskip\noindent\textbf{(ii) Bounding $\Pi(E_{2,N})$.}
Let $I_j:=\mathbbm{1}\{\|\beta_j\|_\infty>\varepsilon_N\}$ and $S_N:=\sum_{j=1}^{p_N} I_j$.
Then $E_{2,N}=\{S_N>s_N\}$ with $s_N=\lceil C_s\log N\rceil$. Following exactly the same set of arguments to bound the number of in-actives in the proof of Corollary \ref{cor:joint-smallball}, the inequality in \eqref{eq:min-logineq} yields with $a \coloneqq \tau/(\epsilon_N \sqrt{M}), $
\[
\mathbb{P}(I_j =1 \mid \tau) \le C_1 \frac{\tau}{\epsilon_N \sqrt{M}} \log \left(1+ \frac{\epsilon_N \sqrt{M}}{\tau} \right).
\]

Manipulating our prior assumption \textbf{(P2)} slightly by taking $\tau \in (0,T_N (d)]$ a.s. with $T_N(d)=b_{\tau} A_{g,d},$ we have $T_N(d)/(\epsilon_N \sqrt{M}) \lesssim d p_N^{-1} (\log p_N)^{-2} (\log N)^{-1}.$ thus for all $\tau \le T_N(d),$ 
\[
\frac{\tau}{\epsilon_N \sqrt{M}} \log \left(1+\frac{\epsilon_N \sqrt{M}}{\tau} \right) \le \frac{T_N(d)}{\epsilon_N \sqrt{M}} \log \left(1+\frac{\epsilon_N \sqrt{M}}{T_n(d)} \right) \lesssim \frac{d}{p_N \log p_N \log N}.
\]
Using the above inequality, we have the \emph{uniform} per-predictor bound over $\tau,$

\begin{equation}\label{eq:pstar}
p_\star(d):=\sup_{1\le j\le p_N}\sup_{0 < \tau \le T_N(d)} \mathbb P(I_j=1\mid \tau)
\ \le\ C_2\frac{d}{p_N\,\log p_N \log N}.
\end{equation}
Thus the conditional mean
\begin{equation}\label{eq:muN}
\mu_N(d) := \mathbb E(S_N\mid d)\ =\ \sum_{j=1}^{p_N}\mathbb P(I_j=1\mid d)\ \le p_N p_{\star}(d) \le C_2 \frac{d}{\log p_N \log N}.
\end{equation}
Fix $C_s >0$ large so that, for all large $N$, $s_N = \lceil C_s \log N \rceil \ge 2 \mu_{N}(d)$ whenever $d \le D_N \coloneqq \lfloor C_d (\log N)^2 \rfloor$, which is possible by \eqref{eq:muN}. 

\emph{Chernoff bound.}
The multiplicative Chernoff bound gives, for all $d \le D_N,$
\begin{equation}\label{eq:chernoffbd}
\begin{split}
\mathbb P\bigl(S_N> s_N\mid d\bigr) & \le\ \exp \Bigl(-\frac{(s_N-\mu_N (d))^2}{3\mu_N(d)}\Bigr) \\
& \le\ \exp \Bigl(-c\cdot\frac{s_N^2}{\mu_N(d)}\Bigr)\\
& \le \exp(-c'(\log N)(\log p_N)),
\end{split}    
\end{equation}
where the last step uses $\mu_N(d) \le C (\log N)/\log p_N$ when $d \le D_N$.

\emph{Average over $d$.} We split on $\{d \le D_N\} \cup \{d > D_N\}:$ 
\[
\Pi(E_{2,N}) = \sum_{d=0}^{p_N} \pi(d) \mathbb{P}(S_N > s_N \mid d) \le \underbrace{\sum_{d \le D_N} \pi(d) \exp(-c'(\log N)(\log p_N))}_{by \eqref{eq:chernoffbd}} + \sum_{d > D_N} \pi(d).
\]

Since $\pi(d) \propto 1/d! \implies \sum_{d > D_N} \pi(d) \le C/(D_N !),$ Stirling's bound yields
\begin{equation*}
    \begin{split}
        \sum_{d > D_N} \pi(d) & \le \exp\{-D_N \log D_N + \mathcal{O}(D_N)\} \\
        & = \exp\{-c''(\log N)^2 \log \log N\} \\
        & \le \exp\{-c'''(\log N)(\log p_N)\},
    \end{split}
\end{equation*}
for all large $N$. Therefore, \[
\Pi(E_{2,N}) \le \exp\{-c_2 (\log N)(\log p_N)\}.
\]
By definition $N \bar{r}_N ^2 = s_N \log p_N + N s_N r_{N,\textrm{wc}}^2 \ge s_N \log p_N \asymp (\log N)(\log p_N)$, so we have \begin{equation}\label{eq:E2N-final}
\Pi(E_{2,N}) \le \exp\{-c_2 N \bar{r}_N ^2\}    
\end{equation}
% \medskip\noindent\textbf{(ii) Bounding $\Pi(E_{2,N})$.}
% Let $I_j:=\mathbbm{1}\{\|\beta_j\|_\infty>\varepsilon_N\}$ and $S_N:=\sum_{j=1}^{p_N} I_j$,
% so $E_{2,N}=\{S_N>s_N\}$ with $s_N=\lceil C_s\log N\rceil$.
% Under \textnormal{(P1)}–\textnormal{(P3)}, conditional on the global scale $\tau$ the ensembles across $j$ are independent. 
% By the leaf–prior small–ball bound (as in \eqref{eq:min-logineq}), for $a:=\tau/(\varepsilon_N\sqrt M)$,
% \[
% \mathbb P(I_j=1\mid \tau)\ \le\ C_1\,a\,\log(1+a^{-1}).
% \]
% Assume $\tau$ is a.s. supported on $(0,T_N]$ with a deterministic $T_N$ satisfying $T_N/(\varepsilon_N\sqrt M)\lesssim [p_N(\log p_N)^2\log N]^{-1}$.
% Then for all $\tau\le T_N$,
% \[
% \mathbb P(I_j=1\mid \tau)\ \le\ C_2\,\frac{1}{p_N\,\log p_N\,\log N}\,.
% \]
% Hence, $\mu_N(\tau):=\mathbb E(S_N\mid \tau)\le C_2\,(\log N)^{-1}(\log p_N)^{-1}$ and (by the multiplicative Chernoff bound for independent Bernoulli’s with possibly different means)
% \[
% \mathbb P(S_N>s_N\mid \tau)\ \le\ \exp\!\Big(-c\,\frac{s_N^2}{\mu_N(\tau)}\Big)
% \ \le\ \exp\{-c'\,(\log N)(\log p_N)\}.
% \]
% The bound is uniform over $\tau\in(0,T_N]$, so averaging over $\tau$ yields
% \[
% \Pi(E_{2,N})\ \le\ \exp\{-c'\,(\log N)(\log p_N)\}
% \ \le\ \exp\{-c''\,N\bar r_N^2\},
% \]
% since $N\bar r_N^2\ge s_N\log p_N\asymp(\log N)(\log p_N)$ by construction.

\medskip\noindent\textbf{(iii) Bounding $\Pi(E_{3,N})$.}
Under {\textbf{(P3)}}, $\sigma^2\sim \mathcal{IG}(\nu_N/2,\beta)$ with $\nu_N = \mathcal{O}(N r_N ^2)$, so $Y:=1/\sigma^2\sim\mathrm{Gamma}(\nu_N/2,\beta)$.
For any fixed $c_0>0$, and any $ t \in (0,\beta),$
\begin{equation*}
    \begin{split}
        \mathbb{P}(\sigma^2 > c_0) & = \mathbb{P}(Y < 1/c_0) \\
        & \le \exp(t/c_0) \mathbb{E}(\exp(-tY)) \\
        & = \exp(t/c_0)\left(\frac{\beta}{\beta+t} \right)^{\nu_N/2} \\
        & = \exp \{ (t/c_0)  - (\nu_N/2) \cdot \log(1+t/\beta)\}
    \end{split}
\end{equation*}
Picking $t=\beta/2$, we have 
\[
\mathbb{P}(\sigma^2 > c_0) \le \exp\{\beta/(2c_0) - (\nu_N/2) \cdot \log(3/2) \} \le \exp\{-c \nu_N/2 + C\} \le \exp\{-c_5 ' Nr_N^2\}
\]
% \[
% \Pi(\sigma^2>c_0)\ =\ \mathbb{P}\Bigl(Y<\tfrac{1}{c_0}\Bigr)\ \le\ \exp\{-c_5\,\nu_N\}
% \ =\ \exp\{-c_5'\,N\bar r_N^2\},
% \]
% by a standard Chernoff bound for Gamma left tails.

\medskip
Finally, by a union bound,
\[
\Pi\bigl((\mathcal F_N^\star)^c\bigr)\ \le\ \Pi(E_{1,N})+\Pi(E_{2,N})+\Pi(E_{3,N})
\ \le\ \exp \{-c_1'' N\bar r_N^2\}+\exp\{-c_2 N\bar r_N^2\}+\exp\{-c_5' N r_N^2\}.
\]
Noting that $\bar{r}_N^2 \ge r_N ^2$ and absorbing constants yields the claim with $c_3:=\min\{c_1'',c_2,c_5'\}>0$.
\end{proof}
The following shell-partitioning approach in \Cref{lem:entropy} has appeared in classic Bayesian contraction analyses (see \citet[Thm 8.12]{Ghoshal2017}). Specifically, \citet[Lemma 13]{vandervaart2011} construct tests by a \emph{peeling} sieve: they first partition the parameter set into dyadic shells $\{f \in F_{n,r}: 4^j \epsilon r \le \|f-f_0\|_n < 4^{j+1} \epsilon r\}$, then on each shell take a maximal $(jr/2)-$ packing and build Neyman-Pearson tests at the packing points.
\begin{lemma}[Covering number of the sieve]\label{lem:entropy}
Let $\mathcal F_N^\star$ be the sieve of Definition~\ref{def:sieve_revised} that enforces hard sparsity with $|\text{supp}(\beta)| \le s_N=\lceil C_s\log N\rceil, \ \beta_j \equiv 0 \ \text{for} \ j \notin \text{supp}(\beta)$ and a leaf--budget cap
$L_N=C_L\,N\bar r_N^2/\log N$. Fix $M_\star\ge 1$ and, for $m=0,1,2,\dots$, define the $m$-th shell
\[
\Theta_{N,m}\ :=\ \Bigl\{(\bm\beta,\sigma^2)\in\mathcal F_N^\star:\ 
 2^{m} M_\star r_N < \|\bm\beta-\bm\beta_0\|_{N}\ \le\ 2^{m+1} M_\star r_N,\ 
 0<\sigma^2\le c_0\Bigr\}.
\]
Equip the parameter space with the prediction semimetric
\[
d_n\bigl((\bm\beta,\sigma^2),(\bm\beta',{\sigma'}^{2})\bigr)
:= N^{-1/2}\|\Delta(\theta,\theta')\|_2 \;+\; |\sigma^2-{\sigma'}^{2}|, \ \ \Delta(\theta,\theta')\footnote{For $\theta=(\bm\beta,\sigma^2)$ write $\xi(\theta)=(\xi_i(\theta))_{i=1}^N$ with
$\xi_i(\theta)=\beta_0(\bm Z_i)+\sum_{j=1}^{p_N}\beta_j(\bm Z_i)X_{ij}$} = \xi(\theta)-\xi(\theta').
\]
Assume {\textbf{(A1)}--\textbf{(A6)}}. Then for any $0<\epsilon\le r_N$ there exist constants
$C_4>0$ and $C_{\mathrm{sh}}>0$ (independent of $N,R,p_N,m$) such that, for all large $N$, for $0 < \epsilon \le r_N$, we have the entropy bound:
\begin{equation}\label{eq:localized-entropy-bound}
\log N \bigl(\epsilon,\Theta_{N,m}, d \bigr)\ \le\ C_4\,N\,r_N^2+C_{\mathrm{sh}}\,m\,\log N.
\end{equation}
In particular, since $N r_N^2\ge c\,s_N\log p_N\gtrsim (\log N)\log p_N$ under \textbf{(A5)}, the right-hand side is $\lesssim (1+m)\,N r_N^2$ for all large $N$.
\end{lemma}

\begin{proof}
By the functional RE condition \textbf{(A6)}, for any pair supported on at most
$C_{\mathrm{RE}}d_0$ coordinates,
\[
\sqrt{c_{\min}}\;\|\bm\beta-\bm\beta'\|_{N}\ \le\ \frac{1}{\sqrt N}\,\|\Delta(\theta,\theta')\|_2\ \le\ \sqrt{c_{\max}}\;\|\bm\beta-\bm\beta'\|_{N},
\]
where $\|\bm\beta-\bm\beta'\|_N:=\big(\sum_j\|\beta_j-\beta'_j\|_N^2\big)^{1/2}$. Hence $d_n$ is equivalent (up to universal constants) to
\[
\tilde d\bigl((\bm\beta,\sigma^2),(\bm\beta',{\sigma'}^{2})\bigr):=\|\bm\beta-\bm\beta'\|_{N}+|\sigma^2-{\sigma'}^{2}|,
\]
on the sieve $\mathcal F_N^\star$ (which enforces $|S|\le s_N\le C_{\mathrm{RE}}d_0$). It therefore suffices to bound covering numbers under $\tilde d$; the same bound (up to constants) then holds for $d_n$.

\textbf{Note:} For the purpose of an \emph{upper} bound on metric entropy, we replace the soft inactivity constraint  $\|\beta_j\|_\infty\le \varepsilon_N$ for $j\notin S$ by the hard constraint $\beta_j\equiv 0$; this defines a subclass of $\mathcal F_N^\star$ and therefore cannot increase covering numbers.

\textbf{Support decomposition.}
Write
\[
\mathcal F_N^\star=\bigcup_{k=0}^{s_N}\ \bigcup_{\substack{S\subset\{1,\dots,p_N\}\\ |S|=k}}\ \mathcal F_N(S),
\qquad
\Theta_{N,m}(S):=\Theta_{N,m}\cap\mathcal F_N(S),
\]
where, \emph{by the definition of the sieve}, $\mathcal F_N(S)$ consists of parameters with $\beta_j\equiv 0$ for all $j\notin S$ (hard sparsity) and with total leaf budget $\sum_{j\in S}K_j\le L_N$. Then
\begin{equation}\label{eq:union}
N(\epsilon,\Theta_{N,m},\tilde d)
\ \le\ \sum_{k=0}^{s_N}\ \sum_{\substack{S\subset\{1,\dots,p_N\}\\ |S|=k}}\ 
N(\epsilon,\Theta_{N,m}(S),\tilde d).
\end{equation}
Taking logs and using $\sum_{k\le s_N}\binom{p_N}{k}\le (s_N+1)\binom{p_N}{s_N}$ together with
$\max_{|S|=k}\log N(\epsilon,\Theta_{N,m}(S),\tilde d)\le \sup_{|S|=s_N}\log N(\epsilon,\Theta_{N,m}(S),\tilde d)$,
\begin{equation}\label{eq:model-select-layer}
\log N(\epsilon,\Theta_{N,m},\tilde d)\ \le\
\underbrace{\log \bigl((s_N+1)\tbinom{p_N}{s_N}\bigr)}_{\mathcal T_1}
\ +\ \sup_{|S|=s_N}\ \log N(\epsilon,\Theta_{N,m}(S),\tilde d).
\end{equation}

By Stirling’s bound, $\log\binom{p_N}{s_N}\le s_N\log(e p_N/s_N)\le s_N\log p_N$ and $\log(s_N+1)\le s_N$, so
\begin{equation}\label{eq:T1}
\mathcal T_1\ \le\ c_1\,s_N\log p_N.
\end{equation}

On fixing $S$ with $|S|=k\le s_N$. Since $\beta_j\equiv 0$ for $j\notin S$ inside $\mathcal F_N(S)$, the inactive block contributes nothing to $\tilde d$, and we only need to cover $\{\beta_j:j\in S\}$ together with $\sigma^2$:
\begin{equation}\label{eq:isometryineq}
N\bigl(\epsilon,\ \Theta_{N,m}(S),\ \tilde d\bigr)\ \le\
N\Bigl(\epsilon,\ \Big(\prod_{j\in S}\mathcal G_{j,m}\Bigr)\times[0,c_0],\ \tilde d_S\Bigr),
\end{equation}
where $\tilde d_S\bigl((\{\beta_j\},\sigma^2),(\{\beta'_j\},{\sigma'}^2)\bigr):=\big(\sum_{j\in S}\|\beta_j-\beta'_j\|_N^2\big)^{1/2}+|\sigma^2-{\sigma'}^2|$, and $\mathcal G_{j,m}$ is the localized class for coordinate $j$ defined next.

Let $K_j$ denote the total leaf count for predictor $j$ (across its $M=\mathcal{O}(1)$ trees), and let $\mathcal F_j(K)$ be the class of step functions on $[0,1]^R$ with at most $K$ leaves, induced by axis-aligned splits at empirical cutpoints. On the $m$-th shell,
\begin{align*}
\|\beta_j\|_N
&\le \|\beta_j-\beta_{0,j}\|_N+\|\beta_{0,j}\|_N
\ \le\ 2^{m+1}M_\star r_N + B_{0,j}
\ \le\ R_{j,m},
\end{align*}
where $B_{0,j}:=\|\beta_{0,j}\|_N<\infty$ (bounded under \textbf{(A4)}), and we set $R_{j,m}:=B_0+2^{m+1}M_\star r_N$ with $B_0:=\max_{u\in S_0}B_{0,u}$. Define the localized active class
\begin{equation*}
\mathcal G_{j,m}\ :=\ \bigcup_{K=1}^{L_N}\ \Bigl(\ \mathcal F_j(K)\ \cap\ \{\|\beta_j\|_N\le R_{j,m}\}\ \Bigr).
\end{equation*}
To see that \eqref{eq:isometryineq} holds, we first observe the set inclusion $\Theta_{N,m}(S) \subset \left(\prod_{j \in S} \mathcal{G}_{j,m} \right) \times [0,c_0].$ Now, if we define the projection $\pi_S : \Theta_{N,m}(S) \mapsto \left( \prod_{j \in S} \mathcal{G}_{j,m} \right) \times [0,c_0],$ then $\pi_S$ is an isometry between $(\Theta_{N,m}(S),\tilde{d})$ and its image $\pi_S(\Theta_{N,m}(S))$ endowed with $\tilde{d}_S$, as $\tilde{d}(\theta,\theta')=\tilde{d}_S(\pi_S (\theta),\pi_S (\theta')), \ \forall \ \theta,\theta' \in \Theta_{N,m}(S).$ Again, harping on the set-inclusion, any $\epsilon-$cover of $\left(\prod_{j \in S} \mathcal{G}_{j,m}\right) \times [0,c_0]$ under $\tilde{d}_S$ is also an $\epsilon-$cover of $\pi_S(\Theta_{N,m}(S)),$ and by the isometry of $\Theta_{N,m}(S)$ under $\tilde{d}$ precisely gives \eqref{eq:isometryineq}.

\textbf{Covering $\mathcal F_j(K)\cap\{\|\beta_j\|_N\le R_{j,m}\}$ under $\|\cdot\|_N$.}
A $K$-leaf tree is determined by: (i) a topology ($\le 4^K$ choices)\footnote{The number of rooted full binary trees with $K$ leaves is the $(K-1)$-st Catalan number $C_{K-1}$ and $C_{K-1} \le 4^{K-1} \le 4^K$}, (ii) $K-1$ split coordinates ($\le R^K$ choices), (iii) $K-1$ split locations chosen among empirical cutpoints ($\le N^K$ choices). Thus, the number of induced sample partitions is $\le (C_0 RN)^{K}$ for universal $C_0 \ge 4$. Fix such a partition with cells $\{\Omega_1,\dots,\Omega_K\}$; any $f$ is specified by $\bm{v}=(v_1,\dots,v_K)\in\mathbb R^K$ (cell heights). Writing $w_k:=N^{-1}\sum_{i=1}^N\mathbbm{1}\{Z_i\in\Omega_k\}\in[0,1]$, $\sum_k w_k=1$, then for any two functions $f,g$ with vectors $\bm{v},\bm{u} \in \mathbb{R}^{K},$ \[
\|f-g\|_N ^2 = \sum_{k=1}^{K} w_k(v_k-u_k)^2 =: \|\bm{v}-\bm{u}\|_{2,w} ^2, 
\]
we note that the map $f \leftrightarrow \bm{v}$ is an isometry between the function class (under $\|.\|_N)$ and $(\mathbb{R}^K,\|.\|_{2,w}),$ where $\|\bm{x}\|_{2,w} ^2 \coloneqq \sum_{k=1}^{K} w_k x_k ^2.$ The ball of functions $\{f:\|f\|_N \le R_{j,m}\}$ is exactly the weighted Euclidean ball $B_{2,w}(R_{j,m}) \coloneqq \{\bm{v} \in \mathbb{R}^K: \|\bm{v}\|_{2,w} \le R_{j,m}\}$ under the isometry. Applying the linear change of variables $T:\mathbb{R}^{K} \mapsto \mathbb{R}^{K}$ given by $(T \bm{v})_k = \sqrt{w_k}v_k,$ the weighted ball $B_{2,w}(R_{j,m})$ maps to the usual Euclidean ball $B_2 (R_{j,m}) \coloneqq \{\bm{u} \in \mathbb{R}^{K}: \|\bm{u}\|_2 \le R_{j,m}\}.$ Thus, we have the exact equality of the covering numbers: 
\[
N(\delta,\{f:\|f\|_N \le R_{j,m}\},\|.\|_N)=N(\delta,B_{2,w}(R_{j,m}),\|.\|_{2,w})=N(\delta,B_2(R_{j,m}),\|.\|_2)
\]
This yields the standard volume bound:
\[
N\Bigl(\delta,\{\|\beta_j\|_N\le R_{j,m}\},\|\cdot\|_N\Bigr)\ \le\ \left(\frac{C R_{j,m}}{\delta} \right)^K.
\]
Multiplying by $(C_0 RN)^K$ for partitions yields
\begin{equation}\label{eq:perK-entropy}
\log N\bigl(\delta,\ \mathcal F_j(K)\cap\{\|\beta_j\|_N\le R_{j,m}\},\ \|\cdot\|_N\bigr)
\ \le\ K\left[c\,\log(RN)\ +\ \log\left(\frac{C R_{j,m}}{\delta}\right)\right].
\end{equation}

\textbf{Product covering and $\sigma^2$.}
Let $\mathcal K:=\{1,2,4,\dots,L_N\}$ (so $|\mathcal K|\le 1+\log_2 L_N\lesssim\log N$). Then
\[
\prod_{j\in S}\mathcal G_{j,m}\ \subset\ \bigcup_{(K_j)\in\mathcal K^{k}}\ \prod_{j\in S}\Bigl(\mathcal F_j(K_j)\cap\{\|\beta_j\|_N\le R_{j,m}\}\Bigr).
\]
Recall that we want an $\epsilon-$ cover of $\Theta_{N,m}(S) \subseteq \left(\prod_{j \in S} \mathcal{G}_{j,m} \right) \times [0,c_0]$ under the metric $\tilde{d}_S$. From the definition of $\Theta_{n,m}$, 
\[
N(\epsilon,\Theta_{N,m}(S),\tilde{d}_S) \le \sum_{(K_j) \in \mathcal{K}^k} N \left( \epsilon, \left(\prod_{j \in S} \underbrace{\left(\mathcal{F}_j (K_j) \cap \{\|\beta_j\|_N \le R_{j,m}\}\right)}_{=:A_{j,K_j}}\right)\times[0,c_0],\tilde{d}_S \right). 
\]
Taking logs, we have the following inequality:
\[
\log N(\epsilon,\Theta_{N,m}(S),\tilde{d}_S) \le k \log|\mathcal{K}| + \sup_{(K_j)\in \mathcal{K}^k} \log N \left(\epsilon, \left(\prod_{j \in S} A_{j,K_j}\right) \times [0,c_0],\tilde{d}_s \right).
\]
We fix a tuple $(K_j)$. Let $\delta_j > 0$ be radii such that $\sum_{j \in S} \delta_j ^2 \le (\epsilon/2)^2.$ For each $j$, we pick a $\delta_j-$net of $A_{j,K_j}$ in $\|.\|_N$ with size $N(\delta_j,A_{j,K_j},\|.\|_N).$ Also we cover the interval $[0,c_0]$ by an $(\epsilon/2)-$net in $|.|$ of size $\le 1+2c_0/\epsilon,$ hence with log-cost $\le C_{\sigma} \log(c_0/\epsilon).$ Then the Cartesian product of those nets is an $\epsilon-$net of $(\prod_{j \in S} A_{j,K_j}) \times [0,c_0]$ in $\tilde{d}_S$. Therefore,
\[
N \left(\epsilon, \left(\prod_{j\in S} A_{j,k} \right) \times [0,c_0],\tilde{d}_S \right) \le \left(\prod_{j \in S} N(\delta_j,A_{j,K_j},\|.\|_N) \right). \left\lceil \frac{2 c_0}{\epsilon} \right\rceil 
\]
This implies \[
\log N \left(\epsilon, \left(\prod_{j\in S} A_{j,k} \right) \times [0,c_0],\tilde{d}_S \right) \le \sum_{j \in S} \log N(\delta_j,A_{j,K_j},\|.\|_N) + C_{\sigma} \log \left(\frac{c_0}{\epsilon} \right).
\]
By our \eqref{eq:perK-entropy} bound,
\[
\log N(\delta_j,A_{j,K_j},\|.\|_N) \le K_j \left[c \log(RN) + \log \left( \frac{CR_{j,m}}{\delta_j} \right) \right].
\]
Summing over $j \in S$ yields
\[
\sum_{j \in S} \log N(\delta_j,A_{j,K_j},\|.\|_N) \le \sum_{j \in S} K_j \left[c \log (RN)+ \log \left(\frac{CR_{j,m}}{\delta_j}\right)\right].
\]
Combining these steps and using the fact that $|\mathcal{K}| \lesssim \log N$,
% Cover the product under the $\ell_2$-sum metric $\big(\sum_{j\in S}\|\cdot\|_N^2\big)^{1/2}$ by allocating squared radii $\delta_j^2$ with $\sum_{j\in S}\delta_j^2\le (\epsilon/2)^2$, and cover $[0,c_0]$ in the Euclidean metric by $\lceil 2c_0/\epsilon\rceil$ points (cost $\lesssim \log(c_0/\epsilon)$). Using \eqref{eq:perK-entropy} with $\delta_j$, the subadditivity of $\log N(\cdot)$ over products, and summing over dyadic $(K_j)\in\mathcal K^k$,
\begin{align}
\log N(\epsilon,\Theta_{N,m}(S),\tilde d)
&\le k\log|\mathcal K|\ +\ \sup_{(K_j)\in\mathcal K^k}\;\sum_{j\in S}\;K_j\Bigl[c\,\log(RN)\ +\ \log\Bigl(\frac{C R_{j,m}}{\delta_j}\Bigr)\Bigr]\nonumber\\
&\qquad\qquad +\ C_\sigma\log\Bigl(\frac{c_0}{\epsilon}\Bigr).\label{eq:prod}
\end{align}
Now, we choose the standard allocation $\delta_j:=\frac{\epsilon}{2}\sqrt{K_j/\sum_{\ell\in S}K_\ell}$, which minimizes $\sum_j K_j\log(R_{j,m}/\delta_j)$ subject to $\sum_j\delta_j^2\le (\epsilon/2)^2$. Then
\[
\sum_{j\in S} K_j\log \left(\frac{C R_{j,m}}{\delta_j}\right)
= \Bigl(\sum_{j\in S}K_j\Bigr)\log\left(\frac{C \overline R_m \sqrt{\sum_{j\in S}K_j}}{\epsilon}\right)
\]
with $\overline R_m:=\max_{j\in S}R_{j,m}\le B_0+2^{m+1}M_\star r_N$. Using $k\log|\mathcal K|\lesssim (\log N)^2$, $\log\overline R_m\le C_1\log N + C_2 m$ (since $\epsilon\le r_N$), and $c\,\log(RN)\lesssim\log N+\log R$,
\begin{equation}\label{eq:prod2}
\log N(\epsilon,\Theta_{N,m}(S),\tilde d)
\ \le\ C_3(\log N)^2\;+\; C_4\bigl(\log N + m\bigr)\sum_{j\in S}K_j\;+\; C_5\log\Bigl(\tfrac{c_0}{\epsilon}\Bigr).
\end{equation}

\textbf{Choosing $K_j$ under the leaf budget.}
We take the dyadic choice
\[
K_j^\star\ :=\
\begin{cases}
\asymp \dfrac{N\,r_{jN,\mathrm{ad}}^2}{\log N}, & j\in S_0,\\[8pt]
1, & j\in S\setminus S_0,
\end{cases}
\quad\text{rounded up to }\ K_j\in\mathcal K.
\]
For any $x \ge 1,$ dyadic rounding satisfies $\lceil x\rceil_{\text{dyadic}} \le 2x.$ Thus, for $j \in S \cap S_0, \ K_j \le 2K_j ^{\star},$ and for $j \in S \setminus S_0, \ K_j=1 \le 2.$ 

Since $L_N=C_L\,N\bar r_N^2/\log N$ with $\bar r_N^2\ge \sum_{u\in S_0}r_{uN,\mathrm{ad}}^2$, this choice satisfies $\sum_{j\in S}K_j\le L_N$ for $C_L$ large enough. Consequently,
\[
\sum_{j\in S}K_j\ \le\ C\Bigl(\sum_{j\in S_0}\frac{N\,r_{jN,\mathrm{ad}}^2}{\log N}\ +\ s_N\Bigr).
\]
This can be observed by a dyadic-rounding bound:
\[
\sum_{j \in S} K_j \le 2 \sum_{j \in S \cap S_0} K_j ^{\star} + 2|S \setminus S_0| =\frac{2cN}{\log N} \sum_{j \in S \cap S_0} r_{jN,\textrm{ad}} ^2 + 2|S\setminus S_0|.
\]
Since $S \cap S_0 \subseteq S_0$ and $|S \setminus S_0| \le s_N,$ the inequality follows from \[
\sum_{j \in S} K_j \le \frac{2cN}{\log N} \sum_{j \in S_0} r_{jN,\textrm{ad}}^2 + 2s_N \le C \left(\sum_{j \in S_0} \frac{N r_{jN,\textrm{ad}}^2}{\log N} + s_N \right).
\]
Plugging into \eqref{eq:prod2} and using $s_N=\mathcal{O}(\log N)$ and $\log(c_0/\epsilon)\lesssim\log N$,
\begin{align}
\log N(\epsilon,\Theta_{N,m}(S),\tilde d)
&\le\ C_6(\log N)^2\;+\; C_7(\log N + m)\left(\sum_{j\in S_0}\frac{N\,r_{jN,\mathrm{ad}}^2}{\log N} + s_N\right)\nonumber\\
&\le\ C_8\,N\sum_{j\in S_0}r_{jN,\mathrm{ad}}^2\;+\; C_9\,(\log N + m)\,s_N\;+\;C_{10}(\log N)^2.\label{eq:fixedS}
\end{align}
(Note that the factor $m\,\sum_{j\in S_0}\frac{N r_{jN,\mathrm{ad}}^2}{\log N}$ is $\le C\,N\sum_{j\in S_0}r_{jN,\mathrm{ad}}^2$ for all $m\lesssim \log N$ and can thus be absorbed into the first term by enlarging $C_8$.)

Taking the supremum over $|S|=s_N$ in \eqref{eq:fixedS} and combining with \eqref{eq:model-select-layer}--\eqref{eq:T1},
\[
\log N(\epsilon,\Theta_{N,m},\tilde d)
\ \le\ c_1 s_N\log p_N\;+\; C_8\,N\sum_{j\in S_0}r_{jN,\mathrm{ad}}^2\;+\; C_{11}(\log N)^2\;+\; C_{12}\,m\,\log N.
\]
Since $N r_N^2=d_0\log p_N+N\sum_{j\in S_0}r_{jN,\mathrm{ad}}^2$, $d_0 = \mathcal{O}(\log N)$, and $s_N=\Theta(\log N)$, we have $s_N\log p_N\lesssim N r_N^2$ and $(\log N)^2\ll N r_N^2$ under \textbf{(A5)}. Therefore, for all large $N$ and $0<\epsilon\le r_N$,
\[
\log N \bigl(\epsilon,\Theta_{N,m}, d_n \bigr)\ \lesssim\ \log N \bigl(\epsilon, \Theta_{N,m}, \tilde d \bigr)
\le\ C_4\,N r_N^2+C_{\mathrm{sh}}\,m\,\log N,
\]
which is \eqref{eq:localized-entropy-bound}.
\end{proof}

\begin{lemma}[Uniformly exponentially consistent tests]\label{lem:tests}
Let $\mathcal F_N^\star$ be the sieve in Definition~\ref{def:sieve_revised}, and assume \textbf{(A1)}–\textbf{(A6)} and that $0<\sigma^2\le c_0$ on $\mathcal F_N^\star$.
Choose the sieve constant $C_s$ so that $2s_N\le C_{\mathrm{RE}}\,d_0$ for the constant $C_{\mathrm{RE}}$ from \textbf{(A6)} (this is possible since $s_N=\Theta(\log N)$ and $d_0=\mathcal{O}(\log N)$).
Then there exist constants $M_\star>0$ and $c_5>0$ and a sequence of tests $\varphi_N$ such that
\[
\mathbb{E}_{\theta_0}\varphi_N \;\le\; \exp\{-c_5\,N r_N^2\}
\quad \text{and} \quad
\sup_{\substack{\theta\in\mathcal F_N^\star:\\ \|\bm\beta-\bm\beta_0\|_N> M_\star r_N}}
\mathbb{E}_{\theta}(1-\varphi_N) \;\le\; \exp\{-c_5\,N r_N^2\}.
\]
\end{lemma}

\begin{proof}
% For $\theta=(\bm\beta,\sigma^2)$ write $\xi(\theta)=(\xi_i(\theta))_{i=1}^N$ with
% $\xi_i(\theta)=\beta_0(\bm Z_i)+\sum_{j=1}^{p_N}\beta_j(\bm Z_i)X_{ij}$ and
% $\Delta(\theta,\theta'):=\xi(\theta)-\xi(\theta')$.
By \textbf{(A6)} (functional RE), for any $\theta,\theta'$ whose predictor support union has size $\le C_{\mathrm{RE}}d_0$,
\begin{equation}\label{eq:tests-RE}
c_{\min}\,N\,\|\bm\beta-\bm\beta'\|_N^2 \;\le\; \|\Delta(\theta,\theta')\|_2^2 \;\le\; c_{\max}\,N\,\|\bm\beta-\bm\beta'\|_N^2.
\end{equation}
Because any $\theta\in\mathcal F_N^\star$ has at most $s_N$ active coordinates and $2s_N\le C_{\mathrm{RE}}d_0$ by assumption, \eqref{eq:tests-RE} applies to all pairs in the sieve.

Let
\[
\Theta_N \;:=\; \Bigl\{\theta\in\mathcal F_N^\star:\ \|\bm\beta-\bm\beta_0\|_N> M_\star r_N\Bigr\},
\]
where $M_\star>0$ will be fixed below.

\medskip
\textbf{Covering by shells in the prediction semimetric.}
Equip the parameter space with the prediction semimetric
\[
d_n\bigl((\bm\beta,\sigma^2),(\bm\beta',{\sigma'}^{2})\bigr)
:= \frac{1}{\sqrt N}\,\big\|\Delta(\theta,\theta')\big\|_2 \;+\; \big|\sigma^2-{\sigma'}^{2}\big|.
\]
Decompose $\Theta_N$ into shells (as in\Cref{lem:entropy}):
\[
\Theta_{N,m}\ :=\ \Bigl\{(\bm\beta,\sigma^2)\in\mathcal F_N^\star:\ 
 2^{m} M_\star r_N < \|\bm\beta-\bm\beta_0\|_{N}\ \le\ 2^{m+1} M_\star r_N,\ 
 0<\sigma^2\le c_0\Bigr\},\quad m=0,1,2,\dots
\]
Fix $\gamma:=r_N/8 \le r_N$. For each $m\ge 0$, let $\{\theta_{m,k}\}_{k=1}^{K_m}$ be a $\gamma$-net of $\Theta_{N,m}$ in the metric $d_n$.
By \Cref{lem:entropy},
\begin{equation}\label{eq:tests-Km}
\log K_m \;\le\; C_4\,N\,r_N^2\;+\;C_{\mathrm{sh}}\,m\,\log N
\qquad\text{for all }m\ge 0\text{ and large }N.
\end{equation}

\medskip
\textbf{Pointwise tests at the net points.}
For each net point $\theta_{m,k}=(\bm\beta_{m,k},\sigma_{m,k}^2)$ define
\[
\Delta_{m,k}:=\xi(\theta_{m,k})-\xi(\theta_0),\qquad
\phi_{m,k}(\bm{Y})\ :=\ \mathbbm{1}\!\left\{
\Big\langle \bm{Y}-\frac{\xi(\theta_0)+\xi(\theta_{m,k})}{2},\, \Delta_{m,k}\Big\rangle \ge 0
\right\}.
\]
Under $H_0:\ \bm{Y}\sim \mathcal N(\xi(\theta_0),\,\sigma_0^2 I_N)$ and $H_{1,m,k}:\ \bm{Y}\sim \mathcal N(\xi(\theta_{m,k}),\,\sigma_{m,k}^2 I_N),$ we can show
\[
\mathbb P_{\theta_0}(\phi_{m,k}=1)=\Phi\Big(-\tfrac{\|\Delta_{m,k}\|_2}{2\sigma_0}\Big)
\quad\text{and}\quad
\mathbb P_{\theta_{m,k}}(\phi_{m,k}=0)=\Phi\Big(-\tfrac{\|\Delta_{m,k}\|_2}{2\sigma_{m,k}}\Big),
\]
so, by the Chernoff bound and $0<\sigma_0^2,\sigma_{m,k}^2\le c_0$,
\begin{equation}\label{eq:tests-NP}
\mathbb E_{\theta_0}\phi_{m,k}\ \le\ \exp\!\Big(-\tfrac{\|\Delta_{m,k}\|_2^2}{8c_0}\Big),
\qquad
\mathbb E_{\theta_{m,k}}(1-\phi_{m,k})\ \le\ \exp\!\Big(-\tfrac{\|\Delta_{m,k}\|_2^2}{8c_0}\Big).
\end{equation}
Define the composite test
\[
\varphi_N(\bm{Y})\ :=\ \max_{m\ge 0}\ \max_{1\le k\le K_m}\ \phi_{m,k}(\bm{Y}).
\]

\medskip
\textbf{Type I error.}
By the union bound and \eqref{eq:tests-NP},
\[
\mathbb E_{\theta_0}\varphi_N \ \le\ \sum_{m\ge 0}\ \sum_{k=1}^{K_m}\ \exp\!\Big(-\tfrac{\|\Delta_{m,k}\|_2^2}{8c_0}\Big).
\]
We fix $m$. Any $\theta_{m,k}\in\Theta_{N,m}$ satisfies $\|\bm\beta_{m,k}-\bm\beta_0\|_N> 2^m M_\star r_N$. By \eqref{eq:tests-RE},
\[
\|\Delta_{m,k}\|_2^2\ \ge\ c_{\min} N (2^m M_\star r_N)^2.
\]
Therefore,
\begin{equation*}
\begin{split}
\mathbb E_{\theta_0}\varphi_N & \le \sum_{m \ge 0}\sum_{k=1}^{K_m} \mathbb{E}_{\theta_0} \phi_{m,k} \\
& \le K_m \exp \left(-\frac{c_{\min}}{8c_0} 4^m M_{\star} ^2 N r_N ^2 \right)\\
& =\sum_{m\ge 0}\ \exp\!\Big\{\,\log K_m - \frac{c_{\min}}{8c_0}\,4^m M_\star^2\,N r_N^2\Big\}.
\end{split}
\end{equation*}
Applying \eqref{eq:tests-Km},
\[
\mathbb E_{\theta_0}\varphi_N
\ \le\ \sum_{m\ge 0}\ \exp\!\Big\{\,C_4 N r_N^2 + C_{\mathrm{sh}}\,m\log N - \frac{c_{\min}}{8c_0}\,4^m M_\star^2\,N r_N^2\Big\}.
\]
Pick $M_\star$ large (depending only on $c_{\min},c_0,C_4$) so that
\(
({c_{\min}}/{8c_0})M_\star^2 \ \ge\ C_4 + 2c_5
\)
for some $c_5>0$. Since $4^m\ge 1$ and, for $m\ge 1$, $4^m N r_N^2 \gg m\log N$ (because $N r_N^2 \gtrsim (\log N)\log p_N$ by \textbf{(A5)} and $\log p_N\to\infty$), the $m\ge 1$ terms are exponentially dominated. Thus
\[
\mathbb E_{\theta_0}\varphi_N
\ \le\ \exp\!\big\{-(C_4+2c_5-C_4)N r_N^2\big\}\ +\ \sum_{m\ge 1} \exp\!\big\{-c\,4^m N r_N^2\big\}
\ \le\ \exp\{-c_5\,N r_N^2\}
\]
for large $N$ (absorbing the tail sum into the constant), which proves the desired Type I bound.

\medskip
\textbf{Type II error.}
Fix any $\theta\in\Theta_N$ and let $m\ge 0$ be the shell index such that
$2^{m}M_\star r_N<\|\bm\beta-\bm\beta_0\|_N\le 2^{m+1}M_\star r_N$.
Choose $k$ so that $d_n(\theta,\theta_{m,k})\le \gamma=r_N/8$.
By definition of $d_n$,
\begin{equation}\label{eq:tests-close}
\frac{1}{\sqrt N}\,\|\Delta(\theta,\theta_{m,k})\|_2 \ \le\ \gamma
\qquad\Longrightarrow\qquad
\|\Delta(\theta,\theta_{m,k})\|_2 \ \le\ \gamma\,\sqrt N.
\end{equation}
Furthermore, by \eqref{eq:tests-RE} and the choice of $m$,
\[
\|\Delta(\theta,\theta_0)\|_2 \ \ge\ \sqrt{c_{\min}}\,2^{m}M_\star r_N \sqrt N.
\]
Therefore, by the triangle inequality and \eqref{eq:tests-close},
\[
\|\Delta_{m,k}\|_2
\ =\ \|\Delta(\theta_{m,k},\theta_0)\|_2
\ \ge\ \|\Delta(\theta,\theta_0)\|_2 - \|\Delta(\theta,\theta_{m,k})\|_2
\ \ge\ \Big(\sqrt{c_{\min}}\,2^{m}M_\star - \tfrac18\Big)\,r_N \sqrt N.
\]
Plugging this into \eqref{eq:tests-NP} gives
\[
\mathbb E_{\theta}(1-\varphi_N)
\ \le\ \exp\!\Big(-\tfrac{\|\Delta_{m,k}\|_2^2}{8c_0}\Big)
\ \le\ \exp\!\Big(- \frac{1}{8c_0}\,\big(\sqrt{c_{\min}}\,2^{m}M_\star - \tfrac18\big)^2\,N r_N^2\Big).
\]
Enlarging $M_\star$ if necessary so that $\frac{1}{8c_0}\,(\sqrt{c_{\min}}\,M_\star - 1/8)^2\ge c_5$, and using $2^m\ge 1$, we obtain
\[
\sup_{\theta\in\Theta_N}\ \mathbb E_{\theta}(1-\varphi_N) \ \le\ \exp\{-c_5\,N r_N^2\}.
\]

Combining the Type I and Type II bounds completes the proof.
\end{proof}

\subsection{Proof of near-minimax optimality of the contraction rate}
\label{sec:minimax}
Our near-minimax lower bound follows the standard information–theoretic template used for high-dimensional nonparametric models: we build a large, well-separated finite subset of the parameter space and control the pairwise Kullback–Leibler (KL) divergences so that Fano and Assouad inequalities apply. This strategy goes back to classical treatments in \citet{Tsybakov2009} and, in the sparse additive setting, is carried out in closely related forms by \citet{Koltchinskii2010} and \citet{raskutti2009minimaxratesestimationhighdimensional}. Following a similar style of proofs as in these works, the lower bound naturally splits into a selection piece, reflecting the combinatorial difficulty of identifying an active subset among $p_N$ candidates (giving a $d_0 \log(p_N/d_0)/N$ term), and a nonparametric piece, reflecting the functional difficulty of estimating $s-$dimensional $\alpha-$H\"older components (giving $N^{-2\alpha/(2 \alpha+s)}$ terms, possibly heterogeneous across coordinates). 

We need the following lemma to show minimax optimality of the rate $r_N$. It is the key local building block that lets us realize, inside the H\"older class, tiny ``on/off" perturbations supported on $s_{0,j}-$ dimensional cubes of side $h$, with amplitude $a \asymp h^{\alpha_j}$, such that their empirical $L_2(\mathbb{P}_N)$ energy scales like $a^2 h^{s_{0,j}}.$ This guarantees separation -- neighboring hypotheses that differ by flipping a single bump are separated by $\|f_{j,1}-f_{j,0}\|_{N} ^2 \asymp h^{2 \alpha_j+s_{0,j}}$ and small KL.
\begin{lemma}[Modifier--sparse bump construction under k-d regularity]\label{lem:bump}
Fix an active predictor $j$ with modifier set $S_{0,j}\subset\{1,\dots,R\}$ of size $s:=s_{0,j}\in\{1,\dots,R\}$ and smoothness $0<\alpha_j\le 1$. Assume {\textbf{(A2)}} (k-d regularity) holds on the design $\{\bm Z_i\}_{i=1}^N\subset[0,1]^R$.
There exist constants $c,C\in(0,\infty)$ (independent of $N$) and, for any bandwidth $h\in(0,h_0]$ with $h_0$ small enough, one can construct a pair of functions
\[
f_{j,0}\equiv 0,
\quad
f_{j,1}(\bm z)= a\,\varphi \Bigl(\frac{\bm z_{S_{0,j}}-x}{h}\Bigr),
\quad \bm z\in[0,1]^R,
\]
with the following properties:
\begin{enumerate}
\item[(i)] $\varphi\in C^\infty_c([0,1]^s)$\footnote{$C^{\infty}_c([0,1]^s)$ denotes the class of infinitely differentiable functions on $\mathbb{R}^s$ whose support lies in $(0,1)^s$, considered on $[0,1]^s$}, $\|\varphi\|_\infty\le 1$, and its H\"older seminorms up to order $\alpha_j$ are bounded by a constant depending only on $\alpha_j,s$.
\item[(ii)] For the amplitude $a:=c\,h^{\alpha_j}$, one has $f_{j,1}\in\mathcal H^{\alpha_j}_B([0,1]^R)$, with H\"older ball radius $B$ independent of $N$ (and of $h$ for $h\le h_0$).
\item[(iii)] There exists a center $x\in(0,1)^s$ (depending on the design and $h$) such that
\[
c\,a^2 h^{s}\ \le\ \|f_{j,1}-f_{j,0}\|_{N}^2\ \le\ C\,a^2 h^{s}.
\]
\end{enumerate}
Consequently, choosing $h\asymp N^{-1/(2\alpha_j+s)}$ yields
\[
\|f_{j,1}-f_{j,0}\|_{N}^2\ \asymp\ N^{-\frac{2\alpha_j}{2\alpha_j+s}}.
\]
\end{lemma}

\begin{proof}
\textbf{Proving (i).} Let $\psi\in C^\infty_c((-1,1))$ be a standard bump with $\psi\ge 0$, $\psi(0)=1$, and $\|\psi\|_\infty=1$. A standard one-dimensional $C^{\infty}$ bump can be constructed as \[
\psi_0(t) = \begin{cases}
\exp \left(-\frac{1}{1-t^2} \right), & \ |t|<1 \\
0, & \ |t| \ge 1.
\end{cases}
\]
Note that $\textrm{supp}(\psi_0)=[-1,1],$ we can normalize to have a peak at $1$: $\psi(t) = \psi_0(t)/\psi_0(0) = e \psi_0(t),$ so $\psi(0)=1$ and $\|\psi\|_{\infty}=1.$ We set
\(
\varphi(u_1,\dots,u_s)=\prod_{\ell=1}^s \psi(u_\ell),
\)
then $\varphi \in C^{\infty} _c (\mathbb{R}^s)$ with $\textrm{supp}(\varphi) \subset [-1,1]^s,$ so after translation, $\varphi \in C^{\infty} _{\infty}((0,1)^s)$. Trivially, it also satisfies $\|\varphi\|_{\infty} \le 1.$

Because $\psi$ is smooth with bounded derivatives on a compact support, all partial derivatives of $\varphi$ are bounded. For $0<\alpha \le 1$ the $\alpha-$ H\"older semi-norm of $\varphi$ is finite and bounded by a constant depending only on $s$ and bounds on $\psi':$
\[
|\varphi(\bm{u})-\varphi(\bm{v})| \le \| \nabla \varphi \|_{\infty} \| \bm{u}-\bm{v} \|_{2} \ \implies [\varphi]_{C^{0,\alpha}} \le \|\nabla \varphi\|_{\infty} \text{diam}([-1,1]^s)^{1-\alpha}.
\]
Thus $\varphi$ satisfies (i) with constants depending only on $s$ and $\alpha_j$.

\textbf{Proving (ii).} Because $\varphi$ has a compact support contained in $(0,1)^s,$ the function $\bz \mapsto \varphi((\bz_{S_{0,j}}-x)/h)$ is supported in the cube $\{\bz : \|\bz_{S_{0,j}}-x\|_{\infty} < h\} \subset (0,1)^s,$ which holds since $x \in (h,1-h)^s.$ Using $\|\varphi\|_{\infty} \le 1,$
\[
\|f_{j,1}^{(x)}\|_{\infty} = \sup_{\bz \in [0,1]^R} a| \varphi((\bz_{S_{0,j}}-x)/h)| \le a = ch^{\alpha_j} \le ch_0 ^{\alpha_j}.
\]
Thus, if we choose $c \le B/h_0 ^{\alpha_j},$ then $\|f_{j,1} ^{(x)}\|_{\infty} \le B.$ Now, for any $\bz,\bz' \in [0,1]^R,$ we set $\bm{u}=(\bz_{S_{0,j}}-x)/h$ and $\bz'_{S_{0,j}}-x)/h.$ Then, 
\begin{align*}
    \left|f_{j,1}^{(x)}(\bz)-f_{j,1} ^{(x)}(\bz')\right| & = a | \varphi(\bm{u})-\varphi(\bm{v})| \\
    & \le a K_{\alpha_j} \|\bm{u}-\bm{v}\|_{2} ^{\alpha_j} \\
    & = a K_{\alpha_j} h^{-\alpha_j} \|\bz_{S_{0,j}} - \bz' _{S_{0,j}} \|_2 ^{\alpha_j}.
\end{align*}
Since $\|\bz_{S_{0.j}}-\bz'_{S_{0,j}}\|_2 \le \| \bz-\bz'\|_2,$ 
\[
\frac{\left| f_{j,1} ^{(x)} (\bz)-f_{j,1}^{(x)} (\bz') \right|}{\|\bz-\bz'\|_2 ^{\alpha_j}} \le a K_{\alpha_j} h^{-\alpha_j} = c K_{\alpha_j}.
\]
Taking the supremum over $\bz \neq \bz'$ yields $\left[f_{j,1} ^{(x)} \right]_{C^{0,\alpha_j}} \le c K_{\alpha_j}.$ So, if we also choose $c \le B/K_{\alpha_j},$ we obtain $\left[f_{j,1}^{(x)}\right]_{C^{0,\alpha_j}} \le B.$

Thus, for any $c \le \min \{B/h_0 ^{\alpha_j},B/K_{\alpha_j}\}$ we can ensure $\|f_{j,1} ^{(x)}\|_{\infty} \le B$ and $\left[f_{j,1} ^{(x)}\right]_{C^{0,\alpha_j}} \le B.$ This proves (ii).
% For any $h\in(0,h_0]$, choose a grid of $s$-dimensional cubes of side length $h$, i.e. $G_h=\{x\in(h,1-h)^s: x\ \text{on a regular grid of mesh }h\}$. For each $x\in G_h$ define
% \[
% f_{j,1}^{(x)}(\bm z)=a\,\varphi\Bigl(\frac{\bm z_{S_{0,j}}-x}{h}\Bigr),
% \quad a=c\,h^{\alpha_j},
% \]
% with $c>0$ chosen so that the H\"older seminorms of $f_{j,1}^{(x)}$ of order up to $\alpha_j$ are bounded by $B$ (this is standard: derivatives of $\varphi((\cdot-x)/h)$ scale like $h^{-m}$, and multiplying by $a=h^{\alpha_j}$ keeps all H\"older seminorms bounded when $m\le \alpha_j$; see e.g. classical scaling arguments for H\"older balls). This gives (ii).

\textbf{Proving (iii).} It remains to show the empirical $L_2$ energy bound (iii). By k-d regularity \textbf{(A2)}, there exist constants $0<c_\mathrm{kd}\le C_\mathrm{kd}<\infty$ such that for any axis-aligned rectangle $A\subset[0,1]^R$,
\[
c_\mathrm{kd}\,\mathrm{vol}(A)\ \le\ \frac{1}{N}\#\{i:\ \bm Z_i\in A\}\ \le\ C_\mathrm{kd}\,\mathrm{vol}(A).
\]
For our $f_{j,1}^{(x)}$, the support is $\Omega_x:=\{ \bm z:\ \|\bm z_{S_{0,j}}-x\|_\infty\le h\}\times [0,1]^{R-s}$, with $\mathrm{vol}(\Omega_x)= (2h)^s$. Hence
\[
\|f_{j,1}^{(x)}\|_N^2
=\frac{1}{N}\sum_{i=1}^N a^2\,\varphi^2\left(\frac{\bm Z_{i,S_{0,j}}-x}{h}\right)
= a^2 \cdot \frac{1}{N}\sum_{i:\ \bm Z_{i,S_{0,j}}\in x+[-h,h]^s}\ \varphi^2\left(\frac{\bm Z_{i,S_{0,j}}-x}{h}\right).
\]
By the boundedness of $\varphi^2$ above by $1$, and its positivity on a smaller cube (say on $[-1/2,1/2]^s$), one has
\[
c'\,a^2\cdot \frac{1}{N} \#\{i:\ \bm Z_{i,S_{0,j}}\in [x-h/2,x+h/2]^s\}\ \le\ \|f_{j,1}^{(x)}\|_N^2\ \le\ C'\,a^2\cdot \frac{1}{N} \#\{i:\ \bm Z_{i,S_{0,j}}\in [x-h,x+h]^s\},
\]
for some $0<c'\le C'<\infty$. Applying k-d regularity with the corresponding rectangles $A_x$ gives
\[
c\,a^2 h^s\ \le\ \|f_{j,1}^{(x)}\|_N^2\ \le\ C\,a^2 h^s,
\]
uniformly over $x\in \{x\in(h,1-h)^s: x\ \text{on a regular grid of mesh }h\}$, proving (iii). Taking $x$ to be any grid point completes the construction. Finally, setting $h\asymp N^{-\frac{1}{2\alpha_j+s}}$ yields the displayed rate.
\end{proof}

\paragraph{Proof overview of the near-minimax theorem.} 
We now move on to prove our near-minimaxity result. To that endeavor, we first split the domain into two disjoint regions via smooth windows: one used purely for selection and the other for nonparametric difficulty. On the selection region, we assign constant signals of size $a_N$ to a Hamming-separated family of supports $T$ via a Gilbert-Varshamov code. With $a_N ^2 \asymp \log |\mathcal{V}|/[N(d_0-m)],$ the pairwise KL's are small while the $L_2$ separation is of order $(d_0-m)a_N ^2,$ and Fano's inequality yields the $d_0 \log (p_N/d_0)/N$ contribution. On the nonparametric region, for a fixed subset $S^{\star}$ of $m$ active predictors we place multi-bump codes: each predictor $j \in S^{\star}$ carries a grid of disjoint $s_{0,j}-$ dimensional bumps at bandwidth $h_j \asymp N^{-1/(2 \alpha_j +s_{0,j})};$ toggling a single bump changes the mean by size $\asymp N^{-1/2}$ with uniformly small KL. Assouad's lemma over this hypercube then yields the remaining part of the minimaxity bound.
\begin{proof}[Proof of Theorem 2]
Let $\xi(\theta)=(\xi_i(\theta))_{i=1}^N$ be the mean vector for parameter $\theta$ and
$\Delta(\theta,\theta'):=\xi(\theta)-\xi(\theta')$.
By {\textbf{(A6)}} (functional RE), there exist $0<c_{\min}\le c_{\max}<\infty$ such that,
for any arrays $u_{ij}$ indexed by a set $S$ with $|S|\le C_{\mathrm{RE}}\,d_0$,
\begin{equation}\label{eq:RE-global}
c_{\min}\sum_{j\in S}\Big(\frac{1}{N}\sum_{i=1}^N u_{ij}(\bm Z_i)^2\Big)
\ \le\
\frac{1}{N}\sum_{i=1}^N\Big(\sum_{j\in S} X_{ij}\,u_{ij}(\bm Z_i)\Big)^2
\ \le\
c_{\max}\sum_{j\in S}\Big(\frac{1}{N}\sum_{i=1}^N u_{ij}(\bm Z_i)^2\Big).
\end{equation}
Write the empirical coefficient-function norm as
$\|f\|_N^2:=N^{-1}\sum_{i=1}^N f(\bm Z_i)^2$ and the additive loss
$L(\widehat{\bm\beta},\bm\beta):=\sum_{j=1}^{p_N}\|\widehat\beta_j-\beta_j\|_N^2$.
By \eqref{eq:RE-global}, for any $\theta,\theta'$ supported on $\le C_{\mathrm{RE}}\,d_0$ coordinates,
\begin{equation}\label{eq:RE-link-loss}
\frac{1}{N}\big\|\Delta(\theta,\theta')\big\|_2^2
\ge c_{\min} \sum_{j \in S}\frac{1}{N} \sum_{i=1}^{N}u_{ij}(\bm{Z}_i)^2 = c_{\min} \sum_{j \in S} \|\beta_j-\beta_J ' \|_N ^2 = c_{\min}\,L(\bm\beta,\bm\beta'),
\end{equation}
and analogously, for the Gaussian model with noise variance $\sigma^2\in[\underline\sigma^2,\overline\sigma^2]$,
\begin{equation}\label{eq:RE-link-KL}
\mathrm{KL}\!\left(P_\theta,P_{\theta'}\right)
=\frac{1}{2\sigma^2}\|\Delta(\theta,\theta')\|_2^2
\ \le\ \frac{c_{\max}}{2\underline\sigma^2}\,N\,L(\bm\beta,\bm\beta').
\end{equation}

\paragraph{Smooth domain split.}
Choose two disjoint axis-aligned rectangles
\[
\Omega_{\mathrm{sel}}^\circ:=[0,1/3]^R,\qquad
\Omega_{\mathrm{np}}^\circ:=[2/3,1]^R,
\]
and smooth cutoffs $\omega_{\mathrm{sel}},\omega_{\mathrm{np}}\in C_c^\infty([0,1]^R)$ such that
both are bounded within $[0,1]$, $\omega_{\mathrm{sel}}=1$ on $\Omega_{\mathrm{sel}}^\circ$, 
$\omega_{\mathrm{np}}=1$ on $\Omega_{\mathrm{np}}^\circ$, and $\mathrm{supp}(\omega_{\mathrm{sel}})$, $\mathrm{supp}(\omega_{\mathrm{np}})$ are disjoint.
By k--d regularity {\textbf{(A2)}}, there exist constants $0<c_\Omega\le C_\Omega<\infty$ such that
\begin{equation}\label{eq:mass-rectangles-smooth}
c_\Omega\ \le\ \frac{1}{N}\sum_{i=1}^N \omega_{\mathrm{sel}}(\bm Z_i)^2\ \le\ C_\Omega,
\qquad
c_\Omega\ \le\ \frac{1}{N}\sum_{i=1}^N \omega_{\mathrm{np}}(\bm Z_i)^2\ \le\ C_\Omega,
\end{equation}
uniformly in $N$. These windows keep H\"older smoothness while making the two regions disjoint -- they are exactly $1$ on the inner rectangles $\Omega_{\mathrm{sel}} ^{\circ}$ and $\Omega_{\mathrm{np}} ^{\circ},$ while smoothly tapering to $0$ outside these rectangles.

By \Cref{lem:bump}, for each active predictor $j$ with modifier set $S_{0,j}$ and smoothness $\alpha_j\in(0,1]$, there exist compactly supported $C^\infty$ bumps supported in small $s_{0,j}$-cubes, with amplitude $a_{j}\asymp h_j^{\alpha_j}$ at bandwidth $h_j$, such that
\begin{equation}\label{eq:bump-energy}
c\,a_{j}^2 h_j^{s_{0,j}}\ \le\ \|f_{j,1}\|_N^2\ \le\ C\,a_{j}^2 h_j^{s_{0,j}},
\end{equation}
and $f_{j,1}\in\mathcal H^{\alpha_j}_B([0,1]^R)$ with $B$ independent of $N$ (for $h_j$ small).

\paragraph{Product packing family.}
Fix $m:=\lfloor d_0/2\rfloor$ and a fixed index set $S^\star\subset\{1,\dots,p_N\}$ with $|S^\star|=m$.
Let $U^\star:=\{1,\dots,p_N\}\setminus S^\star$.
We encode \emph{selection} difficulty on $\Omega_{\mathrm{sel}}$ by choosing $T\subset U^\star$ with $|T|=d_0-m$,
and \emph{nonparametric} (NP) difficulty on $\Omega_{\mathrm{np}}$ using \emph{multi-bump} codes for the $m$ predictors in $S^\star$.

\smallskip
\emph{Selection code.}
By the Gilbert--Varshamov bound, there exists a finite set
$\mathcal V\subset\{T\subset U^\star: |T|=d_0-m\}$ with pairwise Hamming distance
$|T\triangle T'|\ge (d_0-m)/2$ and
\begin{equation}\label{eq:GV-size}
\log|\mathcal V|\ \ge\ c_1\,(d_0-m)\,\log\!\Big(\frac{p_N-m}{d_0-m}\Big)
\ \ge\ c_1'\,d_0\,\log\!\Big(\frac{p_N}{d_0}\Big).
\end{equation}

\smallskip
\emph{NP code.}
For each $j\in S^\star$, fix a bandwidth $h_j\in(0,h_0]$ (to be chosen later) and build a regular grid
$\mathcal G_{j,h_j}$ of centers inside $\Omega_{\mathrm{np}}^\circ$ so that the cubes
$x+[ -h_j, h_j]^{s_{0,j}}$ for $x\in\mathcal G_{j,h_j}$ are pairwise disjoint and all contained in $\Omega_{\mathrm{np}}^\circ$.
Then $M_j:=|\mathcal G_{j,h_j}|\asymp h_j^{-s_{0,j}}$, which denotes the packing number on an $s_{0,j}-$ dimensional cube.
For each $x\in\mathcal G_{j,h_j}$, let
\[
g_{j,x}(\bm z):= a_j\,\varphi \Bigl(\frac{\bm z_{S_{0,j}}-x}{h_j}\Bigr),
\quad a_j:=\kappa_j\,h_j^{\alpha_j},
\]
where $\varphi$ is the bump from \Cref{lem:bump} and the constant $\kappa_j\in(0,1]$ (independent of $N$) will be fixed small enough below to control neighbor Kullback divergences. For a bit-vector $\theta_j=(\theta_{j,x})_{x\in\mathcal G_{j,h_j}}\in\{0,1\}^{M_j}$, define
\[
f_{j,\theta_j}(\bm z) \ :=\ \sum_{x\in\mathcal G_{j,h_j}} \theta_{j,x}\,g_{j,x}(\bm z).
\]
By disjoint supports and \Cref{lem:bump}, $f_{j,\theta_j}\in \mathcal H^{\alpha_j}_B$ with $B$ independent of $N$, and
\begin{equation}\label{eq:additivity-energy}
\|f_{j,\theta_j}\|_{N,\Omega_{\mathrm{np}}}^2
= \sum_{x\in\mathcal G_{j,h_j}} \theta_{j,x}\,\|g_{j,x}\|_{N,\Omega_{\mathrm{np}}}^2,
\qquad
c\,a_j^2 h_j^{s_{0,j}}\le \|g_{j,x}\|_{N,\Omega_{\mathrm{np}}}^2 \le C\,a_j^2 h_j^{s_{0,j}}.
\end{equation}

\smallskip
\emph{Parameters indexed by the product code.}
For each $(T,\theta)$ with $T\in\mathcal V$ and $\theta=(\theta_j)_{j\in S^\star}$, $\theta_j\in\{0,1\}^{M_j}$, define
\[
\beta_j^{(T,\theta)}(\bm z)\ :=
\begin{cases}
a_N\,\omega_{\mathrm{sel}}(\bm z), & j\in T,\\[3pt]
f_{j,\theta_j}(\bm z)\,\omega_{\mathrm{np}}(\bm z), & j\in S^\star,\\[3pt]
0, & j\notin S^\star\cup T,
\end{cases}
\qquad \beta_0^{(T,\theta)}\equiv 0.
\]
Because $\omega_{\mathrm{sel}},\omega_{\mathrm{np}}$ are smooth with disjoint supports, each $\beta_j^{(T,\theta)}$ is in the Hölder class of {\textbf{(A4)}}, and the active set is $S^\star\cup T$ of size $m+(d_0-m)=d_0$.

\paragraph{KL decomposition for the product family.}
Define the losses
\[
L_{\mathrm{sel}}(\widehat{\bm\beta};T):=\sum_{j=1}^{p_N}\|\widehat\beta_j-\beta_j^{(T,\theta)}\|_{N,\mathrm{sel}}^2,\qquad
L_{\mathrm{np}}(\widehat{\bm\beta};\theta):=\sum_{j=1}^{p_N}\|\widehat\beta_j-\beta_j^{(T,\theta)}\|_{N,\mathrm{np}}^2,
\]
where $\|f\|_{N,\mathrm{sel}}^2:=N^{-1}\sum_i f(\bm Z_i)^2\,\omega_{\mathrm{sel}}(\bm Z_i)^2$ and analogously for $\mathrm{np}$. The two depend on $T$ and $\theta$ separately by construction, and
\begin{equation}\label{eq:loss-sum}
L(\widehat{\bm\beta},\bm\beta^{(T,\theta)})=L_{\mathrm{sel}}(\widehat{\bm\beta};T)+L_{\mathrm{np}}(\widehat{\bm\beta};\theta).
\end{equation}
By \eqref{eq:RE-link-loss},
\begin{equation}\label{eq:loss-lower-Delta}
\frac{1}{N}\,\big\|\Delta\big((T,\theta),\widehat\theta\big)\big\|_2^2
\ \ge\ c_{\min}\,\Big( L_{\mathrm{sel}}(\widehat{\bm\beta};T)+L_{\mathrm{np}}(\widehat{\bm\beta};\theta)\Big).
\end{equation}
Moreover, by \eqref{eq:RE-link-KL} and the disjoint supports of $\omega_{\mathrm{sel}}$ and $\omega_{\mathrm{np}}$,
\begin{align}
\mathrm{KL}\big(P_{(T,\theta)},P_{(T',\theta')}\big)
&\le \frac{c_{\max}}{2\underline\sigma^2}\,N\,
\Big( L_{\mathrm{sel}}(\bm\beta^{(T,\cdot)},\bm\beta^{(T',\cdot)})
     +L_{\mathrm{np}}(\bm\beta^{(\cdot,\theta)},\bm\beta^{(\cdot,\theta')})\Big).
\label{eq:KL-sum}
\end{align}
As the suprema decouple across $T$ and $\theta$,
\begin{equation}\label{eq:inf-sup-sum}
\inf_{\widehat{\bm\beta}}\ \sup_{(T,\theta)} \E L(\widehat{\bm\beta},\bm\beta^{(T,\theta)})
\ \ge\
\inf_{\widehat{\bm\beta}} \sup_{T\in\mathcal V}\E L_{\mathrm{sel}}(\widehat{\bm\beta};T)
\ +\
\inf_{\widehat{\bm\beta}} \sup_{\theta}\E L_{\mathrm{np}}(\widehat{\bm\beta};\theta).
\end{equation}

\paragraph{Selection term via Fano.}
Fix any $\bar\theta$. For $T\neq T'$,
\[
L_{\mathrm{sel}}\big(\bm\beta^{(T,\bar\theta)},\bm\beta^{(T',\bar\theta)}\big)
= |T\triangle T'|\cdot a_N^2\cdot \frac{1}{N}\sum_{i=1}^N \omega_{\mathrm{sel}}(\bm Z_i)^2
\ \ge\ c\, (d_0-m)\,a_N^2,
\]
by $|T\triangle T'|\ge (d_0-m)/2$ and \eqref{eq:mass-rectangles-smooth}. Meanwhile, by \eqref{eq:KL-sum},
\[
\mathrm{KL}\big(P_{(T,\bar\theta)},P_{(T',\bar\theta)}\big)\ \le\ 
\frac{c_{\max}}{2\underline\sigma^2}\,N\, L_{\mathrm{sel}}\big(\bm\beta^{(T,\bar\theta)},\bm\beta^{(T',\bar\theta)}\big)
\ \le\ C\,N\,(d_0-m)\,a_N^2.
\]
Choose
\begin{equation}\label{eq:aN-choice}
a_N^2\ :=\ \frac{\kappa}{C}\cdot \frac{\log|\mathcal V|}{N\,(d_0-m)}\qquad\text{with }\ \kappa\in(0,1/8).
\end{equation}
Then the average pairwise KL is $\le \kappa \log|\mathcal V|$.
By Fano’s lemma (e.g., \citealt[Th.\ 2.11]{Tsybakov2009}) and \eqref{eq:loss-lower-Delta},
\[
\inf_{\widehat{\bm\beta}} \sup_{T\in\mathcal V}\ \E L_{\mathrm{sel}}(\widehat{\bm\beta};T)
\ \ge\ c\cdot \frac{\log|\mathcal V|}{N}
\ \ge\ c\,\frac{d_0\log(p_N/d_0)}{N},
\]
using \eqref{eq:GV-size}.

\paragraph{Nonparametric term via Assouad.}
Fix $\bar T\in\mathcal V$. Consider the hypercube indexed by all bump bits
\[
\Theta:=\prod_{j\in S^\star}\{0,1\}^{M_j},
\qquad D:=\sum_{j\in S^\star}M_j.
\]
For neighbors $\theta$ and $\theta^{(j,x)}$ that differ only in bit $(j,x)$, we have \[
\beta_{j'} ^{(\bar{T},\theta)} - \beta_{j'} ^{(\bar{T},\theta^{(j,x)}} = \begin{cases}
    g_{j,x}, & \ j'=j, \\
    0, & \ j' \neq j.
\end{cases}
\]
Hence,
\begin{equation}\label{eq:percoord-sep-refined}
\rho_{j,x}^2\ :=\ L_{\mathrm{np}}\big(\bm\beta^{(\bar T,\theta)},\bm\beta^{(\bar T,\theta^{(j,x)})}\big)
=\| g_{j,x}\|_{N,\Omega_{\mathrm{np}}}^2
\ \asymp\ a_j^2 h_j^{s_{0,j}}
\ =  \kappa_j ^2 h_j^{2\alpha_j+s_{0,j}},
\end{equation}
by \eqref{eq:additivity-energy}. Moreover, by \eqref{eq:KL-sum},
\begin{equation}\label{eq:percoord-KL-refined}
\mathrm{KL}\big(P_{(\bar T,\theta)},P_{(\bar T,\theta^{(j,x)})}\big)
\ \le\ \frac{c_{\max}}{2\underline\sigma^2}\,N\,\rho_{j,x}^2
\ \lesssim\ N\,h_j^{2\alpha_j+s_{0,j}}.
\end{equation}
Choose
\begin{equation}\label{eq:hj-choice}
h_j \ :=\ c_j\,N^{-\frac{1}{2\alpha_j+s_{0,j}}},\qquad
a_j:=\kappa_j\,h_j^{\alpha_j},
\end{equation}
with $c_j>0$ and $\kappa_j\in(0,1]$ small enough so that the RHS of \eqref{eq:percoord-KL-refined} is bounded by a fixed $\kappa_0<(\log 2)/8$ for all $j$ and $N$. Then
\begin{equation}\label{eq:rho-square-final}
\rho_{j,x}^2\ \asymp\ h_j^{2\alpha_j+s_{0,j}}\ \asymp\ N^{-1},
\qquad
M_j\ \asymp\ h_j^{-s_{0,j}}\ \asymp\ N^{\frac{s_{0,j}}{2\alpha_j+s_{0,j}}}.
\end{equation}
By Assouad’s lemma for nonhomogeneous coordinates (see, e.g., \citealt[Th.\ 2.12]{Tsybakov2009}), the uniform neighbor KL bound implies that for any estimator,
\[
\sup_{\theta\in\Theta}\ \E L_{\mathrm{np}}(\widehat{\bm\beta};\theta)
\ \ge\ c\sum_{j\in S^\star}\sum_{x\in\mathcal G_{j,h_j}} \rho_{j,x}^2
\ \asymp\ \sum_{j\in S^\star} M_j\,h_j^{2\alpha_j+s_{0,j}}
\ \asymp\ \sum_{j\in S^\star} h_j^{2\alpha_j}
\ \asymp\ \sum_{j\in S^\star} N^{-\frac{2\alpha_j}{2\alpha_j+s_{0,j}}}.
\]
Hence
\begin{equation}\label{eq:np-final-refined}
\inf_{\widehat{\bm\beta}} \sup_{\theta\in\Theta}\ \E L_{\mathrm{np}}(\widehat{\bm\beta};\theta)
\ \ge\ c\,\sum_{j\in S^\star} N^{-\frac{2\alpha_{j}}{2\alpha_{j}+s_{0,j}}}.
\end{equation}
Since $|S^\star|=m=\lfloor d_0/2\rfloor$, by choosing $S^\star$ as the $m$ active indices with the largest
$N^{-\frac{2\alpha_{j}}{2\alpha_{j}+s_{0,j}}}$ among the $d_0$ active predictors,
\begin{equation}\label{eq:half-max} \sum_{j\in S^\star} N^{-\frac{2\alpha_{j}}{2\alpha_{j}+s_{0,j}}} \ge \frac{m}{d_0} \sum_{j \in S_0} N^{-\frac{2 \alpha_j}{2 \alpha_j + s_{0,j}}} \ge \left(\frac{1}{2}-\frac{1}{2 d_0} \right)\sum_{j\in S_0} N^{-\frac{2\alpha_{j}}{2\alpha_{j}+s_{0,j}}}.
\end{equation}

Combining \eqref{eq:inf-sup-sum}, the selection bound and \eqref{eq:np-final-refined}–\eqref{eq:half-max},
\[
\inf_{\widehat{\bm\beta}}\ \sup_{(T,\theta)} \E L(\widehat{\bm\beta},\bm\beta^{(T,\theta)})
\ \ge\ c\Bigg\{\frac{d_0\log(p_N/d_0)}{N}
+\sum_{j\in S_0} N^{-\frac{2\alpha_{j}}{2\alpha_{j}+s_{0,j}}}\Bigg\}.
\]
Finally, by \eqref{eq:RE-link-loss},
\[
\inf_{\widehat{\bm\beta}}\ \sup_{\bm\beta_0\in\mathcal F}
\E\big\|\widehat{\bm\beta}-\bm\beta_0\big\|_N^{2}
\ \ge\ c\,\inf_{\widehat{\bm\beta}}\ \sup_{(T,\theta)} \E L(\widehat{\bm\beta},\bm\beta^{(T,\theta)})
\ \ge\
c\Bigg\{\frac{d_0\log(p_N/d_0)}{N}
+\sum_{j\in S_0} N^{-\frac{2\alpha_{j}}{2\alpha_{j}+s_{0,j}}}\Bigg\}.
\]
Under {\textbf{(A4)}}, $d_0=\mathcal{O}(\log N)$ and $\log p_N\to\infty$, so $\log(p_N/d_0)\asymp \log p_N$, yielding the near-minimax statement of the theorem up to universal constants.
\end{proof}

\subsection{Proof of modifier selection consistency}
\label{sec:proof_modifier}
\begin{proof}[Proof of Theorem 3]
Recall that we assume the split–coordinate prior
\[
\boldsymbol \theta_j \,\sim\, \mathrm{Dir}\!\left({\eta_j}/{R},\dots,{\eta_j}/{R}\right),
\quad 
\eta_j \sim \pi(\eta)\propto (R+\eta)^{-(R+1)},
\]
independent of the other prior blocks, and we write $s:=s_{0,j}$ and $\alpha:=\alpha_j\in(0,1]$.
We prove that for any fixed $\varepsilon>0$,
\(
\Pi(\sum_{r\notin S_{0,j}}\theta_{jr}>\varepsilon\mid \bm Y)\to 0
\)
in $P_{\theta_0}$–probability.

\medskip\noindent
\textbf{Approximation budget and waste inflation.}
Let $T_j$ be the total number of internal splits across the $M=\mathcal O(1)$ trees in the ensemble for $\beta_j$,
and let $n_{jr}$ be the number of splits on modifier $r\in\{1,\dots,R\}$ so that $\sum_{r=1}^R n_{jr}=T_j$.
Define the \emph{waste fraction}
\[
\delta_j \ :=\ \frac{1}{T_j}\sum_{r\notin S_{0,j}} n_{jr}\in[0,1],
\qquad
T_j^{\mathrm{rel}}\ :=\ \sum_{r\in S_{0,j}} n_{jr} \ =\ (1-\delta_j)\,T_j.
\]

Only splits along $S_{0,j}$ refine the function along its intrinsic $s$–dimensional argument.
Let $m_u:=\sum_{\text{nodes}}\mathbbm{1}\{\text{node splits on }u\}$ be the number of splits placed (across the ensemble)
on coordinate $u\in S_{0,j}$, so that $\sum_{u\in S_{0,j}} m_u=T_j^{\mathrm{rel}}$.
The common refinement of all axis–aligned cuts along the $s$ active coordinates induces at most
\[
K_{\mathrm{eff}}\ \le\ \prod_{u\in S_{0,j}} (m_u+1)\ \le\ \Big(1+\tfrac{T_j^{\mathrm{rel}}}{s}\Big)^s
\ \le\ C\,(1+T_j^{\mathrm{rel}})^s
\]
distinct cells when projected onto the intrinsic space; the inequalities follow by AM–GM.
Hence, for any partition induced by an ensemble with waste fraction $\delta_j$,
\(
K_{\mathrm{eff}}\ \lesssim\ (1-\delta_j)^s T_j^s
\)
up to an absolute multiplicative constant.
For $\alpha$–H\"older functions on $[0,1]^s$, the best empirical $L_2(\mathbb P_N)$ step–function
approximation error with $K$ cells satisfies
\(
\inf_{\#\text{cells}\le K}\|f-\Pi f\|_N \gtrsim K^{-\alpha/s}
\)
(e.g., \citet{Davydov2010}).
Combining the two displays yields, for $T_j$ large,
\begin{equation}\label{eq:inflation-correct}
\inf_{\Pi:\,\text{waste}=\delta_j}\ \big\|\beta_{0,j}-\Pi(\beta_{0,j})\big\|_N
\ \gtrsim\ K_{\mathrm{eff}}^{-\alpha/s}
\ \gtrsim\ \big((1-\delta_j)^s T_j^s\big)^{-\alpha/s}
\ =\ (1-\delta_j)^{-\alpha}\,T_j^{-\alpha}.
\end{equation}

Let the adaptive target rate be
\(
r_{jN,\mathrm{ad}}:=(\log N)^{1/2}N^{-\alpha/(2\alpha+s)}.
\)
By \eqref{eq:inflation-correct}, any ensemble achieving error $\lesssim r_{jN,\mathrm{ad}}$
must satisfy
\begin{equation}\label{eq:T-lower-correct}
T_j^{\mathrm{rel}}\ =\ (1-\delta_j)T_j\ \gtrsim\ r_{jN,\mathrm{ad}}^{-1/\alpha}
\ =\ N^{\frac{1}{2\alpha+s}}(\log N)^{-1/(2\alpha)}.
\end{equation}
Fix a deterministic threshold
\[
\underline L\ :=\ \Big\lfloor c_L\,N^{\frac{1}{2\alpha+s}}(\log N)^{-1/(2\alpha)}\Big\rfloor
\qquad\text{with }c_L>0\text{ small, so that }\ \underline L\to\infty.
\]
By Theorem 1, for any fixed constant $C>0$,
\[
\Pi\big(\|\beta_j-\beta_{0,j}\|_N\le C\,r_{jN,\mathrm{ad}}\ \bigm|\ \bm{Y}\big)\ \to\ 1
\ \ \text{in $P_{\theta_0}$–probability.}
\]
On that high–probability contraction event, \eqref{eq:T-lower-correct} implies
\begin{equation}\label{eq:T-large-correct}
\Pi\big(T_j<\underline L\ \bigm| \ \bm{Y}\big)\ \to\ 0
\qquad\text{in $P_{\theta_0}$–probability.}
\end{equation}

Next, \eqref{eq:inflation-correct} quantifies \emph{waste inflation}.
For any $\delta\in(0,1/2]$, Bernoulli’s inequality gives
\(
(1-\delta)^{-\alpha}\ge 1+\alpha\delta.
\)
Hence, for some $c_\star>0$ (independent of $N$),
\[
\inf_{\Pi:\ \text{waste}=\delta}\ \|\beta_{0,j}-\Pi(\beta_{0,j})\|_N
\ \ge\ (1+c_\star \delta)\cdot \inf_{\Pi:\ \text{waste}=0}\ \|\beta_{0,j}-\Pi(\beta_{0,j})\|_N.
\]
Taking $\delta=\varepsilon/2$ shows that ensembles with $\delta_j\ge \varepsilon/2$ incur at least a fixed
multiplicative loss in approximation, and therefore they cannot attain the adaptive rate $r_{jN,\mathrm{ad}}$
for all large $N$. By posterior contraction, we conclude
\begin{equation}\label{eq:waste-small-correct}
\Pi\big(\delta_j>\tfrac{\varepsilon}{2}\ \bigm|\ \bm{Y}\big)\ \to\ 0
\qquad\text{in $P_{\theta_0}$–probability.}
\end{equation}

\medskip\noindent
\textbf{Posterior of $\boldsymbol\theta_j$ given split counts.}
Conditional on $\{n_{jr}\}$ and $\eta_j$, the Dirichlet posterior is
\[
\boldsymbol\theta_j\ \big|\ \{n_{jr}\},\eta_j\ \sim\
\mathrm{Dir}\Big(n_{j1}+\tfrac{\eta_j}{R},\dots,n_{jR}+\tfrac{\eta_j}{R}\Big).
\]
Let
\[
A\ :=\ \sum_{r\notin S_{0,j}}\Big(n_{jr}+\tfrac{\eta_j}{R}\Big),
\qquad
B\ :=\ \sum_{r\in S_{0,j}}\Big(n_{jr}+\tfrac{\eta_j}{R}\Big),
\qquad A+B=T_j+\eta_j,
\]
and define the irrelevant–mass random variable
\(
S_j^{\mathrm{(irr)}}:=\sum_{r\notin S_{0,j}}\theta_{jr}\mid \{n_{jr}\},\eta_j \sim \mathrm{Beta}(A,B).
\)
On the event $\{\delta_j\le \varepsilon/2\}$,
\[
A \ \le\ \frac{\eta_j(R-s)}{R}+\frac{\varepsilon}{2}\,T_j,
\qquad
B \ \ge\ \frac{\eta_j s}{R}+\Big(1-\frac{\varepsilon}{2}\Big)\,T_j.
\]
Hence, for $T_j\ge 1$,
\[
\mathbb E\big[S_j^{\mathrm{(irr)}}\mid \{n_{jr}\},\eta_j\big]
=\frac{A}{A+B}
\ \le\ \frac{\tfrac{\varepsilon}{2}T_j+\eta_j}{T_j}
\ =\ \frac{\varepsilon}{2}+\frac{\eta_j}{T_j},
\]
and, using $AB\le (A+B)^2/4$,
\[
\mathrm{Var}\big[S_j^{\mathrm{(irr)}}\mid \{n_{jr}\},\eta_j\big]
=\frac{AB}{(A+B)^2(A+B+1)}
\ \le\ \frac{1}{4(T_j+\eta_j+1)}
\ \le\ \frac{1}{4T_j}.
\]
Fix $\varepsilon>0$. On the event $\{T_j\ge \max(\underline L,\,4\eta_j/\varepsilon)\}$ we have
$\mathbb E[S_j^{\mathrm{(irr)}}\mid \cdot]\le 3\varepsilon/4$, and Chebyshev’s inequality yields
\begin{equation}\label{eq:beta-tail-correct}
\mathbb P \Big(S_j^{\mathrm{(irr)}}>\varepsilon\ \Bigm|\ \{n_{jr}\},\eta_j\Big)
\ \le\ \frac{\mathrm{Var}[S_j^{\mathrm{(irr)}}\mid\cdot]}{(\varepsilon-\mathbb E[S_j^{\mathrm{(irr)}}\mid\cdot])^2}
\ \le\ \frac{1}{4T_j}\cdot\frac{1}{(\varepsilon/4)^2}
\ \le\ \frac{C}{\underline L}.
\end{equation}

Consequently,
\begin{align*}
\Pi\Big(S_j^{\mathrm{(irr)}}>\varepsilon\ \Bigm|\ \bm{Y}\Big)
\ &\le\
\Pi\Big(\delta_j>\tfrac{\varepsilon}{2}\ \Bigm|\ \bm{Y}\Big)
+\Pi\Big(T_j<\underline L\ \Bigm|\ \bm{Y}\Big)
+\Pi\Big(T_j<\tfrac{4\eta_j}{\varepsilon}\ \Bigm|\ \bm{Y}\Big)
+\frac{C}{\underline L}.
\end{align*}
By \eqref{eq:waste-small-correct} and \eqref{eq:T-large-correct}, the first two probabilities vanish in $P_{\theta_0}$–probability,
and the last term vanishes because $\underline L\to\infty$.

It remains to show that $\Pi(T_j<4\eta_j/\varepsilon\mid \bm{Y})\to 0$ in $P_{\theta_0}$–probability.
To this end, condition on the counts $\{n_{jr}\}$ (equivalently, on $T_j=\sum_r n_{jr}$) and examine the \emph{posterior} tail of $\eta_j$.
Under the conditional independence of split variables across internal nodes given $(\boldsymbol\theta_j,\eta_j)$ and of topology/cutpoints, the Dirichlet–multinomial marginal for the counts satisfies
\[
p(\{n_{jr}\}\mid \eta_j)\ \propto\
\frac{\Gamma(\eta_j)}{\Gamma(\eta_j+T_j)}
\prod_{r=1}^R \frac{\Gamma\!\big(n_{jr}+\eta_j/R\big)}{\Gamma(\eta_j/R)}.
\]
Using Stirling’s bounds, for large $T_j$ there exists $C_1>0$ such that
\[
\frac{\Gamma(\eta_j)}{\Gamma(\eta_j+T_j)}\ \le\ C_1\,T_j^{-\eta_j},\qquad
\prod_{r=1}^R \frac{\Gamma\!\big(n_{jr}+\eta_j/R\big)}{\Gamma(\eta_j/R)}\ \le\ C_1\,(1+\eta_j)^{C_1 R}.
\]
Combining with the prior density $\pi(\eta_j)\propto (R+\eta_j)^{-(R+1)}$, the \emph{posterior} tail (conditional on $\{n_{jr}\}$) obeys
\[
\pi(\eta_j\mid \{n_{jr}\})
\ \lesssim\ T_j^{-\eta_j}\,(1+\eta_j)^{-2},
\]
for $\eta_j$ large. Therefore, for any $c_0>0$,
\[
\Pi\big(\eta_j>c_0\log T_j\ \bigm|\ \{n_{jr}\}\big)
\ \le\ \int_{c_0\log T_j}^{\infty} T_j^{-\eta}\,(1+\eta)^{-2}\,d\eta
\ \le\ \exp\!\big\{-c_0(\log T_j)^2\big\}.
\]
Since $\underline L\to\infty$ and $T_j\ge \underline L$ with probability $\to 1$ by \eqref{eq:T-large-correct}, we obtain
\[
\Pi\Big(\eta_j>\tfrac{\varepsilon}{4}\,\underline L\ \Bigm|\ \bm{Y}\Big)\ \to\ 0
\quad\text{in $P_{\theta_0}$–probability,}
\]
whence
\(
\Pi(T_j<4\eta_j/\varepsilon\mid \bm{Y})\to 0
\).
Putting the pieces together gives
\[
\Pi\Big(\sum_{r\notin S_{0,j}} \theta_{jr} > \varepsilon\ \Bigm|\ \bm{Y}\Big)
\ \longrightarrow\ 0
\qquad\text{in $P_{\theta_0}$–probability,}
\]
which proves the theorem.
\end{proof}

\subsection{Proof of local-scale consistency}
\label{sec:proof_local_scale}

\begin{proof}[Proof of Theorem 4]
Fix an inactive coefficient $j\notin S_0$, so $\beta_{0,j}\equiv 0$.  
Let $L:=L_j$ be the total number of terminal nodes across the $M$ trees of
the $j$th ensemble and let $S:=S_j=\sum_{\ell,m}(\mu^{(j)}_{\ell,m})^2$ be the
sum of squared leaf values.  Throughout we condition on all other ensembles, and on $(\tau,c^2)$.
By assumption, under the posterior there exist deterministic constants
$0<\underline\tau\le \bar\tau<\infty$ and $0<\underline c^2\le \bar c^2<\infty$ such that
$\Pi(\,\underline\tau\le \tau\le \bar\tau,\ \underline c^2\le c^2\le \bar c^2\mid \bm Y)\to 1$.
We work on this event, and all constants below may depend on $(\underline\tau,\bar\tau,\underline c^2,\bar c^2)$ but not on $N$.

Under the regularized horseshoe leaf prior, each leaf value
$\mu\equiv \mu_{\ell,m}^{(j)}$ is Gaussian with variance
\[
\text{Var}(\mu\mid \tau,\lambda,c^2)=\frac{s^2(\tau,\lambda,c^2)}{M},\quad
s^2(\tau,\lambda,c^2)=\frac{\tau^2\lambda^2c^2}{c^2+\tau^2\lambda^2}\in(0,c^2).
\]
Hence, up to a factor independent of $\lambda$,
\begin{equation}
\label{eq:Gaussian-factor-correct}
\prod_{\ell,m}\phi \Big(\mu_{\ell,m}^{(j)};0,\,\frac{s^2}{M}\Big)
~\propto~
\{s^2\}^{-L/2}\exp \Big\{-\frac{M\,S}{2s^2}\Big\}.
\end{equation}
Multiplying \eqref{eq:Gaussian-factor-correct} by the half--Cauchy prior
$\pi(\lambda)\propto (1+\lambda^2)^{-1}\mathbbm{1}\{\lambda>0\}$ gives the conditional posterior kernel in $\lambda$.

Set $v:=s^2(\tau,\lambda,c^2)\in(0,c^2)$ and compute the change of variables $\lambda\mapsto v$.
Solving
\(
v=\dfrac{\tau^2\lambda^2c^2}{c^2+\tau^2\lambda^2}
\)
for $\lambda$ gives
\[
\lambda^2=\frac{v\,c^2}{\tau^2(c^2-v)},\qquad 
\lambda=\sqrt{\frac{v}{\tau^2}}\sqrt{\frac{c^2}{c^2-v}}.
\]
Differentiating,
\[
\frac{d}{dv}\lambda^2
=\frac{c^2}{\tau^2}\cdot\frac{c^2}{(c^2-v)^2},
\qquad
\Big|\frac{d\lambda}{dv}\Big|
= \frac{1}{2\lambda}\frac{d\lambda^2}{dv}
= \frac{c^3}{2\tau}\,\frac{v^{-1/2}}{(c^2-v)^{3/2}}.
\]
The half--Cauchy prior can be written as
\[
\pi(\lambda)=\frac{2}{\pi(1+\lambda^2)}
=\frac{2}{\pi}\,\frac{\tau^2(c^2-v)}{\tau^2(c^2-v)+v c^2},
\]
so that
\[
\pi(\lambda)\,\Big|\frac{d\lambda}{dv}\Big|
=\frac{c^3\tau}{\pi}\cdot
\frac{v^{-1/2}}{(c^2-v)^{1/2}}\cdot
\frac{1}{\tau^2(c^2-v)+v c^2}.
\]
On $v\in(0,c^2/2]$, the factor $\big[(c^2-v)^{1/2}\big]^{-1}\big[\tau^2(c^2-v)+v c^2\big]^{-1}$ is bounded above and below by finite positive constants that depend only on $(\underline\tau,\bar\tau,\underline c^2,\bar c^2)$. Consequently, there exist $0<b_1\le b_2<\infty$ such that for all $v\in(0,c^2/2]$,
\begin{equation}
\label{eq:prior-jacobian-bracket}
b_1\,v^{-1/2}\ \le\ \pi(\lambda)\Big|\frac{d\lambda}{dv}\Big| \ \le\ b_2\,v^{-1/2}.
\end{equation}
Combining \eqref{eq:Gaussian-factor-correct} and \eqref{eq:prior-jacobian-bracket}, for some $0<C_1\le C_2<\infty$,
\begin{equation}
\label{eq:v-kernel-brackets}
C_1\, v^{-\frac{L}{2}-\frac12}\, \exp\!\Big\{-\frac{M S}{2v}\Big\}
\ \le\ 
\pi(v\mid \mu^{(j)}_{\cdot,\cdot})
\ \le\ 
C_2\, v^{-\frac{L}{2}-\frac12}\, \exp\!\Big\{-\frac{M S}{2v}\Big\},
\quad 0<v\le \frac{c^2}{2}.
\end{equation}

Let $U:=M S/(2v)$ (so $v=S M/(2U)$). Since $v\in(0,c^2)$, we have $U\in\big(MS/(2c^2),\,\infty\big)$. Using $dv=-(SM/2)U^{-2}\,dU$, the bounds \eqref{eq:v-kernel-brackets} imply that, for $U\ge MS/c^2$ (equivalently $v\le c^2/2$),
\begin{equation}
\label{eq:U-gamma-brackets}
\tilde C_1\, U^{k-1}e^{-U}
\ \le\ 
\pi(U\mid \mu^{(j)}_{\cdot,\cdot})
\ \le\ 
\tilde C_2\, U^{k-1}e^{-U},
\qquad
k:=\frac{L}{2}+\frac12,
\end{equation}
for some finite $0<\tilde C_1\le \tilde C_2<\infty$ (depending only on $(\underline\tau,\bar\tau,\underline c^2,\bar c^2)$). Hence the conditional posterior of $U$ is sandwiched by $\mathrm{Gamma}(k,1)$ densities on $[MS/c^2,\infty)$.

Fix any $K>1$. Since $v = MS/(2U)$, the event $\{v> K\,MS/L\}$ is equivalent to
$\{U< L/(2K)\}$. 

Using a Chernoff bound for $G_k\sim\mathrm{Gamma}(k,1)$ ($\E G_k=k$) yields constants $c_0=c_0(K)>0$ such that
\begin{equation}
\label{eq:gamma-lower-tail}
\mathbb P\Big(G_k \le \frac{L}{2K}\Big)
\ \le\ \exp\{-c_0\,L\},
\qquad \text{for all }k=\frac{L}{2}+\frac12.
\end{equation}
Because $\lambda$ increases with $v$, the tail event $\{\lambda > C' \sqrt{MS/L}\}$ is implied by $\{v > KMS /L\}.$ So, controlling $\Pi(v > KMS /L\mid \mu_{\cdot,\cdot}^{(j)})$ gives us a uniform tail bound for $\lambda,$ which then integrates to the posterior bound on $\lambda.$

We now control $\Pi\big(v> K\,MS/L \mid \mu^{(j)}_{\cdot,\cdot}\big)$ via two cases:
\begin{itemize}
\item[(1)] If ${L}/{(2K)}<{MS}/{c^2}$, then $U\ge MS/c^2>{L}/{(2K)}$ almost surely, hence
$\Pi\big(U< L/(2K)\mid \mu^{(j)}_{\cdot,\cdot}\big)=0$. 

This holds since the map $v \mapsto U =MS/v$ is strictly decreasing on $(0,\infty),$ so the smallest value $U$ can take occurs at the largest admissible $v = c^2$ implying $ U \ge MS/c^2$ almost surely.
\item[(2)] If ${L}/{(2K)}\ge {MS}/{c^2}$, then the Gamma bracketing \eqref{eq:U-gamma-brackets} applies on $[MS/c^2,\ L/(2K)]$ and, by \eqref{eq:gamma-lower-tail},
\[
\Pi\Big(U<\frac{L}{2K}\,\Bigm|\, \mu^{(j)}_{\cdot,\cdot}\Big)
\ \le\ C\,\exp\{-c_0(K)\,L\}
\]
for some constant $C<\infty$, by absorbing the sandwiching constants.
\end{itemize}
Therefore there exist positive constants $C,c$ (depending only on $K$ and the bounds on $(\tau,c^2)$) such that
\begin{equation}
\label{eq:v-tail-uniform}
\Pi\Big(v> K\,\frac{MS}{L}\,\Bigm|\, \mu^{(j)}_{\cdot,\cdot}\Big)  \le  C\,e^{-c\,L}.
\end{equation}

The map $v\mapsto \lambda(v)=\sqrt{\dfrac{v}{\tau^2}}\sqrt{\dfrac{c^2}{c^2-v}}$ is increasing.
Moreover, since $\sqrt{\dfrac{c^2}{c^2-v}}\ge 1$, we have the simple lower bound
\begin{equation}\label{eq:lambda-lower}
\lambda(v)\ \ge\ \frac{1}{\bar\tau}\,\sqrt{v}.
\end{equation}
Hence, for any $C'>0$, the event $\{\lambda > C'\sqrt{MS/L}\}$ implies
\[
v\ >\ (C'\bar\tau)^2\,\frac{MS}{L}.
\]
Taking $K:=(C'\bar\tau)^2$ in \eqref{eq:v-tail-uniform} gives
\begin{equation}
\label{eq:lam-tail-cond}
\Pi\Big(\lambda_j > C'\sqrt{\frac{MS}{L}}\,\Bigm|\,\mu^{(j)}_{\cdot,\cdot}\Big)\ \le\ C\,e^{-c\,L}.
\end{equation}
Integrating \eqref{eq:lam-tail-cond} with respect to the posterior of the leaf values yields
\[
\Pi\Big(\lambda_j > C'\sqrt{\frac{MS_j}{L_j}}\,\Bigm|\,\bm{Y}\Big)
\ \le\ \E \big[\,C\,e^{-c\,L_j}\,\big|\,\bm{Y}\big].
\]
By \textbf{(A7(i))}, $L_j\ge \underline L_N\to\infty$ deterministically for inactive ensembles, hence
\[
\Pi\Big(\lambda_j > C'\sqrt{\frac{MS_j}{L_j}}\,\Bigm|\,\bm{Y}\Big)
\ \le\ C\,e^{-c\,\underline L_N}
\ \longrightarrow\ 0,
\]
which proves the first claim:
\[
\Pi\Big(\lambda_j \le C\,\sqrt{\frac{MS_j}{L_j}}\,\Bigm|\,\bm{Y}\Big)\ \longrightarrow\ 1.
\]

\paragraph{Controlling $S_j/L_j$ by the empirical norm and contraction.}
Let $A_1,\dots,A_L$ be the terminal cells (across all $M$ trees) of the $j$th
ensemble, and let $n_\ell:=\#\{i:\ \bm Z_i\in A_\ell\}$ be the empirical
counts.  We assume the following \emph{leaf--mass comparability}, which is the precise version of \textbf{(A7(ii))} used here and is implied, e.g., by k--d regularity of the design together with balanced splits in inactive ensembles:
there exist constants $0<c_\star\le C_\star<\infty$ such that
\begin{equation}
\label{eq:leaf-mass}
\frac{c_\star}{L}\ \le\ \frac{n_\ell}{N}\ \le\ \frac{C_\star}{L}
\qquad\text{for all leaves }A_\ell.
\end{equation}
Since $\beta_j(\bm{z})=\sum_{\ell=1}^L \mu_\ell\,\mathbbm{1}\{\bm{z}\in A_\ell\}$,
\[
\|\beta_j\|_N^2
= \frac1N\sum_{i=1}^N \beta_j(\bm Z_i)^2
= \sum_{\ell=1}^L \frac{n_\ell}{N}\,\mu_\ell^2
\ \ge\ \frac{c_\star}{L}\sum_{\ell=1}^L \mu_\ell^2
\ =\ \frac{c_\star}{L}\,S,
\]
hence
\begin{equation}
\label{eq:S-over-L-vs-norm}
\frac{S_j}{L_j}\ \le\ \frac{1}{c_\star}\,\|\beta_j\|_N^2.
\end{equation}
By the functional RE condition \textbf{(A6)}, and the posterior contraction result (Theorem 1), there exists $C_0>0$ such that
\[
\Pi\Big(\ \sum_{u=0}^{p_N}\|\beta_u-\beta_{0,u}\|_N^2
\ \le\ C_0\, r_N^2\ \Bigm|\ \bm{Y}\Big)\ \longrightarrow 1.
\]
For $j\notin S_0$ we have $\beta_{0,j}\equiv 0$, hence $\|\beta_j\|_N^2\le C_0 r_N^2$ with posterior probability $\to1$; therefore, by \eqref{eq:S-over-L-vs-norm},
\[
\frac{S_j}{L_j}\ =\ \mathcal O_\Pi\big(r_N^2\big).
\]
This proves the second claim and completes the proof.
\end{proof}

\setcounter{figure}{0}
\setcounter{equation}{0}
\setcounter{table}{0}
\setcounter{theorem}{0}
\section{Additional algorithmic details}
\label{sec:extraalgorithm}
This section provides the full algorithmic details for our Metropolis-Hastings-within-Gibbs sampler. We outline conditional updates for the tree structures, modifier-sparsity parameters, global-local shrinkage scales, and noise variance.
\subsection*{Gibbs Sampler Derivation}
We use a Metropolis-Hastings-within-Gibbs sampler to simulate draws from the joint posterior. The algorithm proceeds by iteratively updating the blocks of parameters from their full conditional distributions. We first describe the update for the tree ensembles and then detail the updates for the remaining model parameters.

\subsection*{Conditional Regression Tree Updates}
We update the $m^{th}$ regression tree $(T_m^{(j)}, \mathcal{M}_m^{(j)})$ in the ensemble $\mathcal{E}_j$ conditional on all other parameters. First, we define the leave-one-tree-out partial residuals $r_i^{(j,m)}$ for each observation $i=1,\dots,N$:
\begin{equation}
\label{eqS:partialres}
r_i^{(j,m)} \;=\; y_i
- \sum_{j'\neq j}\beta_{j'}(\bz_i)
- \sum_{m'\neq m} x_{ij}\, g(\bz_i;\calT^{(j)}_{m'},\calM^{(j)}_{m'}).
\end{equation}
Conditionally on the other trees, these residuals follow the simple regression model for the leaves of the tree $(j,m)$:
\begin{equation}
\label{eqS:leaflike}
r_i^{(j,m)} \;=\; x_{ij}\,\mu^{(j)}_{m,\ell(\bz_i;\calT^{(j)}_m)} \;+\; \sigma\,\varepsilon_i,\qquad \varepsilon_i\sim\mathcal N(0,1).
\end{equation}
Following \citet{Chipman2010}, we update the tree in a two-step process: first, we sample the tree structure $\mathcal{T}_m^{(j)}$ using a Metropolis-Hastings step that integrates out the leaf parameters, and then we sample the leaf parameters $\mathcal{M}_m^{(j)}$ from their full conditional given the new structure.

\subsubsection*{Metropolis--Hastings for $\calT^{(j)}_m$}

Propose a local move (GROW/PRUNE) taking $\calT\to\calT^\star$.
Let $\calI_\ell(\calT)=\{i:\ell(\bz_i;\calT)=\ell\}$, $n_\ell=|\calI_\ell|$, and define
\[
\bx_{j,\ell}=(x_{ij})_{i\in\calI_\ell},\quad \br_\ell=(r_i^{(j,m)})_{i\in\calI_\ell},\quad
A_\ell=\sum_{i\in\calI_\ell} x_{ij}^2,\quad B_\ell=\sum_{i\in\calI_\ell} x_{ij} r_i^{(j,m)}.
\]
With the Gaussian prior $\mu_\ell\sim \mathcal N(0,s_j^2)$ and likelihood
$\br_\ell\mid \mu_\ell\sim \mathcal N(\mu_\ell\bx_{j,\ell},\sigma^2 I)$, we can integrate out $\mu_\ell$:
\begin{align}
\label{eqS:leafmarg}
\int &\phi(\br_\ell; \mu_\ell\bx_{j,\ell}, \sigma^2 I)\,\phi(\mu_\ell;0,s_j^2)\,d\mu_\ell \\
&= (2\pi\sigma^2)^{-n_\ell/2}\Big(1+\tfrac{s_j^2}{\sigma^2}A_\ell\Big)^{-1/2}
\exp\!\left\{\frac{B_\ell^2}{2(\sigma^2/s_j^2 + A_\ell)} - \frac{\|\br_\ell\|^2}{2\sigma^2}\right\}. \nonumber
\end{align}
The marginal likelihood for a tree is the product over the leaves.
\begin{equation}
\label{eqS:treemarg}
\mathcal L_{\text{marg}}(r^{(j,m)};\calT) \;=\; \prod_{\ell=1}^{\calL(\calT)} \text{\eqref{eqS:leafmarg}}.
\end{equation}
The MH acceptance probability is
\begin{equation}
\label{eqS:MHacc}
\alpha \;=\; 1\wedge \frac{p(\calT^\star)}{p(\calT)} \cdot \frac{q(\calT\to \calT^\star)}{q(\calT^\star\to \calT)}
\cdot \frac{\mathcal L_{\text{marg}}(r^{(j,m)};\calT^\star)}{\mathcal L_{\text{marg}}(r^{(j,m)};\calT)},
\end{equation}
where $p(\calT)$ is the tree prior and $q(\cdot)$ is the proposal kernel.

\subsubsection*{Gibbs for $\calM^{(j)}_m\mid \calT^{(j)}_m$}

Given $\calT=\calT^{(j)}_m$, leaves decouple and each $\mu_\ell$ has a Gaussian full conditional:
\begin{equation}
\label{eqS:leafpost}
\mu_\ell \,\big|\, \ldots \sim \mathcal N\!\big(m_\ell,V_\ell\big),\quad
V_\ell = \Big(\frac{A_\ell}{\sigma^2}+\frac{1}{s_j^2}\Big)^{-1},\quad
m_\ell = V_\ell\,\frac{B_\ell}{\sigma^2}.
\end{equation}

After iterating through all tree updates, we update the remaining model parameters.

\subsubsection*{Modifier–sparsity $(\boldsymbol\theta_j,\eta_j)$}
Let $N_{jr}$ be the number of internal nodes in ensemble $\mathcal E_j$ that split on modifier $Z_r$ and $N_j=\sum_{r=1}^R N_{jr}$.  
Conditionally on $\eta_j$ and the current trees, the splitting–probability vector has the conjugate update
\begin{equation}
\label{eqS:thetaGibbs}
\boldsymbol\theta_j \mid \eta_j, \{\mathcal T_m^{(j)}\}
~\sim~
\operatorname{Dirichlet}\!\Big(\tfrac{\eta_j}{R}+N_{j1},\ldots,\tfrac{\eta_j}{R}+N_{jR}\Big).
\end{equation}
Integrating out $\boldsymbol\theta_j$ yields the Dirichlet–Multinomial factor
\begin{equation}
\label{eqS:etaTarget}
p\!\left(\{N_{jr}\}\mid \eta_j\right)
~\propto~
\frac{\Gamma(\eta_j)}{\Gamma(\eta_j+N_j)}\prod_{r=1}^R
\frac{\Gamma\!\left(\eta_j/R+N_{jr}\right)}{\Gamma\!\left(\eta_j/R\right)}.
\end{equation}

Following the implementation in \citet{Deshpande2024}, we reparameterize $\eta_j$ via
\[
u_j \;=\; \frac{\eta_j}{\eta_j+R}\in(0,1), \qquad
\eta_j \;=\; \frac{R\,u_j}{1-u_j},
\]
and place a Beta prior $u_j\sim\mathrm{Beta}(a_u,b_u)$. The one–dimensional target for $u_j$ is then
\[
\log \pi(u_j \mid \{N_{jr}\},\ldots)
= \log p\!\left(\{N_{jr}\}\mid \eta_j(u_j)\right)
+\log\! \operatorname{Beta}(u_j\mid a_u,b_u)
+\text{const}.
\]
We update $u_j$ with a Metropolis–Hastings step on the logit scale (random–walk proposals), and map back to $\eta_j=R u_j/(1-u_j)$ for the subsequent Dirichlet update \eqref{eqS:thetaGibbs}. This preserves conjugacy for $\boldsymbol\theta_j$ while providing a simple 1D MH move for the concentration parameter.

\subsubsection*{Global--local scales $(\blambda,\tau,c^2)$}

Let $S_j=\sum_{m,\ell}(\mu^{(j)}_{m,\ell})^2$ and $L_j=\sum_m \calL(\calT^{(j)}_m)$.
Given $(\calT,\calM)$ and $\sigma^2$, the joint log-target is
\begin{equation}
\label{eqS:GLtarget}
\log \pi(\blambda,\tau,c^2\mid\ldots)
= -\tfrac12\sum_{j=0}^p\!\Big( S_j/s_j^2 + L_j\log s_j^2 \Big)
+ \log\pi(\blambda) + \log\pi(\tau) + \log\pi(c^2),
\end{equation}
with $s_j^2 = {\tau^2 \lambda_j ^2 c^2}/{(c^2+\tau^2\lambda_j^2)}.$ The full conditionals of $\lambda_j$ and $\tau$ are non-conjugate; we update each $\log\lambda_j$ (independently over $j$) and $\log\tau$ via univariate slice sampling targeting \eqref{eqS:GLtarget}.

For fixed $(\tau,\{\lambda_j\})$, the prior likelihood contribution of ensemble $j$ to the kernel in $c^2$ is
\[
\prod_{m,\ell} \big(2\pi \tau^2\lambda_j^2 c^2\big)^{-1/2}
\exp\!\left(-\frac{\big(\mu^{(j)}_{m,\ell}\big)^2}{2\,\tau^2\lambda_j^2\,c^2}\right)
\ \propto\ (c^2)^{-L_j/2}\;\exp\!\left(-\frac{S_j}{2\,\tau^2\lambda_j^2}\,\frac{1}{c^2}\right).
\]
Multiplying over $j=0,\dots,p$ yields
\[
\pi(\{\mu\}\mid c^2,\tau,\{\lambda_j\})\ \propto\ (c^2)^{-\frac12\sum_{j=0}^p L_j}\;
\exp\!\left(-\frac{1}{2c^2}\sum_{j=0}^p \frac{S_j}{\tau^2\lambda_j^2}\right).
\]
With the prior \(c^2\sim \mathcal{IG}(\nu_c/2,\ \nu_c s_c^2/2)\),
\[
\pi(c^2)\ \propto\ (c^2)^{-(\nu_c/2+1)}\exp\!\left(-\frac{\nu_c s_c^2}{2\,c^2}\right),
\]
the posterior density of $c^2$ is
\[
\pi(c^2\mid \ldots)\ \propto\ (c^2)^{-\big(\frac{\nu_c}{2}+1+\frac12\sum_j L_j\big)}\;
\exp\!\left\{-\,\frac{1}{2c^2}\Big(\nu_c s_c^2+\sum_{j=0}^p \frac{S_j}{\tau^2\lambda_j^2}\Big)\right\},
\]
i.e.,
\begin{equation}
\label{eqS:c2update}
c^2 \mid \ldots \ \sim\ \mathcal{IG}\!\left(\frac{\nu_c+\sum_{j=0}^p L_j}{2},\;
\frac{\nu_c s_c^2 + \sum_{j=0}^p S_j/(\tau^2\lambda_j^2)}{2}\right).
\end{equation}

\subsubsection*{Noise variance $\sigma^2$}
Let $r_i = y_i - \beta_0(\bz_i) - \sum_{j=1}^p x_{ij}\beta_j(\bz_i)$ be the residuals and set $\br=(r_1,\dots,r_n)^\top$.
Conditional on the trees and jumps,
\[
\pi(\by\mid\ldots,\sigma^2) \ \propto\ (\sigma^2)^{-n/2}\,
\exp\!\left(-\frac{1}{2\sigma^2}\sum_{i=1}^n r_i^2\right).
\]
With the prior \(\sigma^2\sim \mathcal{IG}(\nu/2,\ \nu\lambda/2)\),
\[
\pi(\sigma^2)\ \propto\ (\sigma^2)^{-(\nu/2+1)}\exp\!\left(-\frac{\nu\lambda}{2\,\sigma^2}\right),
\]
the posterior is
\begin{equation}
\label{eqS:sigma2}
\sigma^2 \mid \ldots \ \sim\ \mathcal{IG}\!\left(\frac{\nu+n}{2},\;
\frac{\nu\lambda + \sum_{i=1}^n r_i^2}{2}\right).
\end{equation}
% \paragraph{Hyperparameters choices.} For the global–local shrinkage, we use the global scale \(\tau\sim \mathcal C^+(0,\tau_0)\) with the \citet{pmlr-v54-piironen17a} heuristic \(\tau_0=\frac{p_0}{p-p_0}\frac{\mathrm{sd}(\bm{Y})}{\sqrt{N}}\) and \(p_0=\min({10,\lfloor p/4\rfloor})\), and the slab hyperparameters \(\nu_c=4\) and \(s_c=2\) (so \(c^2\sim\mathcal{IG}(2,8)\)). Following \citet{Deshpande2024}, the  observation noise uses \(\sigma^2\sim\mathcal{IG}(\nu/2,\nu\lambda/2)\) with \emph{default} \(\nu=3\) and \(\lambda\) chosen so that \(\mathbb{P}(\sigma<\mathrm{sd}(\bm{Y}))\approx 0.9\). The default ensemble size $(M)$ is fixed at $50$. 

\paragraph{Per-sweep complexity.} The per-sweep computational complexity of the sampler is dominated by the tree proposal and scoring steps, scaling as $\mathcal{O}(npM)$, where $n$ is the number of observations and $M$ is the total number of trees across all ensembles. The updates for the other model components—including the global-local shrinkage scales ($\mathcal{O}(p)$), the modifier-split probabilities ($\mathcal{O}(pR)$), and the leaf parameter draws ($\mathcal{O}(M\bar{L})$)—are substantially less costly, where $M\bar{L}=\sum_{j,m}L_{jm}$.

\setcounter{figure}{0}
\setcounter{equation}{0}
\setcounter{table}{0}
\setcounter{theorem}{0}
\section{Additional empirical details}
\label{sec:extraempirical}
We provide additional results to supplement our empirical studies. 
\subsection{Predictive performance}\label{sec:predictiveperf}
For completeness, Figure \ref{fig:pred_perf} reports predictive results for Experiment 1: test RMSE (left) and 95\% predictive-interval coverage (right), averaged over 25 replications. We also include two ``black-box” baselines that do \emph{not} estimate coefficient surfaces (thus excluded from the main text):

\begin{itemize}
    \item \texttt{BART}: Bayesian additive regression trees fit to $([\bm{X},\bm{Z}])$ with default priors.
    \item \texttt{ERT}: Extremely randomized trees \citep{Guerts2006}: very large forests with randomized split thresholds; fit to $([\bm{X},\bm{Z}])$.
\end{itemize}
\begin{figure}[t]
    \centering
    \includegraphics[width=0.6\linewidth]{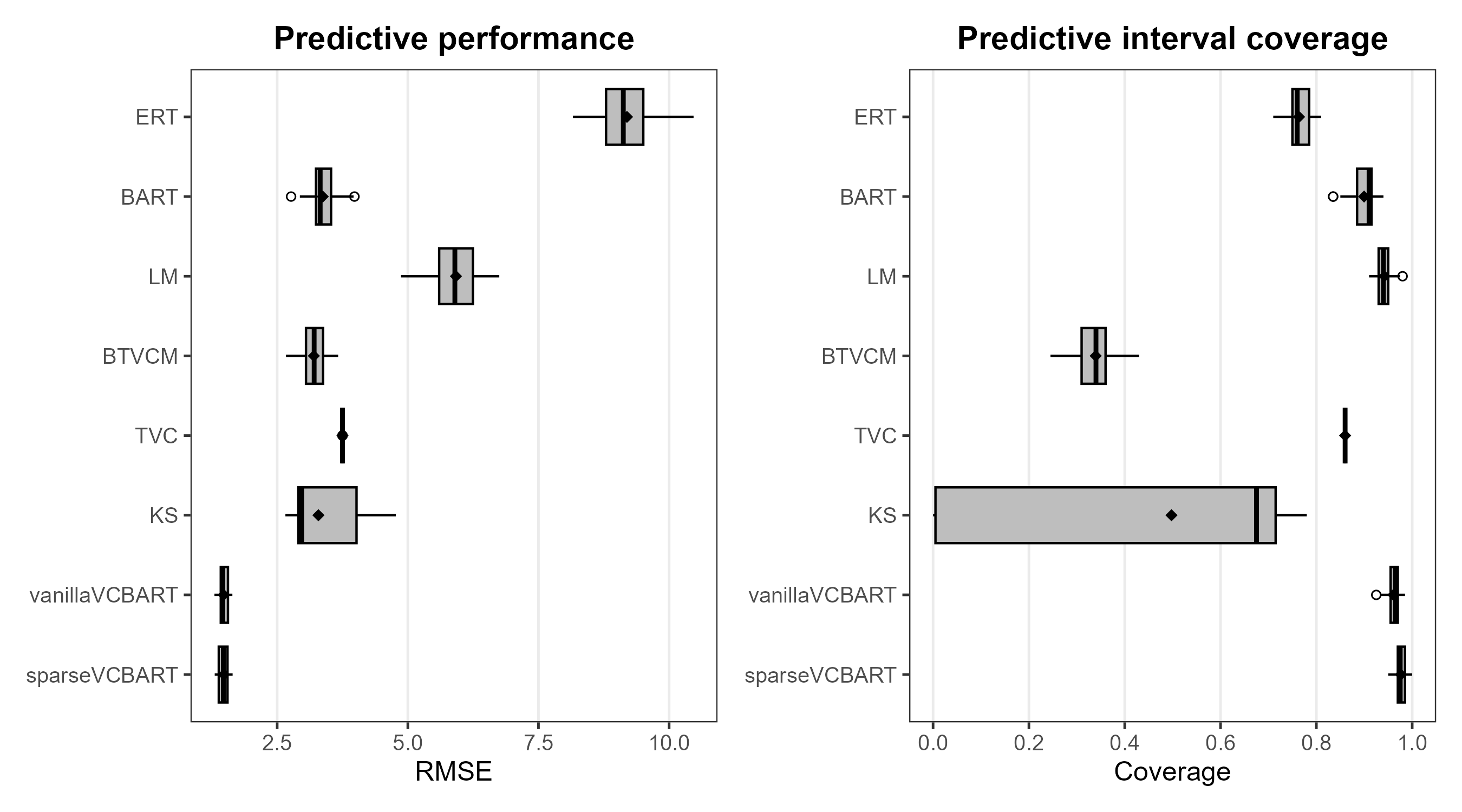}
    \caption{Predictive performance of all the methods averaged over $25$ replications}
    \label{fig:pred_perf}
\end{figure}

Both VCBART variants achieve the lowest predictive RMSE by a clear margin, reflecting the benefit of modeling $(Y)$ as $(\sum_j X_j \beta_j(\bm{Z}))$. \texttt{KS}, \texttt{TVC}, and \texttt{BTVCM} trail, and the homoscedastic \texttt{LM} is substantially worse—as expected when effects vary with $(\bm{Z})$. On coverage, the VCBART methods deliver near-nominal 95\% predictive coverage. \texttt{KS} exhibits unstable and widely varying coverage, and \texttt{BTVCM} under-covers. \texttt{BART} attains reasonable but slightly sub-nominal coverage with higher RMSE than both \texttt{vanillaVCBART} and \texttt{sparseVCBART}, while \texttt{ERT} shows the largest RMSE despite decent coverage—its intervals are wide but predictions are less accurate. Overall, these results corroborate the main-text claim: explicitly modeling varying coefficients improves both accuracy and uncertainty calibration in heterogeneous settings.

\subsection{Sensitivity on the ensemble size}\label{sec:sensitivity}
We replicate Experiment 1 with varying ensemble sizes $M_{\textrm{default}}=50$, $M_{\textrm{default}}/2=25$, and $2 \times M_{\textrm{default}}=100$, keeping the other hyperparameters fixed at their respective defaults to assess the sensitivity of both \texttt{sparseVCBART} and \texttt{vanillaVCBART} to this key hyperparameter. The results, summarized in Table \ref{tab:method-by-M}, show two clear trends. First, for both models, performance on all metrics (except predictive coverage, which remains high) degrades as $M$ decreases. As expected, a larger ensemble of trees provides a more flexible and stable function approximation, leading to better function recovery and predictive accuracy.

Second, while the models perform almost identically at $M=100$, \texttt{sparseVCBART} still shows a notable advantage in function recovery (lower MSE) at smaller ensemble sizes ($M=50$ and $M=25$). This suggests that the \emph{regularization} provided by the global-local priors helps compensate for having fewer trees in the ensemble. The shrinkage prior helps distinguish signal from noise, a task that \texttt{vanillaVCBART} accomplishes primarily by averaging over a larger number of weak learners. This makes \texttt{sparseVCBART} potentially more robust when using a smaller, more computationally efficient ensemble.

\begin{table}[H]
\centering
\caption{Average metrics across ensemble sizes $M \in \{100,50,25\}$. Best performers for each $M$ are bolded.}
\label{tab:method-by-M}
\begin{tabular}{
  l
  c
  S[table-format=1.3]
  S[table-format=1.3]
  S[table-format=1.3]
  S[table-format=1.3]
}
\toprule
\textbf{Method} &
\textbf{$M$} &
{\textbf{MSE $\big(\beta_j(\bm{z})\big)$}} &
{\textbf{Coverage $(\beta_j(\bm{z}))$}} &
{\textbf{Predictive RMSE}} &
{\textbf{Coverage PIs}} \\
\midrule
\texttt{sparseVCBART}  & 100 & \textbf{0.318} & 0.923 & \textbf{1.406} & \textbf{0.978} \\
\texttt{vanillaVCBART}  & 100  & 0.319 & \textbf{0.936} & 1.412 & 0.969 \\
\midrule
\texttt{sparseVCBART}  & 50  & \textbf{0.365} & 0.883 & 1.479 & \textbf{0.976} \\
\texttt{vanillaVCBART} & 50 & 0.411 & \textbf{0.888} & \textbf{1.477} & 0.960 \\
\midrule
\texttt{sparseVCBART} & 25  & \textbf{0.486} & 0.842 & \textbf{1.554} & \textbf{0.976} \\
\texttt{vanillaVCBART} & 25  & 0.506 & \textbf{0.854} & 1.557 & 0.969 \\
\bottomrule
\end{tabular}
\end{table}

\paragraph{Heuristic for choosing $M$.}
The optimal choice of $M$ is problem-dependent and represents a \emph{bias-variance tradeoff}. A larger $M$ allows the model to fit more complex functions with lower variance by averaging over many simple trees, which is the core idea of ensemble methods. However, this comes at a linear increase in computational cost. A mathematically grounded heuristic can be derived from the approximation theory. The theory suggests that a total leaf budget of $K_j \asymp N^{s_{0,j}/(2\alpha_j+s_{0,j})}$ is needed to approximate a function with intrinsic dimension $s_{0,j}$ and smoothness $\alpha_j$. Since the total expected number of leaves is roughly the number of trees ($M$) times the average number of leaves per tree ($\mathbb{E}[L_m]$), we have $M \times \mathbb{E}[L_m] \approx K_j$. This implies that for more complex functions (smaller $\alpha_j$, larger $s_{0,j}$), a larger total leaf budget is required, which can be achieved by increasing $M$.

In practice, since $\alpha_j$ and $s_{0,j}$ are unknown, we recommend starting with a default like $M=50$. This is often large enough for the ensemble effect to work well, as noted by \citet{Deshpande2024}. If we suspect the underlying functions are extremely complex or if performance is suboptimal, we might increase $M$ to $100$ or further, computational budget permitting. If computational cost is a major constraint, \texttt{sparseVCBART} likely allows for a smaller $M$ than \texttt{vanillaVCBART} would require for the same level of regularization, as its shrinkage prior provides an additional mechanism for controlling model complexity.

\subsection{Semi-synthetic data}
\label{sec:datadesc}
Exclusionary attitudes are a form of prejudice wherein individuals tend to favor their own group and oppose policies that promote well-being of individuals of the `outgroup' \citep{kalla_broockman23}. Recent work in political science \citep{broockman_kalla16, williamson_etal21} shows that interpersonal conversations can reduce exclusionary attitudes towards outgroups. \citet{kalla_broockman23} evaluate what types of interpersonal conversations can reduce prejudice. We focus on their door-to-door canvassing experiment conducted in 2019, conducted in Michigan, Pennsylvania and North Carolina. The intervention was designed to reduce prejudice against unauthorized immigrants.  

The intervention involves a canvasser knocking on doors and having conversations with 2838 voters who participated in the study. Voters were assigned to one of two treatment conditions or a placebo condition. The first treatment was a conversation that involved (what the authors call) a “full intervention” wherein the canvasser and voter would share stories about immigrants they knew, and then voters would be asked to discuss a time they needed support. The partial intervention was designed to assess whether only the first part of the intervention— here, the canvasser and voter would only share stories about immigrants they knew. A placebo arm, canvassers had short conversations about unrelated topics. 

Finally, outcomes were collected 1 week, 7 weeks and 4.5 months after the intervention. The primary outcomes are a policy index, that gauges support for including unauthorized immigrants in government programs, and a prejudice index. The study finds that both the treatment arms are similarly effective, and that the effects persist upto 4.5 months after the survey. Participants exposed to the treatment arms are between 3 and 10\% more likely to agree with policies that support unauthorized immigrants. 

Understanding heterogeneity in this context is important for a number of reasons. First, the increase in support for immigrant policies may be stronger when contact is made with an immigrant canvasser. Second, the treatment effects may vary by partisanship, due to the polarizing nature of the topic. Finally, personal characteristics of voters like age, gender, minority status in addition to initial levels of prejudice may be possible modifiers. Thus, the average effect of the treatment may be smaller or larger for certain subgroups. As a result, practitioners working on prejudice reduction at scale may not fully understand which subgroups respond best to an intervention, and may lead to resources being allocated inefficiently. 

The dataset to conduct the replication and additional analysis was downloaded from the Replication Archives of the paper, hosted on the Harvard Dataverse Network \footnote{https://dataverse.harvard.edu/citation?persistentId=doi:10.7910/DVN/MMYAOY}. On cleaning the data and excluding voters with missing data, we end up with 1991 unique voters, with about 600 assigned to each treatment arm and placebo. 

The experiment captures about a 100 covariates before the treatment is assigned. They can be categorized into groups as seen below: 

\begin{center}
\begin{tabular}{|l|l|l|}
\hline
\textbf{Category} & \textbf{Number of Variables} & \textbf{Information} \\
\hline
Ideology & 23 & Ordinal rating on 5 point scale \\
\hline
Policy & 23 & Ordinal rating on 5 point scale \\
\hline
Self Esteem & 9 & Ordinal rating on 5 point scale \\
\hline
Social Group Standing & 9 & Ranking between 1 \& 10 \\
\hline
Partisanship & 8 & Binary, Ordinal rating on five point scale \\
\hline
Race / Ethnicity & 7 & Binary, Categorical \\
\hline
Employment / Education & 7 & Binary, Categorical \\
\hline
Voting History & 4 & Binary \\
\hline
Age, Gender, Site & 5 & Binary, Continuous, Categorical \\
\hline
Views on Economy & 3 & Ordinal rating \\
\hline
Identifiers & 2 & Unique \\
\hline
\end{tabular}
\end{center}

\subsubsection{Local-scale separation}\label{sec:localscale}
Figure~\ref{fig:real-lambdas} plots the posterior medians of the predictor-specific local scales $\lambda_j$ from \texttt{sparseVCBART}. The two treatment predictors show substantially larger median $\lambda_j$ than all 18 noise predictors, producing a clear separation: irrelevant predictors are strongly shrunk, while the active treatments are left largely unshrunk. Figure \ref{fig:real-lambdas} provides a clear empirical validation of Theorem 4. The theorem formally guarantees that the local scales $\lambda_j$ corresponding to the 18 null predictors should concentrate near zero, which is precisely the behavior seen in the figure, distinguishing them from the large scales of the active treatment effects.

\begin{figure}[t]
  \centering
  \includegraphics[width=0.6\linewidth]{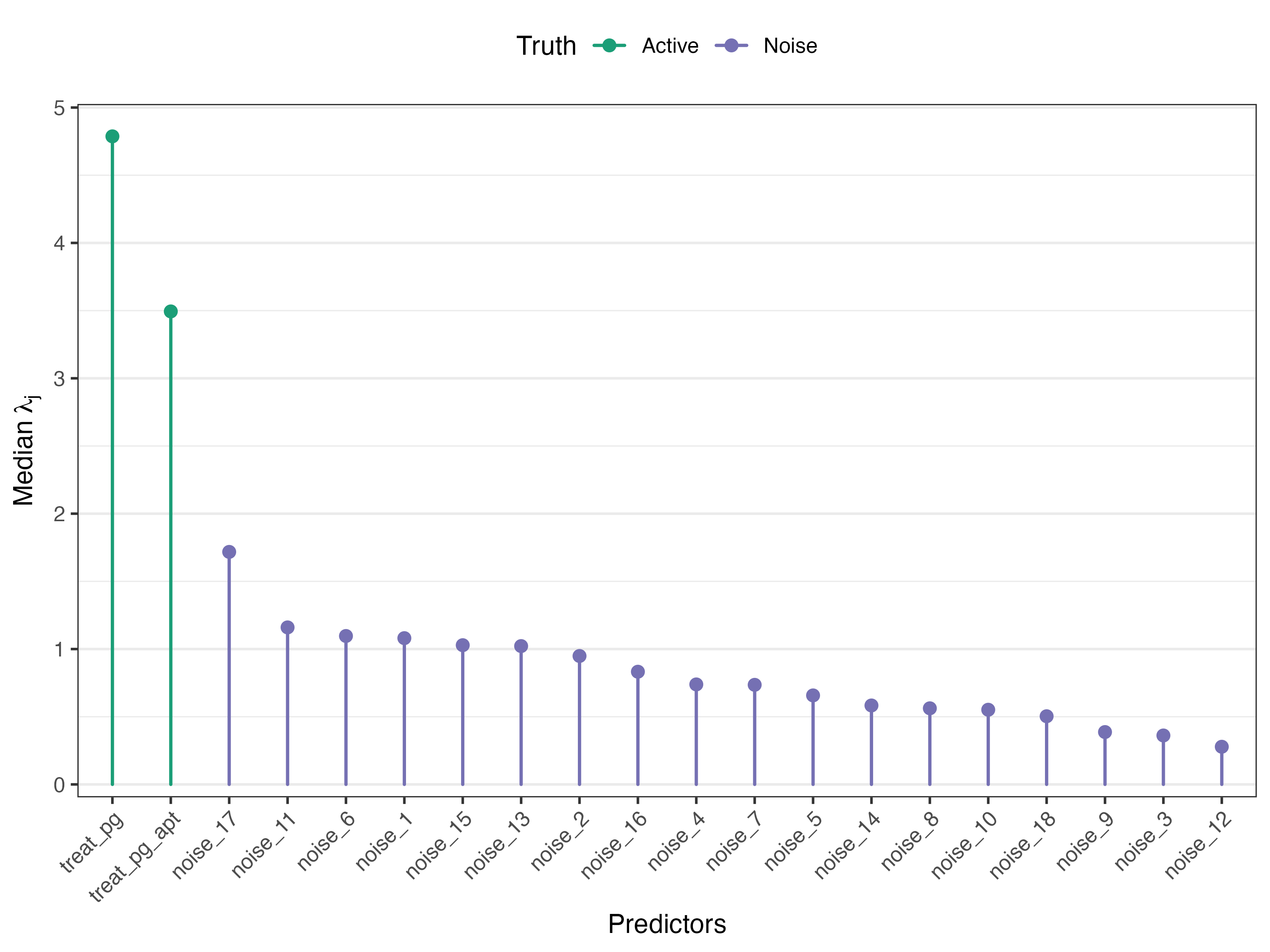}
  \caption{Posterior medians of local scales $\lambda_j$ from \texttt{sparseVCBART}. 
  Treatment predictors (\texttt{treat\_pg}, \texttt{treat\_pg\_apt}) show much larger medians than all 18 noise predictors, indicating strong separation between signal and noise.}
  \label{fig:real-lambdas}
\end{figure}

\subsubsection{Modifiers selected by \texttt{sparseVCBART}}\label{sec:selectedmod}
The model selects a small, interpretable set of modifiers for each treatment:
\begin{itemize}
  \item \textbf{Full intervention} (\texttt{treat\_pg\_apt}): 
  (i) support for health-care coverage to unauthorized immigrants under a public option (\texttt{t0\_immpolicy\_hc\_under\_uhc}); 
  (ii) support for access to food stamps (\texttt{t0\_immpolicy\_foodstamps}); 
  (iii) views on stricter gun laws (\texttt{t0\_gun\_stricter}); 
  (iv) baseline conservative ideology indicator (\texttt{t0\_ideo\_con}).
  \item \textbf{Partial intervention} (\texttt{treat\_pg}): 
  (i) assessment of the current economy (\texttt{t0\_economy\_today}); 
  (ii) support for K--12 access for undocumented children (\texttt{t0\_immpolicy\_k12schools}); 
  (iii) Social Dominance Orientation terciles (\texttt{t0\_sdo\_factor3}).
\end{itemize}

\subsubsection{Fit-the-fit for interpretability}\label{sec:fit-the-fit}
To explain \emph{how} effects vary with modifiers, we apply a ``fit-the-fit'' step \citep[e.g.,][]{fisher2024}: for each treatment, we take the posterior-mean surface $\hat\beta_j(\bz)$ from \texttt{sparseVCBART} and train a shallow CART \citep{BreimanCART} using only the modifiers selected for that treatment. The resulting trees compress a nonparametric surface into a handful of transparent splits. Figure~\ref{fig:cart-treatpg} shows the tree for \texttt{treat\_pg}: the root split is on support for K--12 access (\texttt{t0\_immpolicy\_k12schools}), with subsequent refinements involving the current-economy rating (\texttt{t0\_economy\_today}) and SDO. Subgroup comparisons between sister leaves indicate positive effects in both branches, with larger average effects among respondents more supportive of K--12 access and more optimistic about the economy, while higher SDO tends to dampen the treatment effect. These patterns provide policy-relevant, easily communicated rules for heterogeneity.

\begin{figure}[t]
  \centering
  \includegraphics[width=\linewidth]{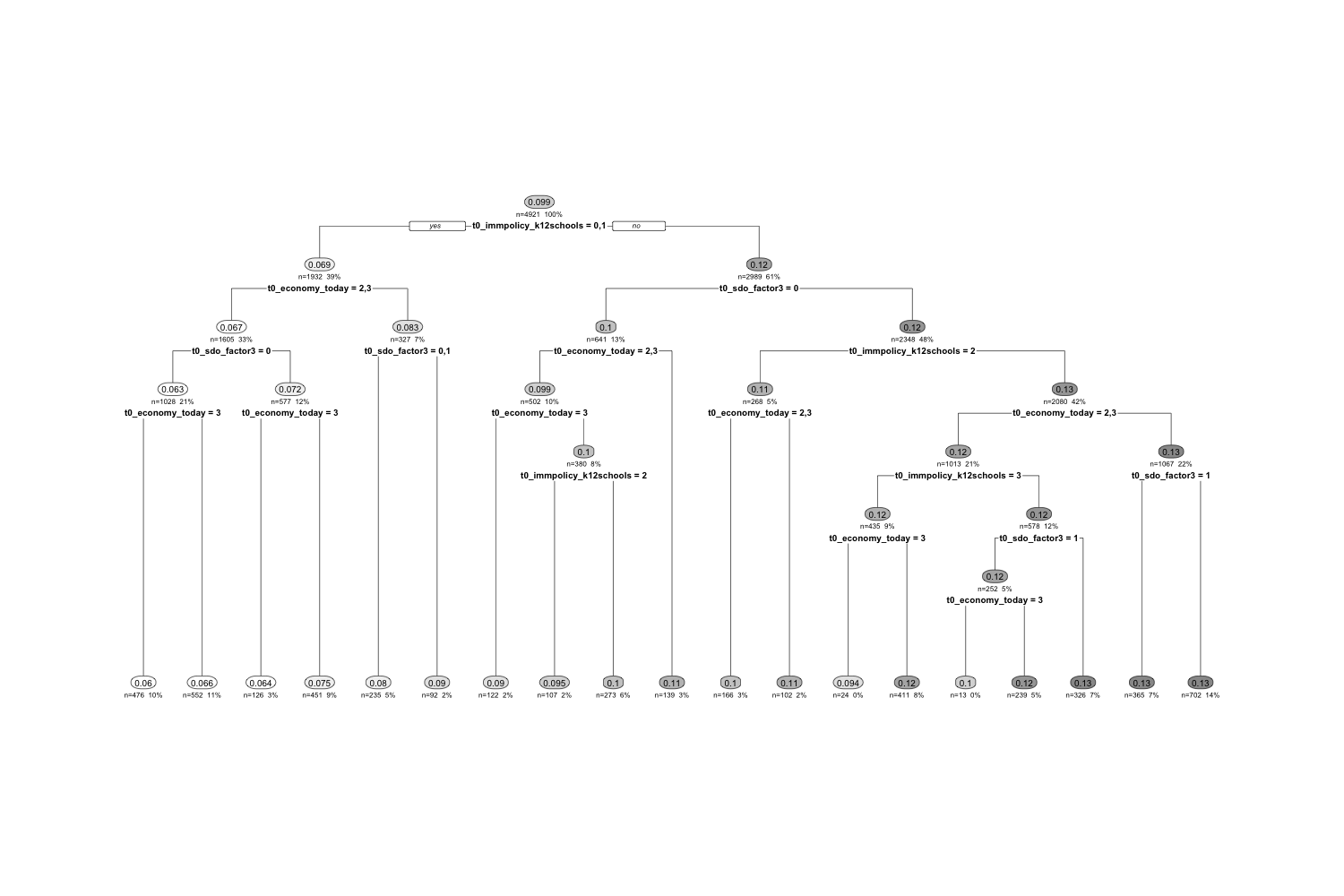}
  \caption{Fit-the-fit CART for the posterior-mean surface $\hat\beta_{\texttt{treat\_pg}}(\bz)$. 
  The root split is on support for K--12 access for undocumented children; downstream splits involve the current-economy rating and SDO terciles. 
  Leaves summarize subgroup-average effects, yielding transparent, policy-relevant heterogeneity.}
  \label{fig:cart-treatpg}
\end{figure}

\end{document}